\documentclass{article}

\usepackage{iclr2026_conference,times}

\usepackage{amsmath,amssymb,amsthm}
\usepackage{algorithm}
\usepackage{algorithmic}

\usepackage{graphicx}
\usepackage{booktabs}
\usepackage{multirow}
\usepackage{float}
\usepackage{tikz}
\usetikzlibrary{positioning}

\usepackage{microtype}
\usepackage{url}
\usepackage{hyperref}
\hypersetup{
  pdftitle={Causal Judge Evaluation: Calibrated Surrogate Metrics for LLM Systems},
  pdfauthor={Eddie Landesberg, Manjari Narayan}
}
\usepackage[capitalize,noabbrev]{cleveref}
\usepackage{enumitem}

\newcommand{\E}{\mathbb{E}}
\newcommand{\Var}{\operatorname{Var}}

\newcommand{\CV}{\operatorname{CV}}
\newcommand{\SE}{\operatorname{SE}}
\newcommand{\ESS}{\mathrm{ESS}}

\newcommand{\ind}[1]{\mathbf{1}_{#1}}

\newcommand{\argmin}{\mathop{\mathrm{arg\,min}}\limits}

\newcommand{\R}{\mathbb{R}}

\newcommand{\pzero}{\pi_0}
\newcommand{\pprime}{\pi'}

\newcommand{\IPS}{\mathrm{IPS}}
\newcommand{\DR}{\mathrm{DR}}

\newcommand{\CI}{\mathrm{CI}}

\newtheorem{theorem}{Theorem}
\newtheorem{lemma}{Lemma}
\newtheorem{prop}{Proposition}
\newtheorem{proposition}[prop]{Proposition} 
\newtheorem{corollary}{Corollary}
\theoremstyle{remark}

\newcommand{\CJE}{\textsc{CJE}}
\newcommand{\RefuseLevel}{\textsc{Refuse-Level}}

\graphicspath{{figures/}{figs/}}

\title{Causal Judge Evaluation: Calibrated \mbox{Surrogate} Metrics for LLM Systems}

\iclrfinalcopy            

\author{Eddie Landesberg\thanks{\texttt{edward.landesberg@gmail.com}} \quad Manjari Narayan}

\begin{document}
\maketitle

\begin{abstract}
Measuring long-run LLM outcomes (user satisfaction, expert judgment, downstream KPIs) is expensive. Teams default to cheap LLM judges, but uncalibrated proxies can invert rankings entirely. \textbf{Causal Judge Evaluation (CJE)} makes it affordable to aim at the right target: calibrate cheap scores against a small oracle slice, then evaluate at scale with valid uncertainty.
We treat surrogate validity as \emph{auditable}: for each policy or deployment context, a small oracle sample tests whether calibration remains unbiased, turning an uncheckable identification condition into a falsifiable diagnostic.
On 4,961 Arena prompts with 5 policies and a 16$\times$ oracle/judge cost ratio: at 5\% oracle fraction, CJE achieves 99\% pairwise ranking accuracy at 14$\times$ lower cost; across all configurations (5--50\% oracle, varying $n$), accuracy averages 94\%. An adversarial policy fails the transport audit and is correctly flagged. CJE refuses level claims rather than reporting biased estimates.
Key findings: naive confidence intervals on raw judge scores achieve 0\% coverage (CJE: ${\sim}$95\%); importance-weighted estimators fail despite 90\%+ effective sample size; the Coverage-Limited Efficiency (CLE) bound and its TTC diagnostic explain why.
\end{abstract}

\section{Introduction}
\label{sec:intro}

Many teams now choose between system configurations (models, prompts, parameters, tools) by generating outputs and scoring them with an LLM judge. This is fast, and often wrong in a systematic way. Consider a common failure: you compared configuration A to B, the judge said A wins, you shipped A, and users hated it. The mechanism: verbosity and style bias. Judges frequently reward longer, more elaborate responses independent of actual quality; optimizing against the judge selected the wrong configuration. The standard workaround (reserve oracle labels for a small validation slice and hope the correlation holds) is ubiquitous in practice but statistically unsound.

CJE is a simple protocol: \textbf{calibrate} the judge to a small set of oracle labels, then \textbf{audit} whether that calibration remains valid for each target policy. When the audit fails, we do not ``hope correlation holds''; we either recalibrate or refuse level claims.

\noindent\textbf{Quick start for practitioners:}
\begin{itemize}[leftmargin=*,nosep]
\item Use Direct + two-stage calibration. For open-ended generation under nontrivial policy shifts, logs-only OPE is impractical: target-typical trajectories receive little logger mass (low TTC), yielding a coverage-limited precision floor that weight stabilization cannot overcome.
\item Audit transport for high-stakes decisions (mean residual test per policy).
\item Gate on TTC $\geq 0.70$ before relying on logged data; in our benchmark, only near-clone policies pass this threshold.
\end{itemize}

\paragraph{Three systematic failures.}
Uncalibrated ``LLM-as-judge'' evaluation exhibits three failure modes that compound in production:

\begin{enumerate}[leftmargin=*,itemsep=4pt,topsep=4pt]
\item \textbf{Rankings can invert.} Judge scores $S$ are treated as if they were rewards $Y$ on the same scale. They are not. Without calibration, higher $S$ can predict \emph{lower} $Y$, inverting rankings entirely. This is the ``surrogate paradox'' in a new guise: surrogates can mislead catastrophically when the surrogate-target relationship shifts across contexts \citep{FlemingDeMets1996}. The ``You're absolutely right!'' phenomenon illustrates the cost: sycophantic responses score high on cheap proxies but low on actual user value. Teams know they should evaluate against long-term outcomes, but oracle labels are expensive, so they settle for proxies that may point the wrong direction.

\item \textbf{Standard confidence intervals fail catastrophically.} Naive CIs on uncalibrated judge scores achieve near-zero coverage: their 95\% CIs capture the true value in 0/50 seeds. Bootstrap inference with bias-corrected estimation restores coverage to ${\sim}$95\% for Direct mode. (\Cref{fig:oua-decomposition}: at 5\% oracle fraction, calibration uncertainty contributes ${\sim}$90\% of total variance.)

\item \textbf{Importance sampling fails despite healthy ESS.} When evaluating target policies $\pi'$ using logged data from $\pi_0$, importance-weighted estimators (inverse propensity scoring, IPS; self-normalized IPS, SNIPS) fail even after weight stabilization boosts effective sample size (ESS) from degenerate ($<1\%$) to healthy ($>80\%$) regimes. The problem is not weight concentration but \emph{coverage}: the logger rarely visits regions where the target policy concentrates. We formalize this via the \textbf{Coverage-Limited Efficiency (CLE)} diagnostic, which explains why IPS-style estimation hits a precision floor under limited overlap. (\Cref{tab:weight-diagnostics}: TTC ranges 19--49\% with CLE factors of 24--61$\times$.)
\end{enumerate}

\noindent\emph{Why logs-only OPE fails for open-ended generation.}
Off-policy evaluation requires the logger to have visited regions where the target concentrates. In autoregressive text generation, the ``action'' is an entire response---a trajectory through a space of $|\mathcal{V}|^T$ possible token sequences. Two issues compound: (i) teacher-forcing propensities are empirically unreliable (even our clone policy achieves only 26\% raw ESS vs.\ theoretical 100\%), and (ii) for meaningfully different policies, TTC is 19--49\%, indicating the logger rarely visits target-typical regions. Even with perfect propensities, CLE implies a hard precision floor when TTC is low. In constrained settings (short outputs, near-clone shifts, tool-choice bandits), OPE may remain viable; for open-ended generation under nontrivial policy shifts, Direct sidesteps these issues entirely.

\paragraph{Related work.}
Recent work recognizes that LLM judges require calibration \citep{Lee2025LLMJudgeReporting}, and concurrent work \citep{Chen2026EfficientLLMJudge} develops semiparametrically efficient, calibration-aware inference for mean parameters in the standard ``large judge-labeled test + small human-labeled calibration'' design. \citet{Guerdan2025DoublyRobustLLM} use doubly robust methods to address external validity when LLM judges are prompted with demographic personas. Our focus is complementary: we study multi-policy ranking under policy/deployment shift, introduce auditable transport diagnostics (policy-wise mean residual tests), and provide overlap/coverage diagnostics (CLE) and weight stabilization for off-policy evaluation, components not addressed by mean-parameter-only treatments. CJE fills this gap.

\paragraph{Auditable assumptions.}
The core insight of CJE is to replace an uncheckable surrogacy assumption with a \emph{policy-wise moment restriction} that we can audit. In causal inference, negative controls and proxy variables provide a template for this move: they turn untestable identification conditions into checkable diagnostics using auxiliary data \citep{Lipsitch2010,TchetgenTchetgen2024}. CJE applies this philosophy to LLM evaluation by budgeting a small oracle slice to test whether the judge-to-oracle calibration remains valid in each target context. The protocol is: calibrate on $\pi_0$, then \emph{audit transport} on each $\pi'$ (or time window, user segment, or task population) with a small oracle sample. Pass $\to$ reuse calibration; fail $\to$ recalibrate or refuse level claims for that cell. This transforms a brittle identification assumption into a standardized audit with an explicit decision rule.

\paragraph{Causal Judge Evaluation (\CJE).}
We introduce CJE, a framework that addresses all three failures.\footnote{Library: \texttt{pip install cje-eval} (\url{https://github.com/cimo-labs/cje}); experiments: \url{https://github.com/cimo-labs/cje-arena-experiments}} CJE casts offline evaluation as calibrated surrogate estimation with three core components:

\begin{itemize}[leftmargin=*,itemsep=2pt,topsep=2pt]
\item Reward calibration: mean-preserving isotonic regression from judge score $S$ to oracle labels $Y$, with a two-stage option for covariate-dependent bias (e.g., response length).
\item Weight stabilization: variance-optimal stacking of monotone weight candidates, stabilizing importance weights for improved effective sample size.
\item Calibration-aware inference: bootstrap with bias correction that propagates calibration uncertainty into confidence intervals, achieving near-nominal coverage.
\end{itemize}

\noindent Conceptually, CJE adopts the \emph{surrogate index} move \citep{AtheyChettyImbensKang2019}, learning $f(S,X) \approx \E[Y \mid S,X]$ from cheap proxies to expensive outcomes, and applies it to LLM evaluation with auditable transport diagnostics. These components instantiate a single design rule: encode justified knowledge as model restrictions (closed convex sets) and project onto them. Building on established semiparametric theory, we show that such restrictions can only reduce variance (\cref{thm:model-restriction}); with cross-fitting, our estimators attain the surrogate information bound. (Component names and technical details in \cref{sec:method}.)

Beyond variance control, CJE includes a mean transport test that operationalizes the auditable assumption protocol: for each target policy $\pi'$ or time period, we test whether the mean residual $\E[Y - f(S,X)]$ is zero. In our Arena benchmark, the base-trained calibration transports cleanly to clone, premium, and prompt-engineered policies, but fails for an adversarial ``unhelpful'' policy, where mean residuals reveal a $-0.31$ level shift. When the test fails, the failure is \emph{visible and actionable}: either recalibrate with policy-specific data or refuse absolute value claims for that cell.

\paragraph{The Arena benchmark.}
We validate CJE on a large-scale benchmark: 4,961 Chatbot Arena prompts, five LLM policies (including an adversarial ``unhelpful'' policy), 13 estimators, and \texttt{gpt-5-2025-08-07} as oracle. Key findings:

\begin{itemize}[leftmargin=*,itemsep=2pt,topsep=2pt]
\item Calibration works. The Direct Method with two-stage calibration achieves 94\% pairwise ranking accuracy on average (vs.\ 38\% for SNIPS). Naive CIs on uncalibrated scores achieve 0\% coverage; bootstrap inference with $\hat\theta_{\mathrm{aug}}$ improves coverage to ${\sim}$95\% (Direct and stacked doubly-robust).
\item OPE fails unexpectedly. We expected weight stabilization to enable OPE; it did not. Despite weight stabilization boosting ESS from $<1\%$ to $>80\%$, calibrated IPS remains near-random (47\% pairwise). CLE explains why: high ESS is necessary but not sufficient when the logger rarely visits target-typical regions.
\item Doubly-robust (DR) does not dominate. We expected DR to outperform Direct by combining logged data with fresh draws. Instead, Direct slightly outperforms DR (94.4\% vs.\ 94.1\%): under low overlap, DR's IPS component adds noise rather than information.
\end{itemize}

\paragraph{Contributions.}
\begin{enumerate}[leftmargin=*,nosep]
\item Framework: CJE unifies calibration, weight stabilization, and uncertainty quantification for surrogate-based LLM evaluation.
\item Methods: Mean-preserving reward calibration, weight stabilization for improved ESS, and calibration-aware bootstrap inference.
\item Theory: CJE adapts the surrogate index framework to LLM evaluation with a key innovation: transportability becomes auditable rather than assumed via the mean transport test. Inference is grounded in the efficient influence function for missing-outcome estimation, explaining why bootstrap with $\hat\theta_{\mathrm{aug}}$ achieves ${\sim}$95\% coverage where naive methods achieve 0\%. The CLE bound provides a novel explanation for OPE failure despite high ESS, formalizing a precision floor under limited overlap.
\item Validation: A 13-estimator benchmark on ${\sim}5$k Arena prompts: at 5\% oracle fraction, 99\% pairwise ranking accuracy; across all configurations, 94\% average. Cost reduction of 14$\times$ (for 5 policies) with near-nominal CI coverage (${\sim}$95\% for both Direct and stacked DR).
\item Operations: A closed-form Square Root Law for optimal budget allocation (\cref{app:budget}), enabling practitioners to minimize variance subject to cost constraints by balancing evaluation and calibration uncertainty.
\end{enumerate}

\noindent\textit{Readers primarily interested in deployment can focus on \cref{sec:method} (pipeline overview), \cref{ssec:takeaways} (seven practical rules), and \cref{tab:estimator-guide} (estimator selection guide).}

\section{Background and Setup}
\label{sec:background}

\paragraph{Setup \& notation.}
We observe i.i.d.\ logs $(X_i,A_i,S_i)$ under a fixed logger $\pi_{0}(\cdot\mid X)$; $S=s(X,A)$ is a
scalar judge score on every row, and a small i.i.d.\ oracle slice provides labels $Y$. For a candidate
policy $\pi'$, the sequence-level importance ratio is
\[
W_{\pi',i} \;=\; \frac{\pi'(A_i\mid X_i)}{\pi_{0}(A_i\mid X_i)}
\;=\; \exp\!\big\{\log p_{\pi'}(A_i\mid X_i)-\log p_{\pi_{0}}(A_i\mid X_i)\big\},
\]
computed via teacher forcing (TF). The target is the counterfactual value $V(\pi')=\E[Y(\pi')]$.
We use self-normalized importance sampling (SNIPS), which produces sample-mean-one weights, when helpful.

\paragraph{The LLM evaluation regime.}
Standard off-policy evaluation (OPE) assumes the target policy \emph{cannot} be executed. LLM evaluation operates in a different regime: generating text is cheap, but grading it with oracle labels (human experts, downstream KPIs) is expensive. This creates three evaluation modes: (i) logs-only (IPS/SNIPS), when we cannot re-run the model; (ii) fresh draws (Direct), when we generate new responses but score them with a calibrated surrogate; and (iii) hybrid (DR), combining both. All three share the core challenge: learning $S \to Y$ calibration from a small oracle slice collected under $\pi_0$, then applying it to evaluate $\pi'$. The calibration is always \emph{off-policy}; the actions may or may not be.

\paragraph{Two operating modes: ranking vs.\ levels.}
CJE supports two use cases with different requirements:
\begin{itemize}[leftmargin=*,nosep]
\item Ranking: Determine which policy is better ($V(\pi_1) > V(\pi_2)$?). Requires calibration to preserve order across policies. If calibration bias differs by policy (e.g., overestimates sycophantic responses), rankings can invert even without level claims.
\item Absolute levels: Report actual policy values ($V(\pi') = 0.72$) or make deployment decisions based on thresholds. Requires mean-unbiasedness: $\E_{\pi'}[Y - f(S,X)] = 0$.
\end{itemize}
\noindent The transport audit catches both failure modes: policy-specific bias that inverts rankings and level shifts that bias absolute values. For low-stakes exploratory analysis, the audit may be deferred; for high-stakes decisions (deployment, cross-time comparisons, or when policy differences are small), run the audit.

\paragraph{CJE as surrogate index for LLM evaluation.}
Conceptually, CJE adopts the \emph{surrogate index} move \citep{AtheyChettyImbensKang2019}: learn a mapping $f(S,X) \approx \E[Y \mid S, X]$ from cheap proxies to expensive outcomes, then use it for downstream estimation. For policy value estimation, we require only mean transport: $\E_{\pi'}[Y - f(S,X)] = 0$, which we treat as auditable via a small oracle slice in each target context (\cref{prop:mean-transport}). We additionally impose monotonicity and mean-preservation as model restrictions. These are not identification assumptions, but deliberate regularization choices that stabilize calibration and prevent preference inversions.

Athey et al.'s framework addresses cross-sample comparability when surrogates are learned in one environment and applied in another. CJE faces two analogous cross-environment challenges: (i) calibration transport (does $f$ learned under $\pi_0$ remain mean-unbiased under $\pi'$?), addressed by the mean residual test; and (ii) action-space coverage (does the logger visit regions where $\pi'$ concentrates?), addressed by CLE/TTC diagnostics. The key difference: we \emph{budget oracle labels in each target context} to audit transport rather than assume it, making CJE an evaluation protocol with explicit failure detection rather than a brittle identification claim.

\paragraph{OPE basics.}
IPS/SNIPS estimate $V(\pi')$ by reweighting logged outcomes
\citep{HorvitzThompson1952,Hajek1964Rejective,LiChuLangfordSchapire2011,SwaminathanJoachims2015CRM}.
The direct method (DM) plugs in $g(x)=\sum_{a}\pi'(a\mid x)\,\hat m(x,a)$. Doubly robust (DR)
estimators combine IPS and DM \citep{BangRobins2005} and, with sample–splitting and cross–fitting, admit $\sqrt{n}$
inference under the standard one–of–two $n^{-1/4}$ product–rate condition
\citep{Bickel1993,VaartWellner2000,Kosorok2008,JiangLi2016,Chernozhukov2018,vanDerLaanRose2011TL, KallusUehara2020DRL}.
Teacher forcing (TF) provides sequence–level propensities/ratios, so these forms apply to sequence policies in principle; however, TF reliability depends on API implementation details (see Limitations).

\paragraph{Variance, overlap, and stabilization.}
IPS variance scales with $\E[W_{\pi'}^{2}]$ and deteriorates under limited overlap
\citep{Crump2009LimitedOverlap}. We monitor stability with the effective sample size (ESS),
\[
\mathrm{ESS}(W)=\frac{\big(\sum_i W_i\big)^2}{\sum_i W_i^2},
\qquad
\frac{\mathrm{ESS}(W)}{n}=\frac{1}{1+\mathrm{CV}^2(W)}\ \ \text{when } \bar W=1\ \text{(global mean–one/SNIPS)}.
\]
We also track tail behavior via diagnostics \citep{Hill1975TailIndex,Liu2001MC,Owen2013MCTME}.
Common stabilizers include truncation/clipping \citep{Ionides2008TruncatedIS}, overlap weighting
\citep{LiMorganZaslavsky2018}, balancing objectives
\citep{Kallus2018Balanced}, and covariate–shift reweighting
\citep{Shimodaira2000CovShift,SugiyamaKrauledatMueller2007IWCV}.

\paragraph{The Coverage-Limited Efficiency (CLE) problem.}
High ESS indicates that weights are not dominated by a few extreme observations, but it does not guarantee that the logger has meaningful \emph{coverage} in regions where the target policy concentrates.
Let $\mathcal{T} \subset \mathcal{X} \times \mathcal{A}$ be a target-relevant region with $\alpha = P_{\pi'}(\mathcal{T})$ (target mass) and $\beta = P_{\pi_0}(\mathcal{T})$ (logger mass). For any IPS-style estimator based on importance-weighted outcomes,
\begin{equation}
\label{eq:cle-bound}
\mathrm{SE}(\hat\Psi_{\mathrm{IPS}}) \;\ge\; \frac{\sigma_{\mathcal{T}} \,\alpha}{\sqrt{\beta\, n}} \sqrt{1 + \chi^2\!\big(\pi'_{\mathcal{T}} \,\|\, \pi_{0,\mathcal{T}}\big)},
\end{equation}
where $\sigma_{\mathcal{T}}^2 := \mathrm{ess\,inf}_{(x,a)\in \mathcal{T}} \mathrm{Var}(Y \mid X{=}x, A{=}a)$ is the minimal outcome noise in $\mathcal{T}$, and $\chi^2(\pi'_{\mathcal{T}} \| \pi_{0,\mathcal{T}})$ measures shape mismatch inside $\mathcal{T}$ (proof in Appendix~\ref{app:proof-cle}).
The bound has three multiplicative factors: (i) the coverage penalty $\alpha/\sqrt{\beta}$, which explodes when the logger rarely visits target-typical regions; (ii) the shape mismatch $\sqrt{1+\chi^2}$, inflating the floor even with good coverage; and (iii) the standard noise term $\sigma_{\mathcal{T}}/\sqrt{n}$.

\emph{Implication.} When $\beta$ is small (poor logger coverage), the CLE floor is prohibitive regardless of weight stabilization. This explains our empirical finding that calibrated IPS fails despite ESS${}>90\%$ (\cref{sec:experiments}): high ESS indicates no weight concentration, but low $\beta$ (poor coverage in target-typical regions) sets a hard precision floor that logs-only methods cannot beat. This is a \emph{structural} limitation: even with perfect propensities, $\beta$ collapses exponentially with sequence length (\cref{prop:ttc-collapse}).

\paragraph{Diagnostics for CLE.}
The CLE bound decomposes into two factors, each with a corresponding diagnostic:
\begin{itemize}[nosep,leftmargin=*]
\item Coverage penalty ($\alpha/\sqrt{\beta}$): Measured by TTC (Target-Typicality Coverage) $= \hat\beta$, the estimated logger mass on $\mathcal{T}$.
We define the target-typical region $\mathcal{T}$ via a surprisal threshold: let $\ell_{\pi'}(x,a) = -\log p_{\pi'}(a \mid x)$ be the total negative log-likelihood under the target policy (from teacher forcing), and let $\bar\ell_{\pi'}(x,a) = \ell_{\pi'}(x,a)/|a|$ be the mean per-token surprisal. Set $\mathcal{T} = \{(x,a) : \bar\ell_{\pi'}(x,a) \le q_{0.9}\}$, where $q_{0.9}$ is the 90th percentile of $\bar\ell_{\pi'}$ under $\pi'$.
Then $\alpha = P_{\pi'}(\mathcal{T}) = 0.9$ by construction, and TTC $= \hat\beta = \hat P_{\pi_0}(\mathcal{T})$ is the empirical fraction of logged samples falling in $\mathcal{T}$.
If TTC${}<0.7$, logs-only IPS will \emph{fail}, meaning the CLE lower bound implies standard errors too large to distinguish realistic policy differences, even under ideal calibration and weight stabilization; prefer Direct or DR methods. (Threshold calibrated for MDE ${\approx}0.02$ at $n{\approx}5\text{k}$; see App.~\ref{app:diag-gates}.)
\emph{Applicability:} TTC requires sampling from or computing surprisal under $\pi'$ (as we do via teacher forcing). When $\pi'$-access is unavailable, fall back to conservative rules (e.g., always prefer Direct over IPS) or use proxy distributions to approximate TTC.
\item Shape mismatch ($\sqrt{1+\chi^2}$): Measured by Bhattacharyya affinity $A_B = \int \sqrt{p_{S|\pi'}(s)\,p_{S|\pi_0}(s)}\,ds$ in the surrogate space. Low affinity ($A_B < 0.5$) indicates severe shape mismatch; the gate threshold ($A_B \geq 0.85$; Table~\ref{tab:gates}) ensures adequate overlap for reliable IPS.
\end{itemize}
Both diagnostics are complementary: TTC checks if the logger visits target-typical regions (action space); Bhattacharyya affinity checks overall alignment of judge score distributions across policies (surrogate space).

\paragraph{Calibration for OPE.}
Calibration enforces identities under $\pi_{0}$: outcome calibration de-biases $g(X)$; ratio calibration
enforces $\E_{\pi_{0}}[W_{\pi'}]=1$ and $\E_{\pi_{0}}[W_{\pi'}h]=\E_{\pi'}[h]$ for a test class $h$;
orthogonal moments enable honest inference \citep{Chernozhukov2018}. Recent work gives
projection-based IPS/DR with stability guarantees
\citep{vanDerLaanLinCaroneLuedtke2025StabIPW,vanDerLaanLuedtkeCarone2024DRCalibration}. For DR, IF
orthogonality renders small calibration error second order
\citep{Bickel1993,Chernozhukov2018,vanDerLaanRose2011TL}.

\paragraph{Shape constraints (isotonic).}
Isotonic regression is the Euclidean projection onto the cone of monotone functions (PAVA)
\citep{Ayer1955PAVA,Barlow1972}; it avoids extrapolation and \emph{weakly reduces dispersion} by
majorization \citep{HardyLittlewoodPolya1952,MarshallOlkinArnold2011}. CJE exploits this via mean-preserving isotonic projections for both reward calibration and weight stabilization (\cref{sec:method}).

\paragraph{Judges as surrogates.}
Automatic judges (LLM-as-judge or preference models) provide scalable scoring
\citep{Ouyang2022Instruct,BaiEtAl2022,Zheng2023LLMasJudge,Kim2024PrometheusJudge,KocmiFedermann2023}
but are correlational and may drift \citep{DietzEtAl2025LLMJudges}. Viewing $S$
as a \emph{surrogate} connects to surrogate endpoints and mediation
\citep{Prentice1989Surrogate,RobinsGreenland1992,FrangakisRubin2002PS,Pearl2012Mediation,VanderWeele2015Book}.
The \emph{surrogate paradox}, wherein a policy improves $S$ and $S$ correlates with $Y$ yet the policy harms $Y$, can arise even with high correlation \citep{VanderWeele2013Surrogates}; our transport test (\cref{prop:mean-transport}) provides a necessary check.
Calibrating $R=f(S)$ enables unbiased estimation of $V(\pi')$ when the \emph{transport condition} holds: $\E_{\pi'}[Y - f(S,X)] = 0$ (\cref{prop:mean-transport}). \emph{Mean sufficiency} ($\E[Y\mid X,A,S]=\mu(S)$) is one sufficient condition guaranteeing transport; in CJE, level claims are gated by a per-policy transport audit. Without that audit, transport must be assumed. The surrogate also supplies a one-dimensional index that stabilizes weights.

\begin{figure}[t]
  \centering
  \begin{tikzpicture}[
    node distance=1.4cm and 1.8cm,
    every node/.style={font=\small},
    obs/.style={circle, draw, minimum size=0.9cm, inner sep=2pt},
    arr/.style={->, >=stealth, thick},
    fail/.style={->, >=stealth, thick, dashed, red!70!black}
  ]
    \node[obs] (X) {$X$};
    \node[obs, right=1.8cm of X] (A) {$A$};

    \node[obs, right=1.8cm of A, yshift=0.6cm] (S) {$S$};
    \node[obs, right=1.8cm of A, yshift=-0.6cm] (Y) {$Y$};

    \node[obs, below=0.8cm of A] (pi) {$\pi$};

    \node[obs, right=1cm of Y] (E) {$E$};

    \draw[arr] (pi) -- (A);
    \draw[arr] (X) -- (A);
    \draw[arr] (X) to[bend left=15] (S);
    \draw[arr] (X) to[bend right=15] (Y);
    \draw[arr] (A) -- (S);
    \draw[arr] (A) -- (Y);

    \draw[fail] (E) -- (Y);
    \draw[fail] (pi) to[bend right=25] (Y);

  \end{tikzpicture}
  \caption{\textbf{Judge-as-surrogate under policy and environment shift.}
  $X$: prompt; $A$: response; $S$: judge score; $Y$: oracle label; $\pi$: policy; $E$: environment.
  $S$ and $Y$ are parallel measurements of the response.
  Solid arrows: causal structure ($\pi \to A$: policy determines response).
  Dashed red arrows: \emph{mechanism-shift edges} indicating the $S$-$Y$ calibration may differ across policies or environments; these are not causal effects of $\pi$ or $E$ on $Y$ holding $(X,A)$ fixed.
  \textbf{Plain English:} Different policies generate different response styles, and the judge may be biased toward certain styles (e.g., verbosity). CJE's transport test catches this via $\E[Y - f(S,X)] = 0$.}
  \label{fig:dag}
\end{figure}
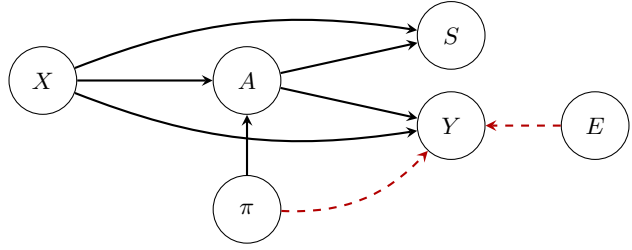

\paragraph{Surrogate interpretation.}
We use the judge score $S$ as a \emph{surrogate measurement} for the expensive outcome $Y$, not as a causal mediator.
Given a calibration function $f$ learned on an oracle slice, define the residual $\varepsilon := Y - f(S,X)$.
For any target policy $\pi'$, the value estimation error decomposes exactly as
\begin{equation}
\label{eq:transport-decomp}
\E_{\pi'}[Y] - \E_{\pi'}[f(S,X)] \;=\; \E_{\pi'}[\varepsilon].
\end{equation}
Thus, surrogate-only evaluation is mean-unbiased for policy $\pi'$ if and only if the mean residual transports: $\E_{\pi'}[\varepsilon] = 0$.
We treat this as an \emph{auditable} assumption rather than an identification axiom: we estimate and test $\E_{\pi'}[\varepsilon]$ on a small oracle audit slice per policy (or time window), and re-anchor calibration (or refuse level estimates) when the test fails (\cref{prop:mean-transport}).

This mean-transport condition is strictly weaker than Prentice-style surrogacy, which requires $Y \perp \pi' \mid (S, X)$; it suffices for unbiased \emph{mean value} estimation but does not guarantee conditional calibration within subgroups or at distribution tails.
When finer guarantees are needed, the audit can be strengthened from a single mean to decile-binned residual checks, detecting compensating errors that cancel in the overall mean.

\fbox{\parbox{0.95\columnwidth}{
\textbf{Transport Failure $\neq$ Distribution Shift}\\[4pt]
\textbf{Distribution shift:} The prompts change between training and deployment.\\[2pt]
\textbf{Transport failure:} Prompts can be \emph{identical}, but your policy changes the kinds of responses it generates, and the judge's bias interacts with that (style, verbosity, self-preference). The mapping $S \to Y$ changes even though $X$ doesn't.\\[2pt]
CJE tests this directly: for each target policy $\pi'$, we check whether $\E_{\pi'}[Y - f(S,X)] = 0$. If not, the calibration doesn't transport---either recalibrate or refuse level claims.
}}

\paragraph{Calibration-aware uncertainty \& influence-function stacking.}
Treating learned $R=\hat f(S)$ as fixed understates uncertainty when the oracle slice is small. CJE addresses this via calibration-aware inference, which propagates calibration uncertainty into confidence intervals using a delete-one-fold jackknife \citep{EfronStein1981,Bickel1993}. For ensemble estimation, CJE stacks multiple estimators by minimizing IF covariance over the simplex \citep{Wolpert1992StackedGeneralization,vanderLaan2007SuperLearner}. Details in \cref{sec:method}.

\section{Methods}
\label{sec:method}

The CJE pipeline has three components: reward calibration to prevent preference inversion, weight stabilization for importance-weighted estimators, and uncertainty-aware inference for valid CI coverage. For open-ended generation under nontrivial policy shifts, reward calibration and inference are the essential components; weight stabilization addresses IPS/DR variance but faces coverage limitations when TTC is low (\cref{sec:intro}). \Cref{fig:pipeline} provides an overview.

\begin{figure}[t]
  \centering
  \includegraphics[width=\columnwidth]{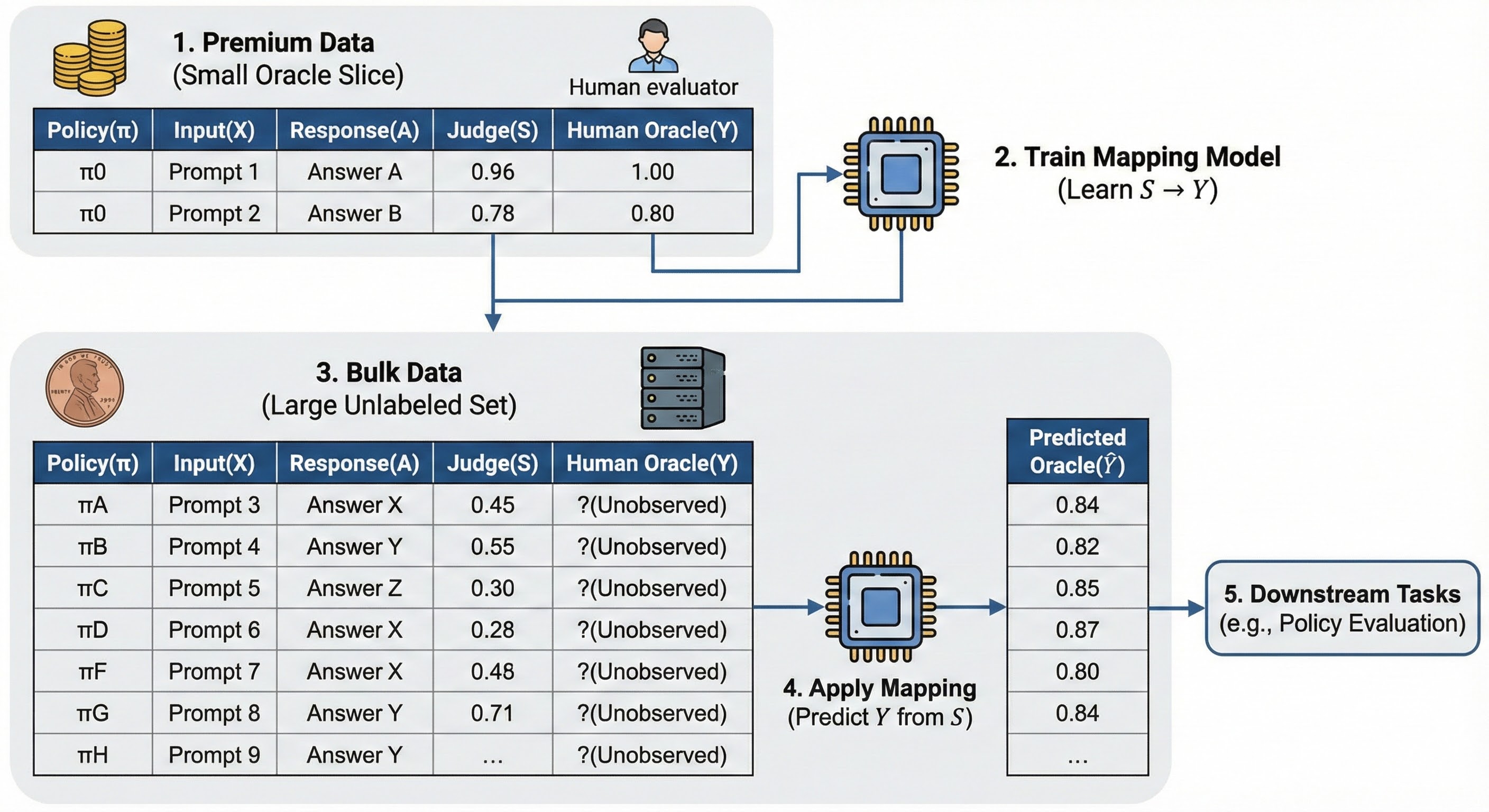}
  \caption{\textbf{CJE pipeline overview.} A small oracle slice (5--25\%) provides expensive oracle labels to train a calibration model ($S \to Y$). The learned mapping is then applied to bulk evaluation data where oracle labels are unavailable, enabling policy evaluation at a fraction of the cost. (In experiments, the oracle is \texttt{gpt-5-2025-08-07}; in production this would typically be human raters or downstream KPIs.)}
  \label{fig:pipeline}
\end{figure}

\CJE\ encodes justified knowledge as model restrictions: (i) monotonicity and mean-preservation for reward calibration, (ii) mean-one constraint with $S$-monotonicity for weights, (iii) nuisance-orthogonality for estimators, and (iv) convex combination for stacking. All learners are \emph{cross-fitted}; by the \emph{Efficiency via Model Restriction} theorem, these restrictions lower the efficiency bound (see \cref{sec:theory}).

\subsection{Reward calibration (two-stage isotonic)}
\label{ssec:method-autocal}
On the oracle slice $\{(S_i,Y_i)\}$, fit a \emph{mean-preserving} calibrator $R=f(Z(S))$ with $K$-fold
cross-fitting:
\begin{itemize}[leftmargin=*,itemsep=2pt,topsep=2pt]
\item Two-stage mode (recommended): Fit a smooth index $Z(S,X)=g(S,X)$
(splines+ridge), map to mid-ranks $U=\mathrm{ECDF}\{Z(S,X)\}$, then fit isotonic
$\hat h_{\uparrow}(U)$. Predictions are $R=\hat h_{\uparrow}(\mathrm{ECDF}\{g(S,X)\})$.
\item Monotone-only mode: Isotonic regression directly on $S$:
$\hat f_{\uparrow}\in\arg\min_{f\in\mathcal{M}_\uparrow}\sum_{i\in O}\big(Y_i-f(S_i)\big)^2$.
PAVA preserves the slice mean exactly. Use when covariates $X$ are unavailable.
\end{itemize}
\emph{Remark (covariate flexibility).} The terminal isotonic step preserves mean-honesty regardless of the first-stage model, so $g$ can be \emph{any} cross-fitted learner (splines, random forests, gradient boosting) and can incorporate covariates $X$ beyond $S$. In our experiments, including response length as a first-stage covariate improves ranking accuracy: LLM judges often favor longer responses independent of quality, and conditioning on length removes this confounder from the $S \to Y$ mapping.
Let $R^{\mathrm{OOF}}$ denote OOF predictions used along the IF path; the point
estimate may use the pooled fit. The terminal isotonic step makes the calibrator \emph{mean-honest} in either
mode. We treat $\hat f$ as learned and propagate its uncertainty via calibration-aware inference (\cref{ssec:method-oua}).

\begin{figure}[t]
  \centering
  \includegraphics[width=\columnwidth]{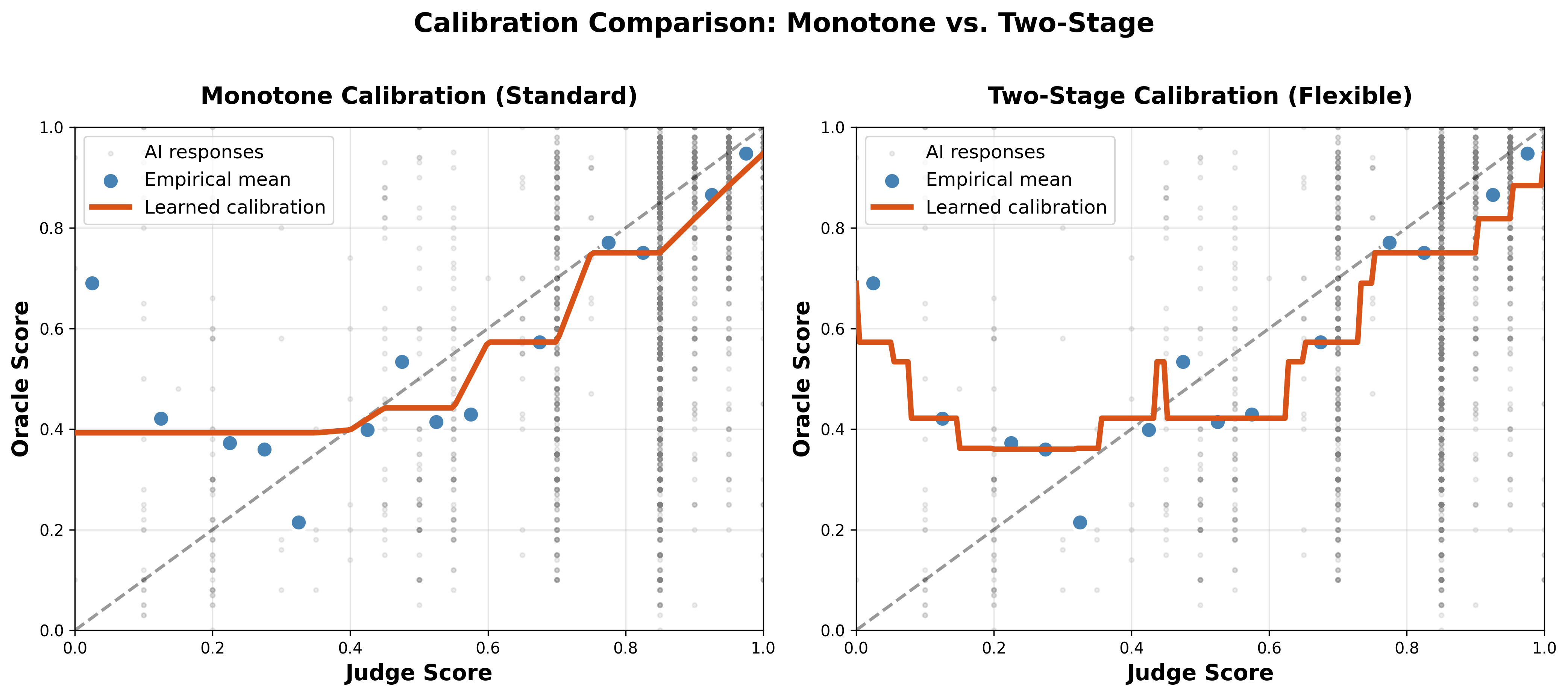}
  \caption{\textbf{Reward calibration: monotone vs.\ two-stage.} Left: monotone-only isotonic regression enforces monotonicity but cannot capture non-monotonic patterns in $\E[Y \mid S]$. Right: two-stage calibration (spline index $\to$ isotonic) can fit flexible patterns while preserving the mean.}
  \label{fig:reward-calibration}
\end{figure}

\subsection{Weight stabilization (unit-mean ratios via OOF stacking of $S$-monotone candidates)}
\label{ssec:method-simcal}
\emph{Scope and practical relevance:} Weight stabilization applies only to IPS and DR estimators. For open-ended generation under nontrivial policy shifts, coverage sparsity (low TTC) limits the utility of logged data (\cref{sec:intro}); in such regimes, Direct is preferred. Weight stabilization remains valuable in constrained settings (short outputs, near-clone shifts, tool-choice bandits) where TTC $\geq 0.70$, or when combining logged data with fresh draws via DR. The Direct method---which generates fresh responses under $\pi'$---requires no importance weights.

Let $W_{\pi'}^{\mathrm{m1}}$ be the self-normalized (SNIPS) baseline. For each fold $k$:
\begin{enumerate}[leftmargin=*,itemsep=2pt,topsep=2pt]
\item Monotone projections (train on $I_{\neg k}$): Fit increasing and decreasing isotonic maps on $S$
(the latter by running PAVA on $-S$), rescale each to mean one on $I_{\neg k}$, and predict OOF candidates on $I_k$:
$W_{\uparrow}^{\mathrm{OOF}}$, $W_{\downarrow}^{\mathrm{OOF}}$; optionally include an identity candidate
$W_{\mathrm{base}}^{\mathrm{OOF}}\equiv 1$.
\item OOF stacking (variance-aware): Define residuals $\Delta_i$ used by the downstream estimator:
$\Delta_i{=}R_i$ for IPS and $\Delta_i{=}R_i-\hat m(X_i,A_i)$ for DR (with $\hat m{=}\hat q$ below). Let
$U_c=W_c^{\mathrm{OOF}}\Delta$ for $c\in\{\mathrm{base},\uparrow,\downarrow\}$ and compute
$\hat\Sigma_{cd}=\mathrm{cov}(U_c,U_d)$ (tiny ridge if needed). Choose simplex weights
\[
\hat\beta\in\arg\min_{\beta\in\Delta_3}\ \beta^\top\hat\Sigma\,\beta,\qquad
W^{\mathrm{stack}}=\sum_c \hat\beta_c\,W_c^{\mathrm{OOF}},
\]
then renormalize $W^{\mathrm{stack}}$ to mean one.
\item Light variance guard (optional; $\rho{=}1$ by default): Cap variance at absolute threshold $\rho$:
\[
\alpha=\min\!\Big\{1,\ \sqrt{\frac{\rho}{\Var(W^{\mathrm{stack}})}}\Big\},\quad
W^{\mathrm{blend}}=1+\alpha\,(W^{\mathrm{stack}}-1),\quad
\hat W_{\pi'}=W^{\mathrm{blend}}/\bar W^{\mathrm{blend}}.
\]
Each step preserves (or restores) the sample mean. \emph{Note:} The theoretical guarantee (Lemma~\ref{lem:guard}) uses a \emph{relative} bound $\Var(\hat W)\le\rho\,\Var(W^{\mathrm{m1}})$; the implementation's absolute cap is more conservative when $\Var(W^{\mathrm{m1}}){>}1$.
\end{enumerate}
\emph{Remark (transport view).} In the continuous case the ideal component is
$m^\star(s)=p_{S\mid\pi'}(s)/p_{S\mid\pi_0}(s)$; weight stabilization stacks increasing and decreasing $S$-monotone candidates to minimize IF variance, yielding approximately monotone weights when one direction dominates.
When $m^\star(S)$ is not approximately monotone, isotonic projection trades bias for variance; in the low-overlap regimes where this matters most, we recommend Direct over IPS/DR (see CLE bound, \cref{sec:background}).
(Weight stabilization visualization in \cref{app:weight-viz}.)

\subsection{Estimators: Direct, Calibrated IPS, and Calibrated DR}
\label{ssec:method-estimators}
Let $\hat q(z)\approx \E[R\mid Z{=}z]$ be the cross-fitted calibrator (in experiments, $\hat q = \hat f^{(-k)}(Z)$; we do not fit a separate text-conditioned reward model).
Define $\hat g_{\pi'}(x)=\E_{\pi'}[\hat q(Z)\mid X{=}x]$; all nuisances are cross-fitted and OOF predictions
are used inside IFs.

\paragraph{Direct (fresh responses under $\pi'$; no overlap needed).}
When responses can be generated under the target policy, the Direct estimator uses calibrated rewards without importance weighting:
\[
\widehat V_{\mathrm{Direct}}(\pi') = \frac{1}{n}\sum_{i=1}^n R_i = \frac{1}{n}\sum_{i=1}^n \hat f(Z_i).
\]
For inference, the bias-corrected augmented estimator $\hat\theta_{\mathrm{aug}}$ (\cref{eq:theta-aug}) achieves ${\sim}$95\% coverage by correcting for calibration uncertainty.
Direct requires no overlap assumption and is the \textbf{recommended default} when fresh generation is feasible (\cref{tab:estimator-guide}).

\paragraph{Calibrated IPS.}
\[
\widehat V_{\mathrm{IPS}}(\pi')=\frac{1}{n}\sum_{i=1}^n \hat W_{\pi',i}\,R_i,\qquad
\phi^{\mathrm{IPS}}_i=\hat W_{\pi',i}\,R^{\mathrm{OOF}}_i-\widehat V_{\mathrm{IPS}}.
\]

\paragraph{Calibrated DR (doubly robust with calibrated rewards).}
\[
\widehat V_{\mathrm{DR}}(\pi')
=\frac{1}{n}\sum_{i=1}^n\Big\{\hat g_{\pi'}(X_i)+\hat W_{\pi',i}\big(R_i-\hat q(Z_i)\big)\Big\},
\quad
\phi^{\mathrm{DR}}_i=\hat g_{\pi'}(X_i)+\hat W_{\pi',i}\big(R^{\mathrm{OOF}}_i-\hat q^{\mathrm{OOF}}_i\big)-\widehat V_{\mathrm{DR}}.
\]

\emph{Remark (DR vs Direct).} DR combines an outcome model with a weighted residual correction. Under good overlap, the IPS component can improve efficiency; under poor overlap (TTC $< 0.7$), it adds noise. In our experiments, DR effectively reduces to its outcome model, and Direct achieves the same accuracy without requiring logprobs (\cref{sec:experiments}, Q3).

\subsection{Influence-function stacking (variance-optimal convex ensembling)}
\label{ssec:method-stack}
For a small library $\mathcal{E}$ of regular estimators (e.g., DR/TMLE/MRDR variants, capped IPS), form
the matrix of centered IF columns $\Phi=[\phi^{(e)}]_{e\in\mathcal{E}}$ (computed OOF on the same folds),
estimate $\hat\Sigma=\tfrac{1}{n}\Phi^\top\Phi+\lambda I$, and solve the simplex QP
\[
\hat\alpha\in\arg\min_{\alpha\in\Delta}\ \alpha^\top \hat\Sigma\,\alpha,\qquad
\widehat V_{\mathrm{stack}}=\sum_{e\in\mathcal{E}}\hat\alpha_e\,\widehat V^{(e)},\quad
\phi^{\mathrm{stack}}=\sum_{e\in\mathcal{E}} \hat\alpha_e\,\phi^{(e)}.
\]

\subsection{Calibration-aware inference}
\label{ssec:method-oua}

Standard CIs treat the calibration function $\hat f$ as fixed, ignoring first-stage estimation error. When the oracle slice is small, this understates total variance and produces severely narrow intervals. In our experiments, naive CIs on uncalibrated scores achieve 0\% coverage; bootstrap inference with bias-corrected estimation achieves ${\sim}$95\% coverage for Direct mode.

\paragraph{Variance decomposition.}
Total variance decomposes into two components:
\[
\Var_{\text{total}}(\widehat V) = \Var_{\text{eval}}(\widehat V \mid \hat f) + \Var_{\text{cal}}(\hat f),
\]
where $\Var_{\text{eval}}$ is the standard evaluation variance (estimated via the IF) and $\Var_{\text{cal}}$ is the variance due to calibration uncertainty.

\paragraph{Delete-one-fold jackknife.}
We estimate $\Var_{\text{cal}}$ via a delete-one-oracle-fold jackknife:
\begin{enumerate}[leftmargin=*,itemsep=2pt,topsep=2pt]
\item Partition the oracle slice into $K$ folds.
\item For each fold $k \in \{1,\ldots,K\}$, refit calibration omitting fold $k$: $\hat f^{(-k)}$.
\item Compute the leave-one-fold-out estimate: $\widehat V^{(-k)} = \widehat V(\hat f^{(-k)})$.
\item Estimate calibration variance:
\[
\widehat{\Var}_{\text{cal}} = \frac{K-1}{K} \sum_{k=1}^K \big(\widehat V^{(-k)} - \bar V\big)^2, \quad \bar V = \frac{1}{K}\sum_k \widehat V^{(-k)}.
\]
\end{enumerate}

\paragraph{Calibration-aware confidence intervals.}
\[
\widehat{\Var}_{\text{total}} = \widehat{\Var}_{\text{eval}} + \widehat{\Var}_{\text{cal}}, \qquad
\text{CI}_{95\%} = \widehat V \pm 1.96 \sqrt{\widehat{\Var}_{\text{total}}}.
\]
The calibration uncertainty share $\widehat{\Var}_{\text{cal}} / \widehat{\Var}_{\text{total}}$ ranges from 0\% (100\% oracle fraction, where no calibration is needed) to over 50\% (5\% oracle fraction), explaining why ignoring it produces catastrophically narrow CIs.

\paragraph{Bootstrap with bias-corrected estimation.}
While the jackknife captures variance, it does not correct for bias in the plug-in estimator $\widehat V = n^{-1}\sum_i \hat f(S_i)$. This bias arises because $\hat f$ is fit on the oracle slice and applied to all samples, creating a covariance between the calibration function and the evaluation samples. We address this with the one-step estimator derived from the efficient influence function for missing-outcome estimation (\cref{app:eif-direct}; see also \citet{Chen2026EfficientLLMJudge} for the LLM-as-judge instantiation):
\begin{equation}
\label{eq:theta-aug}
\hat\theta_{\mathrm{aug}} = \frac{1}{n}\sum_{i=1}^n \hat f(S_i) + \frac{1}{|L|}\sum_{i \in L}\big(Y_i - \hat f^{(-i)}(S_i)\big),
\end{equation}
where $\hat f$ is the full calibrator, $\hat f^{(-i)}$ is the calibrator trained on folds excluding sample $i$'s cluster, and $L$ is the oracle slice. The first term uses the full calibrator for lower variance; the second term uses cluster-out-of-fold predictions to avoid overfitting in the residual correction. If oracle sampling is non-uniform (e.g., oversampling uncertain prompts), replace the constant $|L|/n$ with sample-specific inclusion probabilities $e(Z_i)$ in the general AIPW form (\cref{app:eif-direct}). Combined with bootstrap (refitting the calibrator in each replicate), this achieves ${\sim}$95\% coverage across oracle fractions (\cref{tab:uq-comparison}), compared to 70--89\% for jackknife-only variance estimation.

\paragraph{Computational cost.}
Bootstrap with $\hat\theta_{\mathrm{aug}}$ requires $B$ refits of the reward calibrator (we use $B{=}2{,}000$). Since isotonic regression is $O(n \log n)$, the overhead remains modest relative to the cost of teacher forcing and response generation. The full bootstrap procedure (\cref{alg:bootstrap-direct}) uses prompt-level cluster resampling to preserve within-prompt dependence \citep{CameronGelbachMiller2008}, with a resample-until-valid loop ensuring each replicate has sufficient oracle labels.

\section{Theory: EIF, Model Restriction, and Efficiency}
\label{sec:theory}

This section explains why the projections in CJE are not heuristics but information-optimal. We show: (i) using a judge as a surrogate can only reduce variance, (ii) encoding justified restrictions like monotonicity and mean-one weights tightens efficiency bounds, and (iii) our estimators attain these bounds under standard conditions. Proofs and technical lemmas are deferred to the appendix.

\paragraph{Surrogate model and EIF.}
Let $R^\star=\E[Y\mid S]$ and $m^\star(S)=\E[W_{\pi'}\mid S]$. Under \emph{mean sufficiency}
$\E[Y\mid X,A,S]=R^\star(S)$,
\[
V(\pi')=\E\!\big[m^\star(S)\,R^\star(S)\big],
\qquad
\phi_{\mathrm{sur}}(O;\pi') \;=\; g^\star_{\pi',R}(X) \;+\; m^\star(S)\!\big(R^\star-q^\star_R(X,A)\big) \;-\; V(\pi'),
\]
with $q^\star_R(x,a)=\E[R^\star\mid X{=}x,A{=}a]$ and $g^\star_{\pi',R}(x)=\sum_a \pi'(a\mid x)\,q^\star_R(x,a)$.

\begin{theorem}[Surrogate EIF and variance reduction]
\label{thm:sur-eif}
Let $\phi_{\mathrm{uncon}}$ be the canonical gradient in the nonparametric model that does not use $S$.
Then $\phi_{\mathrm{sur}}$ is the canonical gradient in the surrogate model, and
$\Var(\phi_{\mathrm{sur}})\le \Var(\phi_{\mathrm{uncon}})$,
with strict inequality unless $W_{\pi'}$ is $\sigma(S)$-measurable and $R^\star$ is degenerate.
\end{theorem}

\noindent\emph{In plain terms:} If the judge carries any information about the oracle, using it can only help variance and never hurt it. Coarser judges (garblings) are strictly worse. (This variance reduction from conditioning is a standard result; we state it to establish the foundation for subsequent results.)

\paragraph{Efficiency via Model Restriction.}
We now formalize a standard principle from semiparametric efficiency theory and show how it unifies the CJE projections.
Let $L^2_0$ be the mean-zero Hilbert space with inner product $\langle f,g\rangle=\E[fg]$.
Let $\mathcal{M}$ be a baseline semiparametric model with tangent space $T(P)$ and canonical gradient
(EIF) $\phi^\star$. When justified knowledge defines a restricted model $\mathcal{M}_\mathcal{C} \subset \mathcal{M}$
(e.g., mean sufficiency in $S$, monotonicity of conditional expectations), the restricted tangent space
$T_\mathcal{C}(P) \subseteq T(P)$ yields a new canonical gradient $\phi^\star_\mathcal{C}$.

\begin{theorem}[Efficiency via Model Restriction]
\label{thm:knowledge-riesz}%
\label{thm:ckp}
\label{thm:model-restriction}%
\emph{(i) Efficiency improvement.}
The canonical gradient in the restricted model satisfies
$\|\phi^\star_\mathcal{C}\|_2^2 \le \|\phi^\star\|_2^2$, with equality iff $T_\mathcal{C}(P) = T(P)$.
If $T_{\mathcal{C}_1}(P) \supseteq T_{\mathcal{C}_2}(P)$ (i.e., $\mathcal{C}_2$ imposes stronger restrictions), then $\|\phi^\star_{\mathcal{C}_2}\|_2^2\le \|\phi^\star_{\mathcal{C}_1}\|_2^2$.
\emph{(ii) Attainability.} One-step/TMLE estimators with cross-fitting attain the efficiency bound
$\Var(\phi^\star_\mathcal{C})$ in the restricted model $\mathcal{M}_\mathcal{C}$.
\end{theorem}

\noindent\emph{In plain terms:} Every justified assumption we encode as a model restriction is ``free variance reduction'': we shrink the tangent space without changing the estimand.
If a restriction is violated, projection introduces bias; our diagnostics and gates (\cref{sec:background,sec:limitations}) detect the regimes where this matters.

\noindent\emph{Remark (prior art and novelty).} The efficiency improvement from model restriction is well-known in semiparametric theory \citep{Newey1990Semiparametric,Bickel1993}. Our contribution is not the principle itself but its systematic application: we show that reward calibration, weight stabilization, and influence-function stacking are all instances of model restriction, and we provide the operational machinery (transport audit, CLE bound and TTC/AB diagnostics, calibration-aware inference) that makes the theory usable for LLM evaluation.

\paragraph{Consequences for \CJE}
(i) \emph{Conditioning:} taking $\mathcal{C}=\{f:\,f=\E[f\mid S]\}$ recovers \cref{thm:sur-eif}.
(ii) \emph{Mean-one monotone weights:} restricting the weight component to
$\{w:\E[w]=1,\; w\uparrow S\}$ motivates weight stabilization; the exact $\mathrm{IsoMeanOne}_S$ projection weakly reduces dispersion in finite samples by majorization.
(iii) \emph{Stacking:} restricting to the convex hull of candidate IF columns gives the
variance-optimal convex ensemble.

\paragraph{Direct mode: missing-data estimation.}
When fresh responses can be generated under $\pi'$, evaluation becomes a missing-data problem: we observe cheap surrogate predictions $\hat\mu(Z)$ for all samples but expensive oracle labels $Y$ only for a labeled subset $L$.

\begin{prop}[Direct: EIF and one-step estimator]
\label{prop:direct-eif}
Let $\theta := \E_{\pi'}[Y]$ be the policy value. Under MAR ($L \perp Y \mid Z$) with $e(z) := \Pr(L{=}1 \mid Z{=}z)$, the efficient influence function is
\[
\phi_{\mathrm{Direct}}(O) = \mu(Z) - \theta + \frac{L}{e(Z)}\big(Y - \mu(Z)\big),
\]
where $\mu(z) = \E[Y \mid Z{=}z]$. Under approximately uniform labeling ($e \approx |L|/n$), the one-step estimator is $\hat\theta_{\mathrm{aug}}$ (\cref{eq:theta-aug}), which attains the semiparametric efficiency bound.
\end{prop}

\noindent\emph{In plain terms:} Direct is the recommended default when fresh generation is feasible. It requires no overlap assumption and achieves optimal efficiency for missing-data estimation. (Proof in \cref{app:eif-direct}.)

\paragraph{OPE mode: reweighting logged data.}
When fresh generation is unavailable, importance-weighted estimators reweight logged data from $\pi_0$ to estimate $V(\pi')$. The following results apply to this regime.

\begin{prop}[Cal-IPS: mean correctness and dispersion control]
\label{prop:calips}
Let $R=\hat f(Z(S))$ be the reward calibrator (monotone in $S$ or two-stage index; cross-fitted), and let
$W^{\mathrm{m1}}_{\pi'}\!\triangleq W_{\pi'}/\bar W_n$ denote the self-normalized (SNIPS) ratios, where $\bar W_n = n^{-1}\sum_i W_{\pi',i}$.
Let $\hat W_{\pi'}$ be stabilized weights (OOF stack $+$ variance guard $\rho\ge1$).
Then $\widehat V_{\mathrm{IPS}}=\tfrac{1}{n}\sum_i \hat W_{\pi',i}R_i \to_p V(\pi')$.
Under the exact $\mathrm{IsoMeanOne}_S$ projection (see App.~\ref{app:projections}),
$\operatorname{Var}_n(\hat W_{\pi'})\le \rho\,\operatorname{Var}_n(W^{\mathrm{m1}}_{\pi'})$ and
$\ESS(\hat W_{\pi'})\ge \ESS(W^{\mathrm{m1}}_{\pi'})$; the reference implementation uses a simpler normalization that empirically yields similar ESS gains.
\end{prop}

\begin{theorem}[Calibrated DR: $\sqrt{n}$ limits and efficiency]
\label{thm:dr-eff}
Assume mean sufficiency, suitable tails/moments, and cross-fitted nuisances satisfying the one-of-two
rate condition
$\|\hat q-q^\star_R\|_2\cdot \|\hat W_{\pi'}-m^\star\|_2 = o_p(n^{-1/2})$
(e.g., either factor $=o_p(n^{-1/4})$). Then
\[
\sqrt{n}\big(\widehat V_{\mathrm{DR}}-V(\pi')\big)\ \leadsto\ \mathcal{N}\!\big(0,\Var(\phi_{\mathrm{sur}})\big),
\]
i.e., calibrated DR attains the surrogate efficiency bound. (Efficiency attainment under the product-rate condition is a standard property of DR estimators \citep{BangRobins2005,Chernozhukov2018}; our contribution is verifying that the CJE construction satisfies these conditions.)
\end{theorem}

\paragraph{Additional results (in appendix).}
The budgeted information bound (\cref{thm:budgeted}) formalizes variance caps on importance weights: limiting weight variance ($\rho$) trades a bit of asymptotic efficiency for much better finite-sample behavior. Influence-function stacking (\cref{thm:stack}) shows that variance-optimal convex combinations of estimators achieve variance no worse than the best component. Both results apply to OPE regimes (IPS/DR); Direct mode does not use weights and is unaffected. Full statements and proofs are in \cref{app:budgeted-stacking}.

\begin{prop}[Calibration-aware jackknife variance estimation]
\label{prop:oua}
Let $\widehat V(\hat f)$ be any \CJE\ estimator that uses $R=\hat f(S)$ (cross-fitted along the IF path).
Under $L^2(P_S)$-consistency of $\hat f$ and non-overlapping folds, the delete-one-oracle-fold jackknife provides a variance estimate for the calibration-induced component, so that
$\widehat{\Var}_{\mathrm{total}}=\widehat{\Var}_{\mathrm{main}}+\widehat{\Var}_{\mathrm{cal}}$
estimates total variance.
Theoretical consistency requires regularity conditions on the functional $f\mapsto \widehat V(f)$; isotonic regression's non-smoothness complicates formal guarantees. Empirically, we verify valid coverage across 50 seeds (\cref{sec:experiments}).
\end{prop}

\paragraph{Policy-wise mean transport test.}
Let $f(S,X)$ be a calibration function learned on the oracle slice under the logging policy $\pi_0$. For a target policy $\pi'$, define the oracle and surrogate values
$V_{\pi'}^{\mathrm{oracle}} := \E_{\pi'}[Y]$ and $V_{\pi'}^{\mathrm{sur}} := \E_{\pi'}[f(S,X)]$,
and the residual $\varepsilon_{\pi'} := Y - f(S,X)$.

\begin{prop}[Mean transport equivalence]
\label{prop:mean-transport}
$V_{\pi'}^{\mathrm{oracle}} - V_{\pi'}^{\mathrm{sur}} = \E_{\pi'}[\varepsilon_{\pi'}]$.
Therefore, the null hypothesis $H_{0,\pi'}: \E_{\pi'}[\varepsilon_{\pi'}] = 0$ is equivalent to mean-unbiasedness of the surrogate for policy $\pi'$: $H_{0,\pi'} \iff V_{\pi'}^{\mathrm{oracle}} = V_{\pi'}^{\mathrm{sur}}$.
\end{prop}

\noindent
\emph{Remark.} The mean residual decomposes into a \emph{transportability gap} (failure of S2: $\E_{\pi'}[\mu_{\pi'}(S,X) - \mu_{\pi_0}(S,X)]$) and \emph{calibration estimation error} ($\E_{\pi'}[\mu_{\pi_0}(S,X) - f(S,X)]$). The latter vanishes as the oracle slice grows, so the test asymptotically isolates transportability failures. This test is strictly weaker than full transportability (S2): it guarantees mean-unbiasedness for policy-wise values but does not rule out conditional miscalibration within subgroups or at tails.

\noindent
\emph{Data requirement.} Applying this test requires a small oracle slice \emph{under each target policy $\pi'$} (or per evaluation environment, e.g., time window or subgroup). We recommend collecting a modest audit sample (${\sim}50$--$200$ labels per policy) to validate that the base-trained calibrator can be safely reused; if a policy fails the test, either recalibrate with policy-specific data or fall back to oracle-only evaluation for that cell.

\paragraph{Discussion.}
CJE addresses two distinct statistical problems: (i) \emph{missing-data estimation} (Direct mode, \cref{prop:direct-eif}), where fresh responses are generated under $\pi'$ and the only missing data are oracle labels; and (ii) \emph{off-policy evaluation} (IPS/DR modes), where logged data from $\pi_0$ must be reweighted to estimate $V(\pi')$.
Direct is the recommended default when fresh generation is feasible: it requires no overlap assumption, achieves the semiparametric efficiency bound for missing-data estimation, and attains the best ranking accuracy in our experiments.
For OPE regimes, \cref{thm:model-restriction} formalizes the model restriction principle: encoding justified knowledge (conditioning on $S$, monotone weights, convex IF combinations) tightens efficiency bounds. The budgeted bound and IF stacking theorems (\cref{app:budgeted-stacking}) provide additional variance control for IPS/DR.
In finite samples, the exact $\mathrm{IsoMeanOne}_S$ projection additionally \emph{majorizes} dispersion (\cref{lem:maj}), explaining the ESS gains delivered by weight stabilization.

\paragraph{Scope and limitations.}
\emph{(i) Mean sufficiency and transport.} The EIF derivations (\cref{thm:sur-eif,thm:dr-eff}) assume mean sufficiency ($\E[Y\mid X,A,S]=R^\star(S)$), a structural assumption enabling the efficiency theory. For valid estimation, the fundamental requirement is transport ($\E_{\pi'}[Y-f(S,X)]=0$), which is strictly weaker: mean sufficiency implies transport, but transport can hold (and be verified via audit) without it. The two-stage calibration mode uses a learned index $Z(S,X)$; the theory generalizes by replacing $S$ with $Z(S,X)$ throughout, provided $Z$ is consistently estimated.
\emph{(ii) Double robustness and weight misspecification.} \cref{thm:dr-eff} assumes at least one nuisance
converges at rate $n^{-1/4}$. When teacher-forced weights are unreliable (as in high-overlap-failure regimes),
DR remains consistent via the outcome model alone; however, the efficiency bound $\Var(\phi_{\mathrm{sur}})$ is
attained only when both nuisances are well-specified. In practice, when IPS fails due to poor coverage, DR
effectively reduces to the Direct Method, which achieves the same consistency without the variance penalty
from unstable weights (see \cref{sec:experiments}, Q3).
\emph{(iii) Weight stabilization convergence.} \cref{thm:budgeted} assumes weight stabilization converges to $m^\star_\rho$; this holds
under standard regularity conditions for isotonic regression and cross-validated stacking (see, e.g.,
\citet{vanderLaan2007SuperLearner}).

\section{Experiments}
\label{sec:experiments}

We evaluate CJE on a large-scale benchmark. We initially expected doubly robust methods to dominate by combining logged data with fresh draws. Instead, we found two surprises: (i) OPE fails catastrophically even after weight stabilization boosts ESS above 90\%, and (ii) DR provides no advantage over Direct under low overlap, effectively reducing to the outcome model alone. These findings motivated the CLE bound and its TTC/AB diagnostics, which explain why high ESS is insufficient.

Our experiments address three questions:
\begin{enumerate}[leftmargin=*,itemsep=2pt,topsep=2pt]
\item Does calibration work? Reward calibration should reduce RMSE, improve ranking accuracy, and with calibration-aware inference restore valid CI coverage.
\item Why does OPE fail despite high ESS? We expected IPS to work after stabilization; CLE predicts failure when TTC is low.
\item Why doesn't DR dominate? Under low overlap, DR's IPS component contributes noise rather than information.
\end{enumerate}

\subsection{Experimental design}
\label{ssec:exp-setup}

\paragraph{Data.}
We started with 5,000 prompts from Chatbot Arena \citep{Zheng2023LLMasJudge}, a crowdsourced platform where users chat with anonymous LLMs. Prompts were randomly selected from first-turn English conversations; after filtering for teacher-forcing reliability, $n{=}4{,}961$ samples remain with complete data for all policies (see Limitations, \cref{sec:limitations}).

\paragraph{Policies.}
Five LLM configurations, each generating responses to all prompts:
\begin{itemize}[leftmargin=*,nosep]
\item \texttt{base}: Llama 3.3 70B with standard system prompt (logging policy $\pi_0$)
\item \texttt{clone}: Same model and prompt as base, different seed (TF reliability test)
\item \texttt{premium}: Llama 3.1 405B (model size effect)
\item \texttt{parallel\_universe\_prompt}: Llama 3.3 70B with modified system prompt (prompt engineering test)
\item \texttt{unhelpful}: Deliberately low-quality, confusing responses (adversarial stress test)
\end{itemize}
Total: ${\sim}25$k responses (${\sim}$5k prompts $\times$ 5 policies; $n{=}4{,}961$ after filtering).

\paragraph{Oracle and judge.}
\textbf{Oracle $Y$:} \texttt{gpt-5-2025-08-07} \citep{OpenAI2025GPT5} quality scores (0--100, normalized to $[0,1]$) collected for every response.
\textbf{Judge $S$:} \texttt{gpt-4.1-nano-2025-04-14} scores on all responses (temperature 0); approximately $16\times$ cheaper than oracle.\footnote{Pricing as of November 2025 via OpenAI API.}
\emph{Note:} The oracle model operates at temperature 1.0 (temperature 0 was not available for this model at experiment time), introducing stochasticity in oracle labels; we did not repeat-score responses, so coverage results may be conservative.
All estimators except \texttt{naive-direct} calibrate judge scores using reward calibration. Base variants use monotone-only calibration ($Z{=}S$); \texttt{+cov} variants use two-stage calibration with response length as a covariate ($Z{=}g(S,\text{length})$).
For DR methods, rewards $R=\hat f_{\mathrm{all}}(Z)$ use the full calibrator while the outcome model $\hat q=\hat f^{(-k)}(Z)$ uses cross-fitted (leave-one-fold-out) predictions; the difference $R-\hat q$ captures residual information for the IPS correction.

\paragraph{Experimental grid.}
We cross 13 estimators $\times$ 5 oracle fractions (5\%, 10\%, 25\%, 50\%, 100\%) $\times$ 5 sample sizes (500, 1k, 2k, 3k, 5k) $\times$ 50 random seeds = 16,250 total runs (1,250 per estimator). Reported metrics are means across seeds within each (oracle fraction, sample size) cell.

\paragraph{Metrics.}
\emph{Ranking:} Pairwise accuracy ($\text{PairAcc} = \binom{P}{2}^{-1}\sum_{j<k}\mathbf{1}\{(\hat V_j - \hat V_k)(V^\star_j - V^\star_k) > 0\}$, the fraction of correctly ordered policy pairs), Top-1 accuracy, Kendall's $\tau$.
\emph{Magnitude:} RMSE$^d$ (oracle-noise-debiased): $\mathrm{RMSE}^d = \sqrt{\max(0,\, \mathrm{MSE} - \widehat{\Var}(\bar Y))}$, where $\widehat{\Var}(\bar Y)$ is the variance of the oracle mean estimate (conservatively bounded by $0.25/n$ per policy). This isolates estimator error from irreducible oracle sampling noise.
\emph{Uncertainty:} Coverage (95\% CI capture rate).
\emph{Note:} Accuracy and uncertainty metrics exclude the \texttt{unhelpful} policy (which intentionally violates transportability); ranking metrics include all five policies to test whether estimators correctly identify the adversarial policy as worst.

\subsection{Results}
\label{ssec:exp-results}

\begin{table}[H]
\centering
\caption{Accuracy \& Uncertainty Metrics}
\label{tab:accuracy}
\setlength{\tabcolsep}{4pt} 
\resizebox{\linewidth}{!}{%
\begin{tabular}{lrrrrrrrr}
\toprule
Estimator & $\mathrm{RMSE}^{\mathrm{d}}$ $\downarrow$ & IS (interval score) $\downarrow$ & Coverage \\
\midrule
naive-direct & 0.0828 & 2.9580 & 0.0 & \textbf{0.0039} & 90.9 & 79.6 & 0.817 & 0.4 \\
direct & \textbf{0.0223} & \textbf{0.0542} & 85.0 & \underline{0.0068} & 91.8 & 84.1 & 0.837 & 0.7 \\
direct+cov & 0.0256 & \underline{0.0568} & 86.1 & 0.0074 & \textbf{94.4} & \textbf{89.5} & \textbf{0.888} & 1.6 \\
SNIPS & 0.1596 & 0.7379 & 98.3 & 0.1815 & 38.3 & 8.7 & -0.235 & \textbf{0.3} \\
SNIPS+cov & 0.1842 & 0.7544 & 97.8 & 0.1842 & 37.4 & 6.4 & -0.252 & \underline{0.3} \\
calibrated-ips & 0.0246 & 0.5727 & 99.1 & 0.0963 & 47.1 & 19.0 & -0.059 & 0.3 \\
calibrated-ips+cov & 0.0262 & 0.6084 & 99.2 & 0.1027 & 46.5 & 20.5 & -0.070 & 0.4 \\
dr-cpo & 0.0459 & 0.2203 & 99.2 & 0.0371 & 78.3 & 46.3 & 0.566 & 3.1 \\
dr-cpo+cov & 0.0570 & 0.2657 & 99.4 & 0.0449 & 79.1 & 50.1 & 0.581 & 9.6 \\
calibrated-dr-cpo & 0.0224 & 0.1357 & 99.5 & 0.0230 & 90.9 & 81.0 & 0.818 & 3.2 \\
calibrated-dr-cpo+cov & 0.0260 & 0.1398 & 99.3 & 0.0236 & \underline{94.1} & \underline{87.7} & \underline{0.882} & 9.7 \\
stacked-dr & \underline{0.0224} & 0.0788 & \textbf{95.6} & 0.0125 & 91.9 & 83.1 & 0.838 & 8.9 \\
stacked-dr+cov & 0.0281 & 0.0806 & \underline{96.1} & 0.0128 & 92.1 & 84.8 & 0.842 & 38.2 \\
\bottomrule
\end{tabular}
}
\footnotesize{
$\downarrow$: lower is better, $\uparrow$: higher is better, $\to 95$: target is 95\%. 
\textbf{Bold}: best (closest to target for Coverage \%), \underline{underlined}: second-best. 
Metrics averaged across all regimes.
}
\end{table}

\paragraph{Q1: Does calibration work?}
Yes, for both accuracy and inference. \Cref{tab:accuracy} shows calibration improves accuracy across method families.
For Direct: uncalibrated \texttt{naive-direct} achieves 90.9\% pairwise vs.\ calibrated \texttt{direct+cov} at 94.4\%.
For DR: raw-weight \texttt{dr} achieves 78.3\% pairwise ($\tau{=}0.57$) vs.\ weight-stabilized \texttt{calibrated-dr} at 90.9\% ($\tau{=}0.82$), a 12.6 percentage point gain (both use reward calibration; the gain is from weight stabilization).
For inference: uncalibrated CIs (\texttt{naive-direct}) achieve 0\% coverage: nominal 95\% CIs capture the true value in 0/50 seeds, a failure mode independently replicated by \citet{Chen2026EfficientLLMJudge}, who observe 0\% coverage for naive estimation under judge bias.
Bootstrap inference with bias-corrected estimation ($\hat\theta_{\mathrm{aug}}$) restores coverage: Direct achieves ${\sim}$95\% across oracle fractions (\cref{tab:uq-comparison}), as does stacked-DR.
The calibration uncertainty share of total variance ranges from 0\% (100\% oracle) to over 50\% (5\% oracle), explaining why ignoring calibration uncertainty produces catastrophically narrow CIs (\cref{fig:oua-decomposition}).

\begin{figure}[t]
  \centering
  \includegraphics[width=\columnwidth]{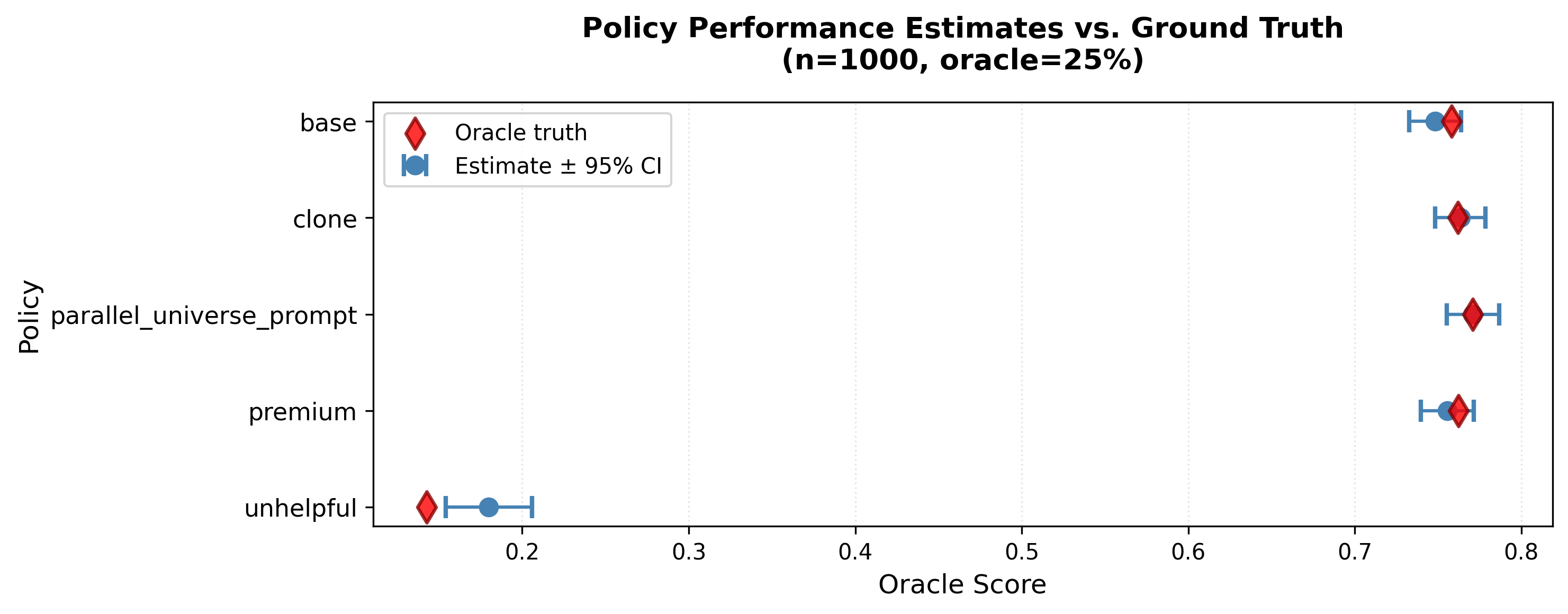}
  \caption{\textbf{CJE output: policy value estimates at $n{=}1000$, 25\% oracle.} Red diamonds show oracle ground truth; blue circles show CJE estimates with 95\% CIs. CIs capture the true value for policies satisfying transportability; \texttt{unhelpful} (which violates transportability) shows slight miscoverage, a failure mode flagged by the transportability test (\cref{fig:transportability-test}).}
  \label{fig:forest-plot}
\end{figure}

\paragraph{Q2: Why does OPE fail despite high ESS?}
This was our main surprise. Despite weight stabilization boosting ESS from $<1\%$ to $>80\%$ for target policies, calibrated IPS ranking performance remains near-random: 47\% pairwise accuracy vs.\ 50\% chance baseline.
We expected weight stabilization to be sufficient; it was not.

OPE faces two compounding challenges in this setting.
\emph{First}, teacher-forcing logprobs are empirically unreliable: even for the clone policy (identical to logger), raw ESS is 26\% rather than the expected ${\sim}$100\% (\cref{tab:weight-diagnostics}), suggesting propensity estimates are corrupted by tokenization mismatches, API nondeterminism, or distribution shift between sampling and scoring.
\emph{Second}, even with stabilized weights, coverage-limited efficiency sets a hard precision floor: TTC ranges from 0.19 to 0.49 for non-clone policies, all below the 0.7 threshold, indicating the logger rarely visits target-typical regions. The resulting CLE factors (24--61$\times$) inflate standard errors beyond what any estimator can overcome.

\paragraph{Q3: Why doesn't DR dominate?}
We expected DR to outperform Direct by combining logged data with fresh draws. Instead, \texttt{direct+cov} slightly outperforms \texttt{calibrated-dr+cov} (94.4\% vs.\ 94.1\%).
The explanation: under low overlap, DR's IPS component contributes noise rather than information.
DR estimators are consistent if \emph{either} the propensity model \emph{or} the outcome model is correct.
Given IPS's failure, DR effectively relies \emph{entirely} on the outcome model $\hat{q}(Z)$, not the propensities.
This is exactly the regime CLE identifies: TTC is low, so logged data contributes little. DR $\approx$ Direct $+$ extra variance from unstable weights.
\emph{Implication:} When teacher forcing is unreliable, prefer Direct over DR; it achieves the same accuracy without requiring logprobs.

\begin{table}[t]
\centering
\caption{\textbf{Ablation: marginal contribution of CJE components.} Metrics averaged over 5 sample sizes $\times$ 5 oracle fractions $\times$ 50 seeds. CI coverage = \% of 95\% CIs containing truth (using bootstrap + $\hat\theta_{\mathrm{aug}}$ for Direct).}
\label{tab:ablation}
\small
\begin{tabular}{@{}lccr@{}}
\toprule
\textbf{Configuration} & \textbf{RMSE$^{\mathrm{d}}$} & \textbf{Coverage} & \textbf{$\Delta$RMSE$^{\mathrm{d}}$} \\
\midrule
\multicolumn{4}{l}{\textit{Effect of Reward Calibration}} \\
~~Direct, no calibration & 0.083 & 0\% & --- \\
~~Direct + reward calibration & 0.023 & 95\% & \textbf{--72\%} \\
\midrule
\multicolumn{4}{l}{\textit{Effect of Weight Stabilization}} \\
~~SNIPS (raw weights) & 0.160 & 98\% & --- \\
~~Calibrated-IPS (+ stabilization) & 0.025 & 99\% & \textbf{--84\%} \\
\midrule
\multicolumn{4}{l}{\textit{Combined Effect on DR}} \\
~~Calibrated DR, raw weights & 0.046 & 99\% & --- \\
~~Calibrated DR + stabilization & 0.023 & 100\% & --50\% \\
~~Calibrated DR + stabilization + cov & 0.026 & 99\% & --43\% \\
\midrule
\multicolumn{4}{l}{\textit{Recommended Configuration}} \\
~~\textbf{Direct + reward calibration + cov} & \textbf{0.026} & 95\% & \textbf{--69\%} \\
\bottomrule
\end{tabular}
\vspace{1mm}

\footnotesize{
\textbf{Key findings:} Reward calibration reduces RMSE by 72\% and restores coverage (0\%→95\% with bootstrap).
Weight stabilization reduces IPS RMSE by 84\%. Both contribute additively.
}
\end{table}

\begin{table}[t]
\centering
\caption{Coverage by UQ Method (target: 95\%). Bootstrap with calibrator refit achieves near-nominal coverage; cluster-robust and jackknife-only methods undercover.}
\label{tab:uq-comparison}
\small
\setlength{\tabcolsep}{4pt}
\begin{tabular}{lcccc}
\toprule
& \multicolumn{2}{c}{Analytic} & \multicolumn{2}{c}{Bootstrap} \\
\cmidrule(lr){2-3} \cmidrule(lr){4-5}
Regime & Cluster-Robust & Cal-Aware Jackknife & $B{=}500$ & $B{=}2000$ \\
\midrule
$n{=}250$, 5\% oracle & \textcolor{red}{21} & \textcolor{red}{85} & 94 & 93 \\
$n{=}250$, 25\% oracle & \textcolor{red}{43} & \textcolor{red}{86} & 95 & 97 \\
$n{=}500$, 5\% oracle & \textcolor{red}{28} & \textcolor{red}{89} & 96 & 97 \\
$n{=}500$, 25\% oracle & \textcolor{red}{40} & \textcolor{red}{84} & 98 & 99 \\
$n{=}1000$, 5\% oracle & \textcolor{red}{21} & \textcolor{red}{75} & 95 & 95 \\
$n{=}1000$, 25\% oracle & \textcolor{red}{35} & \textcolor{red}{83} & 95 & 95 \\
$n{=}2500$, 5\% oracle & \textcolor{red}{15} & \textcolor{red}{70} & 97 & 97 \\
$n{=}2500$, 25\% oracle & \textcolor{red}{52} & \textcolor{red}{85} & 98 & 99 \\
\bottomrule
\end{tabular}

\vspace{1mm}
\footnotesize{\textcolor{red}{Red}: $< 90\%$ (undercoverage). Bootstrap with $\hat\theta_{\mathrm{aug}}$ recommended.}
\end{table}

\paragraph{Component ablation.}
\Cref{tab:ablation} isolates the marginal contribution of each CJE component.
Reward calibration reduces RMSE by 72\% and is necessary to avoid catastrophic undercoverage: without it, 0\% of CIs contain the true value; with bootstrap inference, coverage improves to ${\sim}$95\%.
Weight stabilization reduces IPS RMSE by 84\% while maintaining coverage, validating the $S$-monotone projection approach.
\Cref{tab:uq-comparison} compares UQ methods for Direct mode: naive cluster-robust SEs achieve only 15--52\% coverage; calibration-aware jackknife improves to 70--89\%; bootstrap with $\hat\theta_{\mathrm{aug}}$ (AIPW-style residual augmentation) achieves ${\sim}$95\% by correcting both variance and bias from calibration uncertainty.

\begin{table}[t]
\centering
\caption{Weight Diagnostics: Weight Stabilization. Raw SNIPS weights have catastrophic ESS ($<1\%$) and heavy tails ($\alpha < 2$); weight stabilization recovers ESS $>80\%$ with light tails.}
\label{tab:weight-diagnostics}
\small
\setlength{\tabcolsep}{5pt}
\begin{tabular}{lcccccc}
\toprule
& \multicolumn{2}{c}{ESS (\%)}  & \multicolumn{2}{c}{Tail $\alpha$} & & \\
\cmidrule(lr){2-3} \cmidrule(lr){4-5}
Policy & SNIPS & Stab. & SNIPS & Stab. & TTC & CLE \\
\midrule
Clone & 26.2 & 99.0 & 1.08 & $>$10 & 79.6\% & 1.9$\times$ \\
Parallel Univ. & 0.6 & 95.4 & 0.56 & $>$10 & 49.4\% & 60.5$\times$ \\
Premium & 0.7 & 82.1 & 0.32 & $>$10 & 19.0\% & 23.8$\times$ \\
Unhelpful & 0.4 & 84.6 & 0.13 & $>$10 & 25.6\% & 37.6$\times$ \\
\bottomrule
\end{tabular}

\vspace{1mm}
\footnotesize{TTC = Target-Typicality Coverage; CLE = coverage-limited efficiency factor. Tail $\alpha < 2$ indicates infinite variance.}
\end{table}

\paragraph{Weight diagnostics.}
\Cref{tab:weight-diagnostics} quantifies the effect of weight stabilization and CLE diagnostics.
The mean-one isotonic projection reduces dispersion, yielding large ESS uplifts (e.g., from 0.6\% to 95\% for Parallel Universe), and tail relief (Hill $\alpha > 2$).
Yet ranking performance does not improve.
Note that even the clone policy---identical to logger---has raw ESS of only 26\%, far below the theoretical 100\%, reflecting teacher-forcing brittleness (see Limitations, \cref{sec:limitations}).
For non-clone policies, TTC ranges from 19--49\% with CLE factors of 24--61$\times$, indicating the logger rarely visits target-typical regions. Together, TF unreliability and CLE explain why high stabilized ESS does not translate to accurate ranking: propensity estimates are corrupted, and even perfect propensities would face a hard precision floor.

\begin{figure}[t]
  \centering
  \includegraphics[width=\columnwidth]{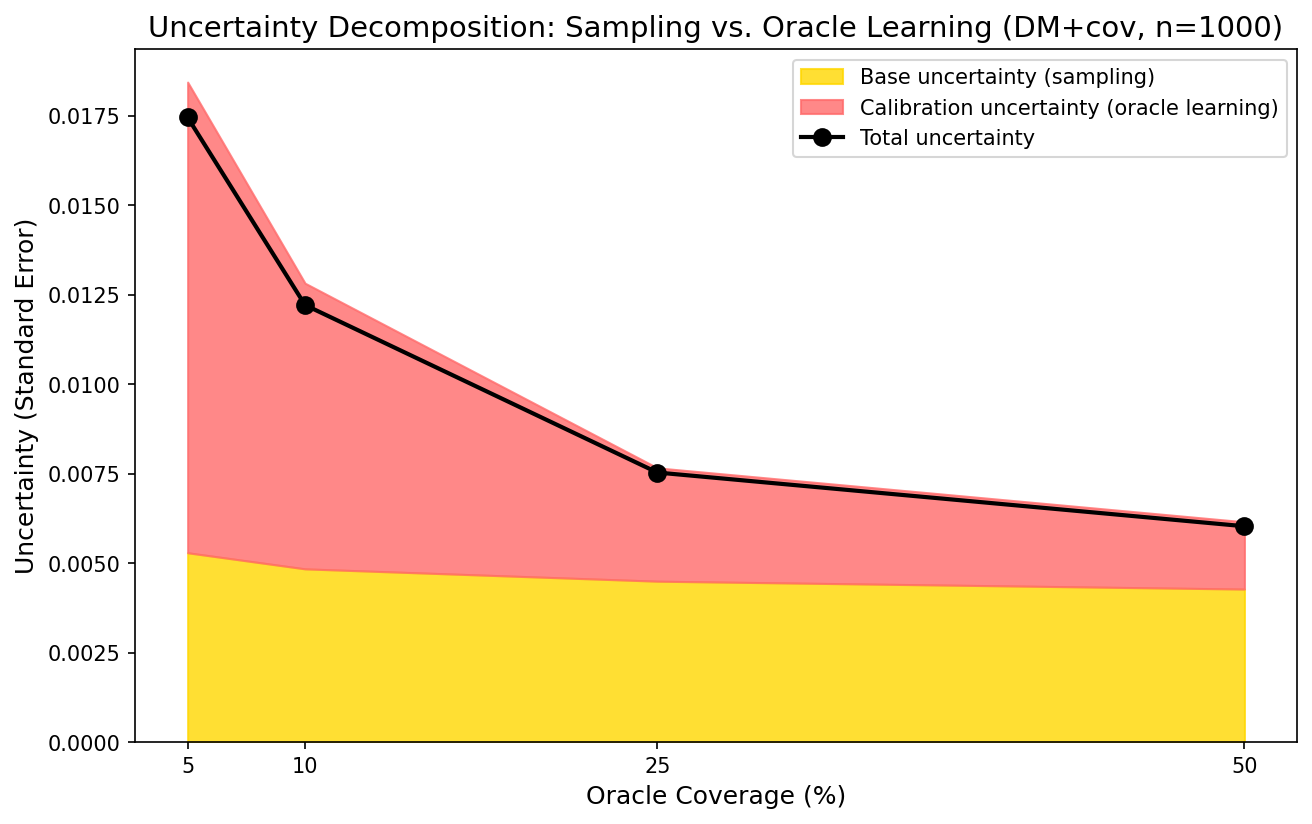}
  \caption{\textbf{Uncertainty decomposition at $n{=}1000$ for \texttt{direct+cov}.}
  Yellow: base sampling uncertainty (constant ${\sim}$0.005). Orange: calibration uncertainty (dominates at low oracle fractions, contributing ${\sim}$90\% of total variance at 5\% oracle; shrinks to ${\sim}$30\% at 100\% because the estimator still fits a calibrator).}
  \label{fig:oua-decomposition}
\end{figure}

\begin{figure}[t]
  \centering
  \includegraphics[width=\columnwidth]{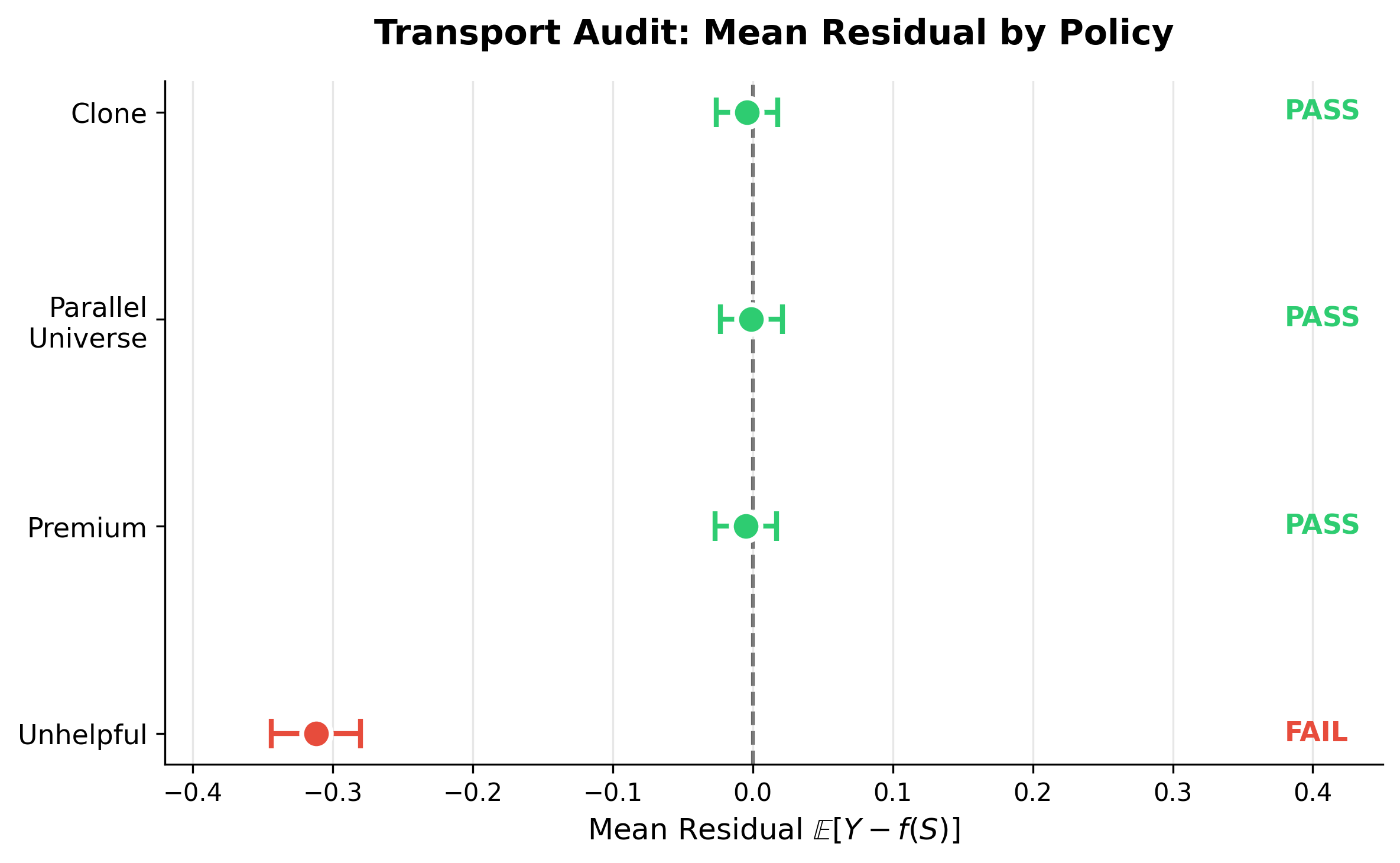}
  \caption{\textbf{Policy-wise mean transport test at 25\% oracle fraction.}
  Testing $H_{0,\pi'}: \E_{\pi'}[Y - f(S, X)] = 0$ per policy (Bonferroni-corrected $\alpha{=}0.0125$).
  Clone, Parallel Universe, and Premium pass (residuals centered at zero).
  Unhelpful fails (mean residual $= -0.31$): the surrogate overestimates oracle scores for adversarial responses by $+0.31$ in absolute level.}
  \label{fig:transportability-test}
\end{figure}

\paragraph{Calibration transportability.}
\Cref{fig:transportability-test} applies the policy-wise mean transport test (\cref{prop:mean-transport}) to validate whether the base-trained calibration can be reused across target policies.
We test $H_{0,\pi'}: \E_{\pi'}[Y - f(S,X)] = 0$ for each policy with Bonferroni correction ($\alpha{=}0.0125$).
Clone, Parallel Universe, and Premium all pass: their residual distributions are centered at zero, implying that Direct + calibrated surrogate remains mean-unbiased for these policies.
Unhelpful fails (mean residual $-0.31$, $p < 0.001$). By \cref{prop:mean-transport}, this implies $V^{\mathrm{oracle}}_{\pi'} - V^{\mathrm{sur}}_{\pi'} \approx -0.31$: the surrogate systematically \emph{overestimates} oracle quality for adversarial responses.
This flags a failure of mean transportability: absolute value estimates for this policy are biased, triggering the \textsc{REFUSE-LEVEL} gate (App.~\ref{app:diag-transport}).
Because Unhelpful's bias ($-0.31$) does not invert its rank relative to other policies (corrected value ${\approx}0.48$ remains well below ${\sim}0.75$), rankings remain valid; only absolute levels are refused.
\emph{Practical implication:} This test decides whether a single calibration can be safely reused across policies; if a policy fails, either recalibrate with policy-specific oracle data or fall back to oracle-only evaluation for that cell.
\emph{Remark (transport correction).} When transport fails, the estimated mean residual $\hat\Delta$ from the audit slice can serve as an intercept correction: $\hat V_{\mathrm{corr}} = \hat V_{\mathrm{sur}} + \hat\Delta$. However, we recommend refusing level claims by default, as the correction only fixes the mean and policy-specific recalibration is preferable when feasible.

\paragraph{Judge-policy confound.}
When a model judges its own outputs, transport can fail due to self-preference bias. On MBPP code generation (\cref{app:mbpp}), GPT-4o-mini judging its own outputs showed a significant mean residual ($-0.065$, $p{=}0.006$). Never judge a model with itself without auditing transport; this failure mode cannot be fixed by calibration alone.

\subsection{Comparison to Misclassification Correction}
\label{ssec:binary}

A natural alternative to continuous calibration is misclassification correction~\citep{RoganGladen1978}, recently adapted for LLM judges by~\citet{Lee2025LLMJudgeReporting}. Both approaches are calibration-based; they differ in what signal they use (continuous vs.\ thresholded) and in how uncertainty propagates. Lee et al.\ assume \emph{confusion-matrix stability}---that $\Pr(\hat Z \mid Z)$ is constant across distributions---analogous to our transport condition.

\paragraph{Setup.} This comparison uses a separate experiment from the main benchmark, designed for an apples-to-apples evaluation with Lee et al.'s binary setting.
We test on 9,999 Arena preference pairs (50.5\% A-wins), the most favorable domain for binary methods: mean Youden's $J = 0.41$ across policies, compared to $J = 0.05$ on code correctness where binary estimation is ill-conditioned. Both methods are evaluated on the same 3 target policies (verbose, terse, contrarian) with identical oracle budgets, test splits, seeds, and judge model (gpt-4o-mini). Confidence intervals follow~\citet{Lee2025LLMJudgeReporting} Eq.~6 with Agresti-Coull adjustment.

\begin{table}[t]
\centering
\small
\caption{\textbf{Misclassification correction vs.\ continuous calibration on Arena preferences} (mean Youden's $J \approx 0.41$). Both methods achieve near-nominal 95\% CI coverage, but CJE produces $9\times$ narrower confidence intervals while achieving 93\% lower RMSE even when Rogan-Gladen uses the optimal threshold.}
\label{tab:binary-comparison}
\begin{tabular}{lccc}
\toprule
Method & RMSE & 95\% CI Coverage & CI Half-Width \\
\midrule
CJE (direct) & $\mathbf{0.004}$ & 96.4\% & $\mathbf{0.015}$ \\
Rogan-Gladen ($\tau{=}0.5$) & $0.063$ & 95.6\% & 0.135 \\
Rogan-Gladen ($\tau{=}0.6$, best) & $0.061$ & 98.9\% & 0.136 \\
\bottomrule
\end{tabular}
\vspace{1mm}
\par\noindent{\footnotesize \emph{Setup:} 9,999 Arena preference pairs (50.5\% A-wins); 3 target policies (verbose, terse, contrarian); identical oracle budgets per cell. CIs follow~\citet{Lee2025LLMJudgeReporting} Eq.~6. Threshold $\tau$ binarizes the continuous judge score; $\tau{=}0.5$ is the natural argmax, $\tau{=}0.6$ is optimal from a sweep over $\{0.3, 0.4, 0.5, 0.6, 0.7\}$.}
\end{table}

\paragraph{Results.} \Cref{tab:binary-comparison} shows both methods achieve near-nominal coverage. However, even with the optimal threshold ($\tau{=}0.6$), CJE achieves 93\% lower RMSE (0.004 vs.\ 0.061) and $9\times$ narrower confidence intervals (half-width 0.015 vs.\ 0.136).

\paragraph{Why the gap?} Two factors explain the efficiency difference:
(1)~\emph{Ill-conditioning}: the Rogan-Gladen estimator $\hat\theta = (\hat p + \widehat{Sp} - 1)/J$ has standard error scaling like $1/J$; at $J \approx 0.44$, small errors in $(\hat p, \widehat{Se}, \widehat{Sp})$ amplify into large errors in $\hat\theta$.
(2)~\emph{Information loss}: binary thresholding discards the continuous probability signal that CJE exploits.
We swept thresholds $\tau \in \{0.3, 0.4, 0.5, 0.6, 0.7\}$ to rule out suboptimal binarization; the best ($\tau{=}0.6$) is shown in the table.

\paragraph{Implication.} Even when misclassification correction is numerically well-conditioned ($J$ well above zero), continuous calibration remains substantially more sample-efficient. Both approaches require transportability diagnostics: Lee et al.'s method relies on confusion-matrix stability; CJE relies on score transport. The key difference is that CJE makes diagnostics \emph{actionable}---we gate the calibrated estimator on transport passing, preventing calibration from damaging estimates when the assumption fails. CJE also generalizes beyond preferences: on MBPP code generation with unit-test ground truth, calibration reduces RMSE by 32--84\% for transport-valid policies at 10\%+ oracle fraction (\cref{app:mbpp}).

\paragraph{Distribution shift robustness.}
\citet{Lee2025LLMJudgeReporting} emphasize that their estimator is robust to shifts in $\Pr(Z)$ between calibration and test sets (their Section~7), since confusion-matrix parameters $\Pr(\hat Z \mid Z)$ do not depend on $\Pr(Z)$.
We tested this claim by varying the calibration-set class balance while holding the test set at its natural rate ($\Pr(Z{=}1) \approx 0.51$).
Results confirm their prediction: Rogan-Gladen bias remains near zero across all shifts (mean $|\mathrm{bias}| < 0.007$), while CJE bias scales linearly with the calibration--test mismatch ($-0.14$ at $\Pr(Z{=}1)_{\mathrm{calib}}{=}0.30$, $+0.14$ at $0.70$).
This is a genuine advantage of confusion-matrix methods: when practitioners cannot control the calibration distribution, misclassification correction provides robustness that continuous calibration lacks.
However, when calibration and test distributions match---the typical experimental design---CJE's efficiency advantage dominates (RMSE $0.009$ vs.\ $0.032$ at no shift).

\paragraph{Sample size planning.}
\Cref{fig:mde-contours} shows MDE contours for the Direct Method: the smallest true policy difference detectable at 80\% power, 5\% two-sided significance.
Key observation: for a fixed label budget $m = n \times \text{oracle fraction}$, larger $n$ with lower oracle fraction outperforms smaller $n$ with higher oracle fraction, suggesting \emph{spread labels across more samples} as a practical rule.
\emph{Practical guidance:} For new domains, oversample the first batch (e.g., $n{=}2{,}000$ with 25\% oracle fraction) to generate a domain-specific MDE plot, then use it to select the minimal viable $(n, \text{oracle fraction})$ for future experiments.

\begin{figure}[t]
  \centering
  \includegraphics[width=\columnwidth]{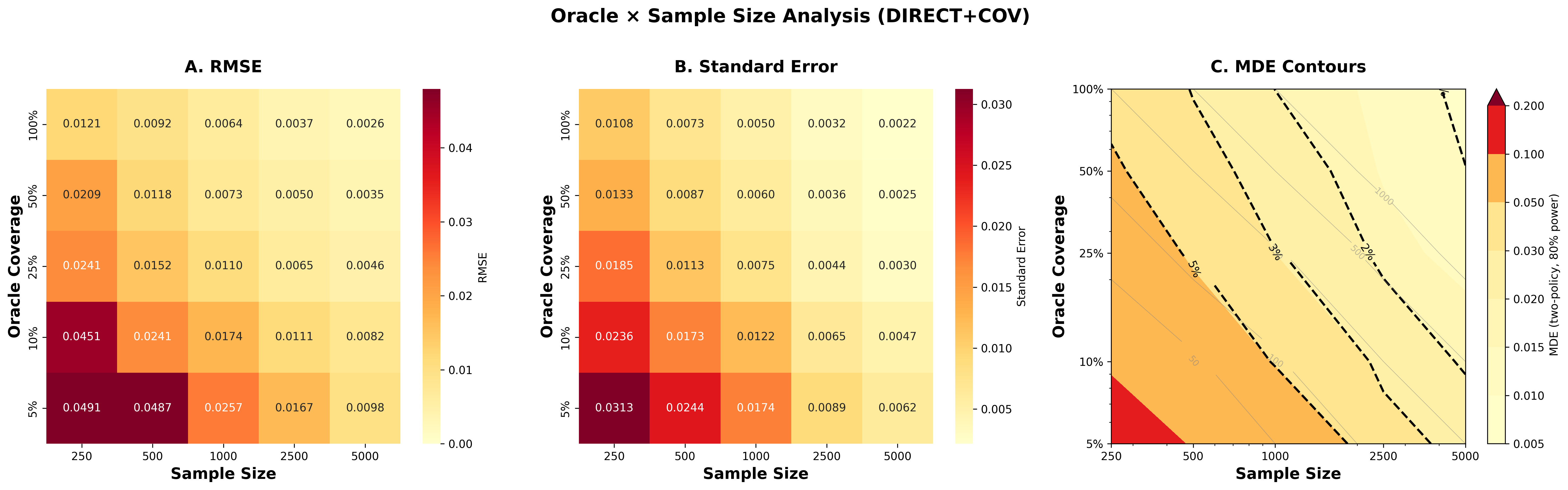}
  \caption{\textbf{Sample size planning for \texttt{direct+cov}.}
  (A) RMSE by sample size and oracle fraction.
  (B) Standard errors including calibration uncertainty.
  (C) MDE contours showing the smallest effect detectable at 80\% power.
  Dashed lines show common detection thresholds; cells with MDE $< 0.02$ enable detecting 2-point quality differences.
  For a fixed label budget, spread labels across more samples (larger $n$, lower oracle \%) rather than concentrating labels on fewer samples.}
  \label{fig:mde-contours}
\end{figure}

\paragraph{Cost analysis.}
\Cref{tab:cost-analysis} quantifies the cost-accuracy tradeoff.
CJE with 5\% oracle calibration achieves 99\% pairwise ranking accuracy at $n{=}5{,}000$ (98.7\% averaged over 50 seeds) at \textbf{14$\times$ lower total cost} than pure oracle labeling (for all 5 policies).
Oracle labels cost approximately $16\times$ more than judge scores, so total cost scales primarily with oracle fraction.
At 5\% oracle fraction, this yields an 8.8$\times$ cost reduction for a single policy.
With amortization across 5 policies, this becomes 14$\times$ (see below), enabling scaling to production workloads where pure oracle labeling is prohibitive.
Amortization improves further: calibrating once on a shared oracle slice enables evaluating multiple policies at marginal judge-only cost, approaching the 16$\times$ cost ratio as $P$ grows.
\emph{Sensitivity:} At smaller $n$ or higher oracle fractions, the cost advantage shrinks; at 25\% oracle fraction the reduction is 3.5$\times$. Ranking accuracy depends on policy gaps and transport validity: the 94\% average across all configurations (\cref{tab:accuracy}) is more representative than the 99\% peak.

\begin{table}[h]
\centering
\small
\caption{\textbf{Cost-accuracy tradeoff} for $n{=}5{,}000$ prompts $\times$ 5 policies (rounded for cost arithmetic; experiments use $n{=}4{,}961$ after filtering). CJE achieves near-oracle accuracy at 14$\times$ lower cost.\textsuperscript{$\dagger$}}
\label{tab:cost-analysis}
\begin{tabular}{lrrrr}
\toprule
Method & Oracle Cost & Judge Cost & Total & Pairwise \% \\
\midrule
Pure oracle (100\%) & \$55.00 & -- & \$55.00 & 100\% \\
CJE (5\% oracle) & \$0.55 & \$3.50 & \textbf{\$4.05} & 99\% \\
\bottomrule
\end{tabular}
\vspace{1mm}
\par\noindent{\footnotesize \textsuperscript{$\dagger$}Costs exclude one-time calibration audits (transportability tests, residual diagnostics). These add ${\sim}$10--20\% overhead initially but amortize across subsequent evaluations on the same domain.}
\end{table}

\begin{table}[h]
\centering
\small
\caption{\textbf{Estimator Selection Guide.} Requirements and recommended use cases for each estimator family.}
\label{tab:estimator-guide}
\begin{tabular}{@{}lcccp{4cm}@{}}
\toprule
\textbf{Estimator} & \textbf{Fresh} & \textbf{Logprobs} & \textbf{Overlap} & \textbf{When to Use} \\
\midrule
Direct & \checkmark & -- & -- & \textbf{Default} for LLM eval \\
IPS/SNIPS & -- & \checkmark & Required & Small policy changes only \\
Calibrated IPS & -- & \checkmark & Required & + weight stabilization \\
DR & \checkmark & \checkmark & Helpful & Logs + fresh draws available \\
Stacked-DR & \checkmark & \checkmark & Helpful & Production ensemble \\
\bottomrule
\end{tabular}
\end{table}

\subsection{Practical takeaways}
\label{ssec:takeaways}

\begin{enumerate}[leftmargin=*,itemsep=2pt,topsep=2pt]
\item Default to Direct + two-stage calibration. For open-ended generation under nontrivial policy shifts, logs-only OPE faces two compounding issues: TF brittleness and coverage sparsity. Direct requires no overlap assumption and achieves the best ranking accuracy in our benchmark (\cref{tab:estimator-guide}).
\item Audit transport for high-stakes decisions. When reporting absolute values or making deployment decisions, run the mean transport test on a small oracle slice per target policy (or time window). Pass $\to$ reuse calibration; fail $\to$ recalibrate with policy-specific data or refuse level claims for that cell. For exploratory analysis or ranking-only comparisons, the audit can be deferred.
\item Gate on TTC before relying on logged data. If TTC $< 0.70$, CLE implies a precision floor that weight stabilization cannot overcome; in our benchmark, only near-clone policies pass this threshold. For constrained settings (short outputs, near-clone shifts), OPE may remain viable; verify with diagnostics before trusting logged-data estimates.
\item Always use bootstrap inference with $\hat\theta_{\mathrm{aug}}$. Standard CIs severely under-cover (15--52\%); calibration-aware jackknife improves but still falls short (70--89\%); bootstrap with bias correction achieves near-nominal coverage (${\sim}$95\%).
\item Use calibration uncertainty share to guide resource allocation. The fraction $\widehat{\mathrm{Var}}_{\mathrm{cal}} / \widehat{\mathrm{Var}}_{\mathrm{total}}$ identifies bottlenecks: if calibration uncertainty share $> 50\%$, collect more oracle labels; if calibration uncertainty share $< 20\%$, collect more evaluation prompts (\cref{fig:oua-decomposition}).
\item Covariates matter. Response length as a calibration covariate (not a reweighting covariate) improves ranking across all methods.
\item Optimize the budget ratio. How many oracle labels do you need? The optimal ratio of labels ($m$) to scores ($n$) follows a square-root law based on costs and variances (\cref{app:budget}). Check the calibration uncertainty share $\omega$ to equalize marginal utility: if $\frac{\omega}{1-\omega} > \frac{\mathrm{Spend}_{\mathrm{oracle}}}{\mathrm{Spend}_{\mathrm{surrogate}}}$, the experiment is under-investing in oracle labels.
\end{enumerate}

\section{Limitations}
\label{sec:limitations}

\paragraph{Oracle alignment (scope).}
We assume the operational oracle $Y$ aligns with stakeholder values; oracle selection itself is a governance question outside the scope of this work (Assumption~A0 in App.~\ref{app:assumptions-ledger}).

\paragraph{Overlap (positivity).}
IPS requires support overlap between $\pi_0$ and each $\pi'$; DR benefits from overlap but remains consistent via its outcome model when overlap is poor; Direct requires no overlap (it generates fresh responses under $\pi'$).
When overlap is poor, raw ratios are heavy-tailed and uncertainty inflates.
\emph{Mitigations:} Weight stabilization reduces dispersion and raises $\ESS$; if tails persist we (i) gate on $\ESS$ and Hill indices, (ii) use overlap weighting or cohort restriction, and (iii) run an online check when $\widehat{\alpha}_{\text{Hill}}{<}1$ or single–row dominance persists (App.~\ref{app:diagnostics}).

\paragraph{Judge assumptions (surrogate validity).}
Our calibration model uses a monotone single-index restriction (J2-M/J2-SI) as a stability/regularization choice. Correctness of level claims across policies is governed by \emph{mean transport} ($\E_{\pi'}[Y - f(S,X)] = 0$), which we audit; when the audit fails or evidence is weak, we widen/refresh the oracle slice or refuse level claims.
\emph{Mitigations:} (i) \textbf{With audit}: run the mean transport test per policy (\cref{prop:mean-transport}); fail $\to$ recalibrate or refuse levels. (ii) \textbf{Without audit}: surface reliability curves as indirect evidence; when evidence is weak, target labels where error concentrates.

\paragraph{Calibration failure modes.}
Reward calibration can fail in three distinct ways, each with observable signatures:
\begin{enumerate}[leftmargin=*,nosep]
\item \textbf{Policy-specific calibration:} The $S \to Y$ mapping differs across policies.
\emph{Signature:} Within-policy rankings accurate, cross-policy comparisons fail.
\emph{Test:} Policy-wise mean transport test (\cref{prop:mean-transport}). If $\E_{\pi'}[Y - f(S,X)] \neq 0$, surrogate values are biased for that policy.
\item \textbf{Extrapolation beyond range:} Target policy produces $S$ values outside the calibration range.
\emph{Signature:} Boundary flatness; \texttt{OutOfRange} $> 5\%$.
\emph{Test:} Coverage diagnostic (App.~\ref{app:diagnostics}).
\item \textbf{Temporal drift:} The $S \to Y$ mapping changes over time.
\emph{Signature:} Recent residuals larger than historical; calibration accuracy degrades.
\emph{Test:} Periodic mean transport test on fresh oracle batches.
\end{enumerate}

\paragraph{Scope of the mean transport test.}
The policy-wise mean transport test (\cref{prop:mean-transport}) is a \emph{necessary but not sufficient} condition for full surrogate transportability (S2): it guarantees that the surrogate is unbiased for policy-wise mean values, but does not rule out conditional miscalibration within subgroups or at the tails of the score distribution.
\emph{Mitigations:} perform subgroup-specific transport tests when fairness concerns apply; monitor reliability diagrams for regional deviations; use richer moment conditions for stronger tests (see remark after \cref{prop:mean-transport}).

\paragraph{Calibration coverage (identification).}
If a $\pi'$ pushes $S$ outside the labeled range, isotonic calibration flattens at the boundary and levels are not point–identified.
\emph{Mitigations:} flag \textsc{Limited Calibration Support} and set \RefuseLevel\ (report rankings and partial–ID bounds) until targeted labels cover the uncovered $S$ region (App.~\ref{app:diagnostics}).

\paragraph{Approximate sufficiency (bias modulus).}
When $\E[Y\mid X,A,S]\neq \mu(S)$, the residual $\Delta(X,A,S)$ induces bias proportional to calibration error.
\emph{Mitigations:} by Cauchy–Schwarz, $|\mathrm{Bias}|\le \|m-W_{\pi'}\|_2\,\|\Delta\|_2$; weight stabilization shrinks $\|m-W_{\pi'}\|_2$, so bias is second order when either calibration is tight or the violation small; we surface this via diagnostics and invoke \RefuseLevel\ when unbounded.

\paragraph{Temporal dependence and logger drift.}
Non–stationarity (launches, safety updates) can bias or widen intervals.
\emph{Mitigations:} report dependence–robust SEs (block/stationary bootstrap), shorten analysis windows, and monitor judge drift via rank–based/residual change detection with FDR control.

\paragraph{Oracle independence and leakage.}
Calibration-aware inference assumes the oracle slice is i.i.d.\ with non-overlapping folds.
\emph{Mitigations:} reuse deterministic folds across modules, de–duplicate the slice, and periodically refresh it.

\paragraph{Selection and multiplicity.}
Scanning many $\pi'$ inflates winner's curse.
\emph{Mitigations:} use FDR control (BH/BY), optionally an outer split for influence-function stacking to reduce selection optimism, and emphasize pre–specified contrasts.

\paragraph{Teacher forcing and API drift.}
Accurate propensities require deterministic, chat–native TF (stable tokenizer/template) with additivity/conditionality invariants; missing/invalid TF corrupts ratios.
In our experiments, 39 of 5,000 prompts (0.8\%) were filtered: 11 due to base-policy TF conformance failures, and 28 due to missing target-policy logprobs. While small, we report this explicitly: practitioners should expect some data loss from TF gaps, and filtering should always be logged for reproducibility. This also highlights the current state of TF support from LLM providers: major closed-source providers (OpenAI, Anthropic, Google) do not currently support teacher forcing for chat models, and among open-weight hosting providers, Fireworks AI was the only service we found with reliable TF support.
Even for our TF reliability test (clone = base with different seed), raw ESS is 26\% rather than the theoretical 100\%, reflecting API-level non-determinism and floating-point divergence in log-probability summation; weight stabilization restores this to 99\% (\cref{tab:weight-diagnostics}). This 26\% is substantially higher than other policies (0.4--0.7\%), confirming TF captures real distributional similarity, but it is far from the theoretical 100\% and suggests propensity estimates carry meaningful error even before overlap considerations.
OPE failure in our experiments reflects two compounding issues: (i) TF brittleness corrupts propensity estimates, and (ii) CLE limits precision even with perfect propensities. We cannot cleanly separate their contributions, but both matter: if TF were the only problem, weight stabilization should improve ranking more than it does; if CLE were the only problem, clone should have ESS $\approx$ 100\%.
\emph{Mitigations:} enforce schema/conformance checks, ledger failures, and treat results as conditional on TF quality (App.~\ref{app:impl}).

\paragraph{Subgroups and fairness.}
Calibration quality and ESS gains may differ across subgroups.
\emph{Mitigations:} provide subgroup diagnostics (ESS, reliability) and, when feasible, use subgroup–specific calibration/weights or constrained pooling.

\paragraph{Judge informativeness (garbling).}
Coarser judges raise the surrogate information bound and widen CIs.
\emph{Mitigations:} prefer richer rubrics (multi–dimensional $S$ with stable aggregation) and validate with coarsening ablations (empirical Blackwell monotonicity).

\paragraph{Compute.}
Calibrated DR adds one rollout\,+\,judge per $(X,\pi')$; reward calibration refits for calibration-aware inference add modest overhead.
We amortize via TF caches, shared folds, and a small stacking library.

\noindent\textbf{Ethics Statement.}
We analyze retrospective logs that may include sensitive content.
Diagnostics/gates prevent overconfident claims under poor overlap/coverage and surface judge drift.
When identification fails, we report rankings only (\RefuseLevel) and recommend targeted labeling or online checks.
Any deployment should assess subgroup reliability and adopt privacy safeguards for logs.

\noindent\textbf{Reproducibility Statement.}
Exact prompts, schema, TF contract, pseudocode, and numerics appear in the appendices.
Code and data are available at \url{https://github.com/cimo-labs/cje-arena-experiments}; the CJE library is available via \texttt{pip install cje-eval} (\url{https://github.com/cimo-labs/cje}).

\section{Conclusion}
\label{sec:conclusion}

Uncalibrated LLM-as-judge evaluation is the default practice, and it produces rankings that can be inverted, confidence intervals with near-zero coverage, and OPE estimates that fail silently. These are not edge cases; they are the norm when calibration is ignored.

We introduced \textbf{\CJE}, a framework that makes it affordable to aim at the right target. By calibrating cheap judges against a small oracle slice, teams no longer need to choose between expensive labels and proxies that may point the wrong direction. CJE fixes all three failures, built around a single principle: encode justified assumptions as model restrictions. Restrictions to \emph{subspaces} (nuisance-orthogonal scores), \emph{monotone cones} (mean-preserving reward/weight calibration), and \emph{simplices} (variance-hedged stacking) preserve the estimand while \emph{weakly reducing variance}.

Concretely, \emph{reward calibration} learns a mean-preserving surrogate $R=f(Z(S))$; \emph{weight stabilization} stacks $S$-monotone candidates to produce unit-mean ratios with empirical $\ESS$ gains via OOF stacking and a variance guard; \emph{calibrated DR} delivers $\sqrt{n}$ inference under cross-fitting; \emph{influence-function stacking} minimizes plug-in IF variance; and \emph{calibration-aware inference} with bootstrap and bias-corrected estimation ($\hat\theta_{\mathrm{aug}}$) achieves near-nominal CI coverage.

Theoretically, the \emph{Efficiency via Model Restriction} result, applying established semiparametric principles to surrogate-based evaluation, explains why encoding justified knowledge as model restrictions lowers the efficiency bound, which is \emph{attainable} with projection-designed estimators and cross-fitting. Two design corollaries follow: (\emph{i}) \emph{Blackwell–efficiency monotonicity}: richer judges (finer $\sigma$-fields) strictly help; (\emph{ii}) under the exact $\mathrm{IsoMeanOne}_S$ projection, isotonic mean-one calibration weakly reduces weight dispersion (the reference implementation uses a simpler normalization that empirically yields similar $\ESS$ gains). Empirically, on Arena-derived logs, weight stabilization turns near-degenerate ratios into stable weights (large $\ESS$ gains), stacked-DR achieves near-nominal coverage and near-$\sqrt{n}$ scaling, and stacking improves ordering; when calibration support is limited, \CJE{} \emph{flags} the issue and reports robust \emph{rankings} with conservative uncertainty (\RefuseLevel).

\paragraph{Takeaways.}
(i) \emph{Default to Direct}: For open-ended generation under nontrivial policy shifts, logs-only OPE faces two compounding issues---TF brittleness and coverage sparsity (low TTC)---that together create a precision floor weight stabilization cannot overcome. In constrained settings or near-clone regimes, OPE may remain viable; gate on TTC $\geq 0.70$ before relying on logged data.\\
(ii) \emph{Always use bootstrap inference with $\hat\theta_{\mathrm{aug}}$}: naive CIs achieve 0\% coverage; bootstrap with bias correction achieves ${\sim}$95\%.\\
(iii) \emph{Validate transportability}: the mean transport test catches real failures (e.g., adversarial policies).

\paragraph{Future work.}
Selection-aware inference over large policy sets; robust/DP isotonic calibration (mirror/Bregman projections) for heavy tails; active oracle budgeting via shadow prices; sequential/agent evaluations with prefix-aware weight stabilization and stepwise DR; and subgroup-aware constraints with fairness diagnostics.

\section{Related Work}
\label{sec:related}

\CJE\ draws on four research threads: surrogate-assisted causal inference, LLM evaluation, probability calibration, and off-policy evaluation. We position our contributions by first clarifying three regimes of surrogacy that organize the literature.

\paragraph{Three regimes of surrogacy.}
Surrogate-assisted estimation admits three regimes with increasing assumptions and decreasing oracle burden:

\begin{enumerate}[leftmargin=*,itemsep=2pt,topsep=2pt]
\item Surrogates for efficiency only \citep{KallusMao2020Surrogates}: The surrogate $S$ improves estimation efficiency (e.g., as features in outcome models) but does not identify treatment effects. Oracle labels $Y$ are required in every evaluation context. This is the fallback when stronger surrogacy assumptions fail.

\item Local surrogacy \citep{AtheyChettyImbensKang2019,Prentice1989Surrogate}: Surrogate sufficiency (S1) holds within a single environment $g^\star$: $\E[Y \mid X, A, S] = f(S, X)$ in $g^\star$. Calibrate once per environment; evaluate many policies within that environment using only $S$.

\item Reusable calibration with audit (CJE): Attempt to reuse calibration $f$ across environments, but budget oracle labels in each target context to audit transport rather than assume it. When the mean transport test fails, flag the policy/context and either recalibrate or fall back to oracle-only evaluation.
\end{enumerate}

\noindent CJE targets Regime~3 subject to audit. Unlike Kallus \& Mao (surrogates for efficiency only) and Athey et al.\ (surrogate index in a single environment), CJE attempts cross-environment reuse, but treats transportability as auditable rather than assumed. The benefit is amortization when transport holds; the cost is budgeting a small oracle slice per target context. Our diagnostics (mean residual test, TTC, Bhattacharyya affinity) detect when transport fails, making CJE an evaluation protocol with explicit failure detection rather than a brittle identification claim.
The key difference from Athey et al.: in their two-sample setting, $Y$ is often missing in the experimental arm, so surrogacy and comparability must be \emph{assumed}. CJE explicitly budgets oracle labels in each target environment, allowing us to \emph{audit} (and if needed, correct) the transport gap $\E_{\pi'}[\varepsilon]$ rather than assume it away (\cref{fig:dag}, \cref{eq:transport-decomp}).

\paragraph{Auditable assumptions via negative controls and proxies.}
CJE's approach to transportability draws on a broader principle in causal inference: turning untestable identification conditions into checkable diagnostics using auxiliary data.
Negative controls provide a canonical example: variables known to have no causal relationship with treatment or outcome can detect residual confounding when associations appear \citep{Lipsitch2010}.
Proximal causal inference extends this idea, using proxy variables to construct testable restrictions that support identification under weaker assumptions \citep{TchetgenTchetgen2024}.
CJE applies a similar philosophy: rather than assuming global surrogacy, we budget oracle labels in each target context to test a \emph{policy-wise moment restriction} ($\E_{\pi'}[Y - f(S,X)] = 0$).
The contribution is operational: we provide an evaluation protocol with explicit pass/fail criteria and decision rules (reuse calibration vs.\ recalibrate vs.\ refuse level claims), making the assumption \emph{auditable} rather than hoping it holds.
CJE is not a negative-control estimator; the parallel is conceptual and design-oriented, not methodological.

\paragraph{LLM-as-judge evaluation.}
Automatic judges (LLMs scoring other LLMs) are now standard for scalable evaluation \citep{Zheng2023LLMasJudge,Kim2024PrometheusJudge,KocmiFedermann2023}.
Benchmarks like AlpacaEval \citep{DuboisLiTaoriEtAl2024AlpacaEval} and RewardBench \citep{LambertPiechChenEtAl2024RewardBench} systematically compare judge quality, revealing biases (verbosity, position, self-preference) and drift under distribution shift \citep{DietzEtAl2025LLMJudges}.
These works focus on \emph{ranking} judges by correlation with human labels; \CJE\ addresses the orthogonal problem of \emph{calibrating} a fixed judge to produce unbiased point estimates and calibration-aware confidence intervals for policy value.
Recent work like CalibraEval \citep{Li2025CalibraEval} addresses position and token bias via label-free calibration; \CJE\ is complementary, providing oracle-grounded calibration with transportability guarantees and uncertainty quantification.

\paragraph{Misclassification correction.}
\citet{RoganGladen1978} introduced prevalence estimation under measurement error; \citet{Lee2025LLMJudgeReporting} adapt this for LLM judges, providing confidence intervals that account for calibration uncertainty. Their approach assumes \emph{confusion-matrix stability}---that sensitivity and specificity transfer from calibration to test---analogous to our transport condition. We compare directly in \cref{ssec:binary} and find continuous calibration $13\times$ more sample-efficient, even when binary correction is numerically well-conditioned (mean $J = 0.41$). The efficiency gap arises from ill-conditioning (the Rogan-Gladen estimator's standard error scales like $1/J$) and information loss from binary thresholding.

\paragraph{EIF-based debiasing.}
Concurrent work by \citet{Chen2026EfficientLLMJudge} formalizes LLM-as-judge debiasing through efficient influence functions (EIF), unifying Rogan--Gladen-style measurement error correction with PPI/PPI++-style calibration. They derive the EIF for mean estimation with a surrogate and labeled subset, showing that naive judge-based CIs can achieve 0\% coverage under bias while EIF-based estimators maintain nominal coverage with substantially narrower intervals than confusion-matrix methods. Their framing---learning $\mu(\hat Y) = \E[Y \mid \hat Y]$ on a small calibration set and using residual augmentation---is closely aligned with CJE's reward calibration. CJE extends this to continuous scores, policy/deployment shift (via the mean transport test), and off-policy evaluation with weight stabilization.

\paragraph{Probability calibration.}
Platt scaling \citep{Platt1999}, isotonic regression \citep{ZadroznyElkan2002,NiculescuMizilCaruana2005}, and temperature scaling \citep{GuoPleissSunWeinberger2017Calibration} are standard post-hoc calibration methods for classification.
Our reward calibration applies isotonic regression not for classification confidence but for \emph{reward calibration}: mapping a continuous judge score $S$ to an expected oracle outcome $\E[Y \mid S]$.
The mean-preserving property (Lemma~\ref{lem:iso-mean}) ensures the expected value of calibrated rewards matches the oracle under the logging distribution.

\paragraph{Off-policy evaluation (OPE).}
IPS \citep{HorvitzThompson1952,Hajek1964Rejective} and doubly robust estimators \citep{BangRobins2005,JiangLi2016,KallusUehara2020DRL} are foundational.
Weight stabilization via clipping \citep{Ionides2008TruncatedIS}, balancing \citep{Kallus2018Balanced}, and overlap weighting \citep{LiMorganZaslavsky2018,Crump2009LimitedOverlap} address variance inflation.
Most relevant is recent work on isotonic-calibrated IPW \citep{vanDerLaanLinCaroneLuedtke2025StabIPW} and DR inference via calibration \citep{vanDerLaanLuedtkeCarone2024DRCalibration}.
Our weight stabilization extends this by projecting importance weights onto the cone of \emph{judge-score-monotone} functions, exploiting surrogate structure absent in generic OPE.

\paragraph{Ensemble methods for estimation.}
Super Learner \citep{vanderLaan2007SuperLearner} provides a powerful framework for combining candidate algorithms by minimizing cross-validated prediction risk.
Recent work extends ensemble ideas to semiparametric estimation \citep{ShinLiuColeFine2020EnsembleSemiparametric}, combining efficient estimators derived from factored likelihoods.
Our influence-function stacking (\cref{thm:stack}) takes a complementary approach: instead of combining predictions to minimize MSE, we combine estimators by minimizing the asymptotic variance of their combined influence function.
This provides a general formulation that directly targets efficiency, is not tied to specific likelihood factorizations, and is explicitly framed as optimization in the Hilbert space of influence functions.

\paragraph{Transportability.}
Generalizing from source to target populations is formalized by the selection diagram framework \citep{BareinboimPearl2013Algorithm,BareinboimPearl2013Meta,PearlBareinboim2014}, which uses selection variables $\mathbf{S}$ to encode mechanisms that differ between domains.
A quantity is transportable when selection nodes can be ``blocked'' by conditioning (graphically, when $\mathbf{S} \perp\!\!\perp Y \mid (X, Z)$ in the manipulated graph).
Our transport assumption instantiates this: calibration transports when no selection mechanism points into $Y$ given $(X, A, S)$.
The dashed red edges in \cref{fig:dag} represent potential transport failures: policy shift ($\pi \to Y$) or environment shift ($E \to Y$) can break the $S$-$Y$ calibration, and the transport test detects when this occurs.
Covariate shift methods \citep{Shimodaira2000CovShift,SugiyamaKrauledatMueller2007IWCV} address the special case where selection affects only covariates ($E \to X$), not outcomes directly.
The CLE bound (\cref{sec:background}) formalizes a complementary failure mode: even with valid transportability, logs-only estimation fails under poor coverage, motivating Direct or DR methods.

\section*{Acknowledgements}

We thank Kevin Zielnicki, Molly Davies, Sven Schmit, Brad Klingenberg, Chris Rinaldi, Viridiana Lourdes, Izzy Farley, Sudhish Kasaba Ramesh, Sean Kelly, Adith Swaminathan, Jasmine Nettiksimmons, Winston Chou, and Tom Cunningham for helpful discussions and feedback on earlier drafts of this work.

\bibliography{references}
\bibliographystyle{iclr2026_conference}

\appendix
\section{Notation and Formal Setup}
\label{app:notation}

\paragraph{Observed data and policies.}
We observe i.i.d.\ logs
\[
O_i=(X_i,A_i,S_i,Y^{\mathrm{obs}}_i,L_i),\qquad i=1,\dots,n,
\]
generated under a fixed logger $A_i\sim \pzero(\cdot\mid X_i)$. A scalar judge
$S_i=s(X_i,A_i)$ is available on \emph{all} rows. The label indicator $L_i\in\{0,1\}$ marks inclusion in the oracle slice; when $L_i=1$ we observe $Y^{\mathrm{obs}}_i=Y_i$, otherwise $Y^{\mathrm{obs}}_i$ is missing.
For a candidate policy $\pprime$, define the sequence-level importance ratio
\[
W_{\pprime,i}
\;=\;
\frac{\pprime(A_i\mid X_i)}{\pzero(A_i\mid X_i)}
\;=\;
\exp\!\Big\{\log p_{\pprime}(A_i\!\mid\!X_i)-\log p_{\pzero}(A_i\!\mid\!X_i)\Big\},
\]
computed via teacher forcing (TF) with the model’s own tokenizer/rendering. Write
\[
W^{\mathrm{m1}}_{\pprime,i}
\;=\;
\frac{W_{\pprime,i}}{\tfrac{1}{n}\sum_{j=1}^n W_{\pprime,j}}
\]
for the sample–mean–one (SNIPS) baseline (global normalization over the evaluation cohort).

\paragraph{Estimand.}
Let $Y(\pprime)$ denote the outcome under the counterfactual draw $A\sim \pprime(\cdot\mid X)$.
The target is
\[
V(\pprime)=\E\!\big[Y(\pprime)\big].
\]

\subsection{Assumptions (compact)}
\label{app:assumptions}
\textbf{(D1) Fixed logger \& i.i.d.} $(X_i,A_i,S_i)$ are i.i.d.\ under $\pzero$; TF log-likelihoods are stable and well-defined.

\noindent\textbf{(D2) Overlap (positivity).} $\pzero(a\mid x)>0$ whenever $\pprime(a\mid x)>0$, and $\E_{\pzero}[W_{\pprime}^2]<\infty$.

\noindent\textbf{(D3) Judge coverage \& stability.} $S$ is well-defined under both $\pzero$ and $\pprime$; the Radon–Nikodym derivative on $\sigma(S)$ exists; the judge/rubric is stable on the analysis window.

\noindent\textbf{(J1) Oracle slice.} There exists an i.i.d.\ subsample $O=\{i:L_i=1\}$ with $m=|O|\ll n$ on which $Y$ is observed.

\noindent\textbf{(J2-M) Mean sufficiency (monotone).} $\E[Y\mid X,A,S]=\mu(S)$ with $\mu$ weakly nondecreasing. This structural assumption enables the EIF derivations (\cref{thm:sur-eif}); it is \emph{sufficient} for transport (T1) but not required when transport is verified via audit.\quad
\textbf{(J2-SI) Single-index fallback.} There exist $g^\star:\mathcal{S}\times\mathcal{X}\!\to\!\R$ and nondecreasing $\mu^\star$ such that $\E[Y\mid X,A,S]=\mu^\star\!\big(g^\star(S,X)\big)$. When $g^\star$ depends only on $S$, this reduces to J2-M; when covariates $X$ (e.g., response length) are included in $g$, the target remains $V(\pi')$ and the estimand is preserved provided the index is consistently estimated.

\noindent\textbf{(T1) Transport.} $\E_{\pi'}[Y - f(S,X)] = 0$. Calibration learned under $\pi_0$ is mean-unbiased under $\pi'$. This is the \emph{fundamental requirement} for valid surrogate-only policy value estimation. With an audit slice under $\pi'$, T1 is testable (\cref{prop:mean-transport}); without audit, T1 must be assumed (e.g., via J2-M/J2-SI or domain knowledge). See the assumption hierarchy in \cref{app:assumptions-ledger}.

\noindent\textbf{(R1) Tails/moments.} $\E[Y^2]<\infty$, $\E[S^2]<\infty$, and $\E_{\pzero}[W_{\pprime}^2]<\infty$. Under the exact $\mathrm{IsoMeanOne}_S$ projection with relative variance bound $\rho\ge1$, $\Var_n(\hat W_{\pprime})\le \rho\,\Var_n(W^{\mathrm{m1}}_{\pprime})$ (i.e., variance at most $\rho$ times the SNIPS baseline).

\noindent\textbf{(R2) Calibration consistency.} Reward calibration satisfies $\|\hat f(Z)-\E[Y\mid Z]\|_{L^2(P_Z)}=o_p(1)$, where $Z=S$ in monotone mode and $Z=g(S,X)$ (learned index) in two-stage mode. Weight stabilization satisfies $\|\hat W_{\pprime}-m^\star_\rho\|_{L^2(P)}=o_p(1)$, where $m^\star_\rho$ is the $S$-monotone, mean-one, variance-capped projection of $W_{\pi'}$ (see \cref{app:projections}); when $\E[W_{\pi'}\mid S]$ is itself monotone and satisfies the variance constraint, $m^\star_\rho=\E[W_{\pi'}\mid S]$.

\noindent\textbf{(R3) One-of-two rates with cross-fitting.} With nuisances $\hat q(x,a)\approx \E[R\mid x,a]$ and $\hat W_{\pprime}$,
\[
\|\hat q-q^\star_R\|_{L^2(P)}\cdot \|\hat W_{\pprime}-m^\star\|_{L^2(P)} \;=\; o_p(n^{-1/2}),
\]
e.g., either factor is $o_p(n^{-1/4})$ with the other consistent.

\subsection{Cross-fitting and folds}
\label{app:folds}
Let $F:\{1,\dots,n\}\!\to\!\{1,\dots,K\}$ be a deterministic fold map (e.g., a hash of $x\_id$). For any learner $\mathcal{L}$,
train $\hat\eta^{(-k)}=\mathcal{L}$ on $\{i:F(i)\neq k\}$ and use out-of-fold predictions
$\hat\eta^{\mathrm{OOF}}_i=\hat\eta^{(-F(i))}(O_i)$ in influence-function (IF) calculations. The same folds are reused across reward calibration, weight stabilization, and DR nuisances.

\subsection{Projection operators used by \CJE}
\label{app:projections}
\textbf{Monotone cone.} $\mathcal{M}_\uparrow=\{f:\R\!\to\!\R \text{ nondecreasing}\}$. The isotonic projector (PAVA) $\Pi_{\mathcal{M}_\uparrow}$ enjoys:
(i) $L^2$ optimality; (ii) \emph{mean preservation} on the training sample; (iii) \emph{dispersion reduction} by majorization.

\noindent\textbf{Mean-one cone for weights.} For $w\in\R^n$ ordered by $S$, define the mean-one isotonic projection
\[
\mathrm{IsoMeanOne}_S(w)
\;=\;
\argmin_{u}\ \sum_i (u_i-w_i)^2
\quad\text{s.t.}\quad
u\in\mathcal{M}_\uparrow(S),\ \tfrac{1}{n}\sum_i u_i=1,
\]
where $\mathcal{M}_\uparrow(S)$ denotes vectors nondecreasing in the $S$-order. This preserves the sample mean and weakly reduces empirical variance; hence $\ESS$ weakly increases (deterministically, by majorization).

\noindent\textbf{Simplex hull.} For centered IF columns $\{\phi^{(e)}\}_{e\in\mathcal{E}}$, let $\Phi=[\phi^{(e)}]$ and $\Delta=\{\alpha:\alpha_e\ge0,\ \sum_e\alpha_e=1\}$. Influence-function stacking solves
$\min_{\alpha\in\Delta}\alpha^\top \hat\Sigma\,\alpha$ with $\hat\Sigma=(1/n)\Phi^\top\Phi+\lambda I$.

\subsection{Reward calibration and weight stabilization primitives}
\label{app:ops}
\textbf{Reward calibration.} On $O=\{i:L_i=1\}$, fit $R=\hat f(Z(S))$ via:
(i) \emph{two-stage} (recommended): $Z(S,X)=\mathrm{ECDF}\{g(S,X)\}$ with a spline $g$, or
(ii) \emph{monotone-only}: $Z(S)=S$ when covariates are unavailable.
Use OOF predictions $R^{\mathrm{OOF}}$ along IF paths; the point estimate may use a pooled fit. A terminal isotonic step enforces \emph{slice-mean preservation}.

\begin{figure}[t]
  \centering
  \includegraphics[width=\textwidth]{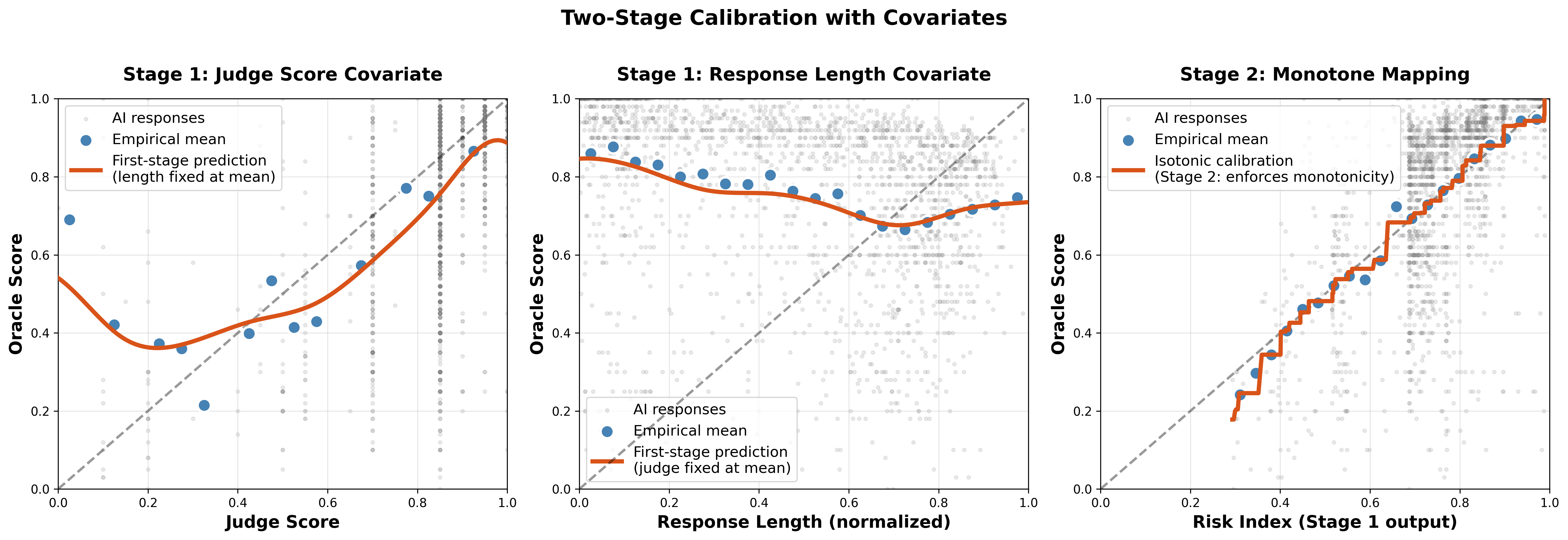}
  \caption{\textbf{Two-stage calibration with covariates.} Left: first-stage spline on judge score (response length fixed at mean). Center: first-stage effect of response length (judge score fixed). Right: second-stage isotonic mapping on the risk index (Stage 1 output) enforces monotonicity while preserving flexibility from the first stage.}
  \label{fig:two-stage-detail}
\end{figure}

\noindent\textbf{Weight stabilization.} (Per fold) Fit up/down isotonic maps on $S$ to obtain $W^{\mathrm{OOF}}_{\uparrow},W^{\mathrm{OOF}}_{\downarrow}$; optionally include $W^{\mathrm{OOF}}_{\mathrm{base}}\equiv 1$. Define residuals $\Delta_i$ (IPS: $R_i$; DR: $R_i-\hat q(X_i,A_i)$). Choose
$\hat\beta\in\arg\min_{\beta\in\Delta_3}\beta^\top \hat\Sigma\,\beta$ where $\hat\Sigma$ is the covariance of $U_c=W^{\mathrm{OOF}}_c \Delta$. Form $W^{\mathrm{stack}}=\sum_c \hat\beta_c W^{\mathrm{OOF}}_c$, renormalize to mean one (global, over the evaluation cohort), optionally apply the \emph{absolute variance guard}
\[
\alpha=\min\!\Big\{1,\ \sqrt{\frac{\rho}{\Var(W^{\mathrm{stack}})}}\Big\},\qquad
W^{\mathrm{blend}}=1+\alpha\,(W^{\mathrm{stack}}-1),
\]
and normalize to mean one: $\hat W_{\pprime}=W^{\mathrm{blend}}/\bar W^{\mathrm{blend}}$. (The theoretical Lemma~\ref{lem:guard} uses a relative bound.)

\subsection{DR nuisances and sequence value}
\label{app:nuisances}
Let $\hat q(x,a)\approx \E[R\mid x,a]$ and define $\hat g_{\pprime}(x)=\sum_a \pprime(a\mid x)\,\hat q(x,a)$. For sequences, approximate $\hat g_{\pprime}(x)$ with one (default) rollout $A'\!\sim\!\pprime(\cdot\mid x)$ and $R'=\hat f\!\big(s(x,A')\big)$; a light smoother (e.g., ridge over $(x,z)$ features) can reduce Monte Carlo noise. Cross-fitting is used throughout.

\subsection{Influence functions and variance}
\label{app:if-var}
Let $\{\phi_i\}_{i=1}^n$ denote the (approximately) centered influence–function contributions of $\hat\psi$, computed with cross–fitted/OOF nuisances and $R^{\mathrm{OOF}}$ along the IF path, so that $\frac{1}{n}\sum_{i=1}^n \phi_i \approx 0$. Under standard regularity conditions,
\[
\sqrt{n}\,\bigl(\hat\psi-\psi\bigr)\ \overset{d}{\longrightarrow}\ \mathcal{N}\!\left(0,\ \Var(\phi)\right),
\quad\text{hence}\quad
\Var(\hat\psi) \approx \frac{\Var(\phi)}{n}.
\]
Since $\E[\phi]=0$, we have $\Var(\phi)\approx \frac{1}{n}\sum_{i=1}^n \phi_i^2$. We estimate the variance of $\hat\psi$ (not $\phi$) and the total variance (including the oracle addition) by
\[
\widehat{\Var}_{\mathrm{main}}
\;=\;
\frac{1}{n^2}\sum_{i=1}^n \phi_i^2,
\qquad
\widehat{\Var}_{\mathrm{total}}
\;=\;
\widehat{\Var}_{\mathrm{main}}+\widehat{\Var}_{\mathrm{cal}},
\]
and report the $(1-\alpha)$ Wald interval
\[
\CI_{1-\alpha}:\ \hat\psi \ \pm\ z_{1-\alpha/2}\,\sqrt{\widehat{\Var}_{\mathrm{total}}}\,.
\]
When serial or cluster dependence is a concern, we additionally report dependence–robust SEs (e.g., cluster–robust sandwich or block/stationary bootstrap) as a sensitivity analysis.

\subsection{Calibration-aware jackknife}
\label{app:oua}
Partition $O$ into $K$ oracle folds $\{O_k\}_{k=1}^K$. For each $k$, refit reward calibration on $O\setminus O_k$, recompute $R^{(-k)}$, and rerun the full pipeline to obtain $\hat\psi^{(-k)}$. Then
\[
\bar\psi=\tfrac{1}{K}\sum_{k=1}^K \hat\psi^{(-k)},
\qquad
\widehat{\Var}_{\mathrm{cal}}=\frac{K-1}{K}\sum_{k=1}^K \big(\hat\psi^{(-k)}-\bar\psi\big)^2.
\]
(With unequal fold sizes, use the standard \emph{weighted} delete-one-group formula.)

\subsection{Diagnostics (definitions)}
\label{app:diagnostics-defs}
\textbf{ESS.} $\ESS(W)=\big(\sum_i W_i\big)^2/\sum_i W_i^2$; we report the fraction $\ESS/n$. Under global mean-one normalization (SNIPS), $\sum_i W_i=n$ and
\[
\frac{\ESS(W)}{n}\;=\;\frac{1}{1+\mathrm{CV}^2(W)}\quad\text{with}\quad \mathrm{CV}^2(W)=\Var(W)\ \text{(since }\E[W]=1\text{)}.
\]
Unless stated otherwise, diagnostics use the global (not per-fold) mean-one scaling.

\noindent\textbf{Max-weight share.} $\max_i W_i\big/\sum_j W_j$.

\noindent\textbf{Tail index (Hill).} For top-$k$ order statistics $W_{(1)}\ge\cdots\ge W_{(k)}$,
$\hat\alpha^{-1}=\tfrac{1}{k}\sum_{j=1}^k \log\!\big(W_{(j)}/W_{(k)}\big)$ (we sweep $k$ over a stability grid and report the plateau).

\noindent\textbf{Bhattacharyya affinity in $S$.} $A_B=\int \sqrt{p_{S|\pprime}(s)\,p_{S|\pzero}(s)}\,ds$ (discrete: sum over bins); $D_B=-\log A_B$.

\noindent\textbf{DR orthogonality score.} $n^{-1}\sum_i \hat W_{\pprime,i}\big(R^{\mathrm{OOF}}_i-\hat q^{\mathrm{OOF}}(X_i,A_i)\big)$ with a Wald CI.

\noindent\textbf{Coverage badge.} Plug-in estimate of $\Pr_{\pprime}\!\big(S\notin [S_{\min}^{\mathrm{orc}},S_{\max}^{\mathrm{orc}}]\big)$; large out-of-range mass with near-flat boundaries triggers \textsc{Limited Calibration Support} and \RefuseLevel.

\subsection{Symbol glossary}
\label{app:glossary}
\begin{center}
\begin{tabular}{ll}
\toprule
Symbol & Meaning \\
\midrule
$X,A$ & Context, action (sequence) \\
$S=s(X,A)$ & Judge score (scalar) \\
$Y$ & Ground-truth outcome (on oracle slice) \\
$\pzero,\pprime$ & Logger and candidate policies \\
$W_{\pprime}$ & Importance ratio $\pprime(A\mid X)/\pzero(A\mid X)$ \\
$W^{\mathrm{m1}}_{\pprime}$ & Mean-one (SNIPS) baseline \\
$R=\hat f(Z(S))$ & Calibrated reward (monotone or two-stage) \\
$\hat W_{\pprime}$ & Stabilized, unit-mean, $S$-monotone weights \\
$\hat q,\ \hat g_{\pprime}$ & Outcome and policy–value nuisances for DR \\
$\phi$ & Per-row centered influence-function contribution \\
$\widehat{\Var}_{\mathrm{main}}$ & Estimator variance $n^{-2}\sum_i\phi_i^2$ \\
$\widehat{\Var}_{\mathrm{cal}}$ & Calibration variance (oracle jackknife) \\
$\ESS(W)$ & Effective sample size \\
\bottomrule
\end{tabular}
\end{center}
\section{Algorithms (extended)}
\label{app:algorithms}

This appendix gives compact, cross-fitted pseudocode for \CJE{} modules: reward calibration, weight stabilization, estimators (calibrated IPS, calibrated DR), influence-function stacking, calibration-aware variance estimation, and cluster bootstrap for Direct mode. We reuse the same $K$-fold map $F(i)\!\in\!\{1{:}K\}$ across all modules. ``OOF'' denotes \emph{out-of-fold} predictions used along the IF path.

\begin{algorithm}[t]
\caption{\textsc{Reward-Calibration}: mean-preserving calibration (cross-fitted; automatic two-stage fallback)}
\label{alg:autocalr}
\begin{algorithmic}[1]
\STATE \textbf{Inputs:} Oracle pairs $\{(S_i,Y_i):L_i{=}1\}$; optional covariates $\{X_i\}$; folds $F(\cdot)$; smooth index class $g(\cdot)$
\STATE \textbf{Outputs:} Calibrated rewards $R_i$ and OOF $R^{\mathrm{OOF}}_i$
\FOR{$k=1$ to $K$}
  \STATE Train set $O_{\neg k}\!=\!\{i:L_i{=}1,\,F(i)\neq k\}$; test set $O_k\!=\!\{i:L_i{=}1,\,F(i)=k\}$
  \STATE \emph{Monotone candidate:} $\hat f^{(-k)}_{\uparrow}\!\in\!\arg\min_{f\in \mathcal M_{\uparrow}}\sum_{i\in O_{\neg k}}(Y_i-f(S_i))^2$; set $R^{\mathrm{OOF}}_{\uparrow,i}\!=\!\hat f^{(-k)}_{\uparrow}(S_i)$ for $i\!\in\!O_k$
  \STATE \emph{Two-stage candidate:} fit $g^{(-k)}(S,X)$ on $O_{\neg k}$; ranks $U_i\!=\!\mathrm{ECDF}_{O_{\neg k}}(g^{(-k)}(S_i,X_i))$; fit $\hat h^{(-k)}_{\uparrow}\!\in\!\arg\min_{h\in \mathcal M_{\uparrow}}\sum_{i\in O_{\neg k}}(Y_i-h(U_i))^2$; set $R^{\mathrm{OOF}}_{2s,i}\!=\!\hat h^{(-k)}_{\uparrow}(U_i)$ for $i\!\in\!O_k$
\ENDFOR
\STATE Default: use two-stage mode (recommended when covariates like response length are available); use monotone-only as fallback
\STATE Refit the selected mode on the full oracle slice to obtain \emph{global} $R_i$ for all $i\in\{1{:}n\}$; retain $R^{\mathrm{OOF}}_i$ per fold for IFs
\STATE \textit{Note:} The terminal isotonic step preserves the oracle-slice mean exactly.
\end{algorithmic}
\end{algorithm}

\begin{algorithm}[t]
\caption{\textsc{Weight-Stabilization}: surrogate-indexed, unit-mean monotone calibration (OOF project $\to$ stack $\to$ cap $\to$ normalize)}
\label{alg:simcalw}
\begin{algorithmic}[1]
\STATE \textbf{Inputs:} Baseline mean-one ratios $W^{\mathrm{m1}}_{\pi'}$; scores $S$; residuals $\Delta$ (\emph{IPS:} $\Delta{=}R$; \emph{DR:} $\Delta{=}R-\hat q$); folds $F(\cdot)$; variance cap $\rho\!\ge\!1$ (default $1$)
\STATE \textbf{Output:} Calibrated weights $\hat W_{\pi'}$ (mean-one, approximately $S$-monotone)
\FOR{$k=1$ to $K$}
  \STATE Train $I_{\neg k}=\{i:F(i)\neq k\}$; test $I_k=\{i:F(i)=k\}$ \COMMENT{OOF candidate projections}
  \STATE Fit isotonic regression of $W^{\mathrm{m1}}_{\pi'}$ on $S$ over $I_{\neg k}$: increasing $m^{(-k)}_{\uparrow}(S)$ and decreasing $m^{(-k)}_{\downarrow}(S)$ (via $-S$); rescale each to mean one on $I_{\neg k}$
  \STATE Predict on $I_k$: $W^{\mathrm{OOF}}_{\uparrow}=m^{(-k)}_{\uparrow}(S)$,\ \ $W^{\mathrm{OOF}}_{\downarrow}=m^{(-k)}_{\downarrow}(S)$,\ \ and optionally include $W^{\mathrm{OOF}}_{\mathrm{base}}\equiv 1$
\ENDFOR
\STATE \textbf{OOF stacking (variance-aware).} Form $U_c=W^{\mathrm{OOF}}_c\,\Delta$ for $c\in\{\mathrm{base},\uparrow,\downarrow\}$; compute $\hat\Sigma_{cd}=\mathrm{cov}(U_c,U_d)+\lambda \mathbf{1}_{c=d}$
\STATE Choose $\hat\beta\in \arg\min_{\beta\in\Delta_3}\beta^\top \hat\Sigma\,\beta$; set $W^{\mathrm{stack}}=\sum_c \hat\beta_c\,W^{\mathrm{OOF}}_c$; renormalize $W^{\mathrm{stack}}$ to sample mean one
\STATE \textbf{Light variance guard (optional).} $\alpha = \min\!\left\{1,\ \sqrt{\rho/\Var(W^{\mathrm{stack}})}\right\}$; set $W^{\mathrm{blend}}=1+\alpha\,(W^{\mathrm{stack}}-1)$
\STATE \textbf{Mean normalization.} $\hat W_{\pi'} \leftarrow W^{\mathrm{blend}}/\bar W^{\mathrm{blend}}$ \hfill(\emph{restores sample mean one})
\end{algorithmic}
\end{algorithm}

\begin{algorithm}[t]
\caption{\textsc{Cal-IPS}: calibrated importance sampling}
\label{alg:calips}
\begin{algorithmic}[1]
\STATE \textbf{Inputs:} Calibrated rewards $R,R^{\mathrm{OOF}}$ (Alg.~\ref{alg:autocalr}); calibrated weights $\hat W_{\pi'}$ (Alg.~\ref{alg:simcalw})
\STATE \textbf{Outputs:} $\widehat V_{\IPS}$ and IF $\phi^{\IPS}$
\STATE $\displaystyle \widehat V_{\IPS}=\frac{1}{n}\sum_{i=1}^n \hat W_{\pi',i}\,R_i$, \quad $\phi^{\IPS}_i=\hat W_{\pi',i}\,R^{\mathrm{OOF}}_i-\widehat V_{\IPS}$
\end{algorithmic}
\end{algorithm}

\begin{algorithm}[t]
\caption{\textsc{Calibrated-DR}: sequence-aware doubly robust estimator (cross-fitted)}
\label{alg:drcpo}
\begin{algorithmic}[1]
\STATE \textbf{Inputs:} $\hat W_{\pi'}$; $R$ and $R^{\mathrm{OOF}}$ from Alg.~\ref{alg:autocalr}; folds $F(\cdot)$
\STATE \textbf{Outputs:} $\widehat V_{\DR}$ and IF $\phi^{\DR}$
\FOR{$k=1$ to $K$}
  \STATE Train $\hat q^{(-k)}(x,a)\!\approx\!\E[R\mid x,a]$ on $\{i:F(i)\neq k\}$; predict OOF $\hat q^{\mathrm{OOF}}_i=\hat q^{(-k)}(X_i,A_i)$ for $i\in I_k$
  \STATE Approximate $\hat g^{(-k)}_{\pi'}(x)=\E_{A\sim \pi'(\cdot\mid x)}[\hat q^{(-k)}(x,A)]$ via one rollout $A'\sim\pi'(\cdot\mid X)$ and optional smoothing; obtain OOF $g^{\mathrm{OOF}}_{\pi',i}$ for $i\in I_k$
\ENDFOR
\STATE $\displaystyle \widehat V_{\DR}=\frac{1}{n}\sum_{i=1}^n \left\{ \hat g_{\pi'}(X_i)+\hat W_{\pi',i}\big(R_i-\hat q(X_i,A_i)\big)\right\}$
\STATE $\displaystyle \phi^{\DR}_i=\hat g_{\pi'}(X_i)+\hat W_{\pi',i}\big(R^{\mathrm{OOF}}_i-\hat q^{\mathrm{OOF}}_i\big)-\widehat V_{\DR}$
\end{algorithmic}
\end{algorithm}

\begin{algorithm}[t]
\caption{\textsc{Influence-Function Stacking}: variance-optimal convex ensembling in influence-function space}
\label{alg:ifstack}
\begin{algorithmic}[1]
\STATE \textbf{Inputs:} Candidates $\{\widehat V^{(e)},\phi^{(e)}\}_{e\in\mathcal E}$ (centered IFs; same folds); ridge $\lambda$
\STATE \textbf{Outputs:} $\widehat V^{(\hat\alpha)}$ and $\phi^{(\hat\alpha)}$
\STATE Form $\Phi=[\phi^{(e)}]_{e\in\mathcal E}$; $\hat\Sigma=(1/n)\Phi^\top\Phi+\lambda I$
\STATE Solve $\hat\alpha\in\arg\min_{\alpha\in\Delta}\alpha^\top \hat\Sigma\,\alpha$
\STATE $\widehat V^{(\hat\alpha)}=\sum_{e\in\mathcal E}\hat\alpha_e\,\widehat V^{(e)}$,\quad $\phi^{(\hat\alpha)}=\sum_{e\in\mathcal E}\hat\alpha_e\,\phi^{(e)}$
\STATE \textbf{Optional outer split:} learn $\hat\alpha$ on one half; apply to the other to reduce selection optimism
\STATE \textit{Support note (Carath\'eodory).} If $\mathrm{rank}(\Phi)=r$, the variance-optimal stack uses at most $r{+}1$ base estimators.
\end{algorithmic}
\end{algorithm}

\begin{algorithm}[t]
\caption{\textsc{Cal-Aware-Jackknife}: calibration-uncertainty-aware variance addition}
\label{alg:oua}
\begin{algorithmic}[1]
\STATE \textbf{Inputs:} Oracle folds $\{O_k\}_{k=1}^K$; end-to-end estimator $\widehat V(\cdot)$
\STATE \textbf{Outputs:} $\widehat{\Var}_{\mathrm{cal}}$ and $\widehat{\Var}_{\mathrm{total}}=\widehat{\Var}_{\mathrm{main}}+\widehat{\Var}_{\mathrm{cal}}$
\FOR{$k=1$ to $K$}
  \STATE Refit reward calibration on $O\setminus O_k$; recompute $R^{(-k)}$ and all downstream nuisances \& weights; run the full pipeline to get $\widehat V^{(-k)}$
\ENDFOR
\STATE $\displaystyle \bar V=\tfrac{1}{K}\sum_k \widehat V^{(-k)}$,\quad $\widehat{\Var}_{\mathrm{cal}}=\frac{K-1}{K}\sum_k\big(\widehat V^{(-k)}-\bar V\big)^2$
\STATE \textbf{Return:} $\widehat{\Var}_{\mathrm{total}}=\widehat{\Var}_{\mathrm{main}}+\widehat{\Var}_{\mathrm{cal}}$
\end{algorithmic}
\end{algorithm}

\begin{algorithm}[t]
\caption{\textsc{Bootstrap-Direct}: cluster bootstrap with calibrator refit for Direct mode}
\label{alg:bootstrap-direct}
\begin{algorithmic}[1]
\STATE \textbf{Inputs:} Eval data $\{(X_i, S_i, Y_i, C_i)\}_{i=1}^n$ with cluster IDs $C_i \in \{1,\ldots,M\}$; calibrator mode $m$ (fixed from full-data selection); oracle mask $L$; replicates $B{=}2000$; min oracle count $n_{\min}{=}30$
\STATE \textbf{Outputs:} Point estimates $\{\hat\theta_p\}$, 95\% CIs for each policy $p$
\STATE Fit calibrator $\hat f$ on oracle slice; compute $\hat\theta_p^{\mathrm{aug}}$ for each policy via \cref{eq:theta-aug}
\FOR{$b=1$ to $B$}
  \REPEAT
    \STATE Draw $M$ clusters with replacement: $\{C_1^*, \ldots, C_M^*\}$
    \STATE Form bootstrap sample $\mathcal{D}^* = \bigcup_{j=1}^M \{i : C_i = C_j^*\}$
  \UNTIL{$|L \cap \mathcal{D}^*| \geq n_{\min}$} \COMMENT{Resample-until-valid}
  \STATE Refit calibrator $\hat f^{(b)}$ on $\{(S_i, Y_i) : i \in L \cap \mathcal{D}^*\}$ with mode $m$
  \STATE Compute $\hat\theta_p^{(b)}$ for each policy using $\hat f^{(b)}$ and $\hat\theta_{\mathrm{aug}}$ on $\mathcal{D}^*$
\ENDFOR
\STATE \textbf{Percentile CIs:} $[\hat\theta_p^{\mathrm{lo}}, \hat\theta_p^{\mathrm{hi}}] = [\hat\theta_p^{(0.025B)}, \hat\theta_p^{(0.975B)}]$
\end{algorithmic}
\end{algorithm}

\vspace{-2pt}
\paragraph{Complexity notes.}
PAVA is $O(n)$ after a shared sort by $S$. Weight stabilization is linear-time per fold; the stacking QP is $3{\times}3$ (weights) or a small $|\mathcal E|{\times}|\mathcal E|$ system. Calibrated DR adds one rollout\,+\,judge per $(X,\pi')$. The calibration-aware jackknife refits the calibrator $K$ times and reruns the pipeline. \textsc{Bootstrap-Direct} requires $B{=}2{,}000$ refits of the reward calibrator; since isotonic regression is $O(n\log n)$, total cost is $O(Bn\log n)$. TF caches and precomputed features amortize cost.
\section{Proofs and Technical Lemmas}
\label{app:theory}

We collect standing identities, shape-constrained facts, and proofs for the results in \Cref{sec:theory}. Unless stated otherwise, expectations are under the logging law $P_{\pi_0}$; $L^2$ norms are with respect to the relevant marginal (e.g., $L^2(P)$ or $L^2(P_S)$). We reuse the fold map $F(\cdot)$ from \Cref{app:folds} and the projection operators from \Cref{app:projections}.

\subsection{Standing identities and tools}
\label{app:theory-tools}

\paragraph{Change of measure.}
For any integrable $h(X,A,S,Y)$ and any candidate $\pi'$,
\begin{equation}
\E\!\big[ W_{\pi'}\, h(X,A,S,Y)\big] \;=\; \E_{\pi'}\!\big[ h(X,A,S,Y)\big],
\qquad \E[W_{\pi'}]=1.
\label{eq:com}
\end{equation}

\paragraph{Doob--Dynkin / conditional expectation as $L^2$ projection.}
Let $\mathcal{G}=\sigma(S)$. Then $m^\star(S):=\E[W_{\pi'}\mid\mathcal{G}]$ is the $L^2$ projection of $W_{\pi'}$ onto the closed subspace $L^2(\mathcal{G})\subset L^2(P)$, i.e.,
\begin{equation}
\E\!\big[(W_{\pi'}-U(S))^2\big]
=\E\!\big[(W_{\pi'}-m^\star(S))^2\big]+\E\!\big[(m^\star(S)-U(S))^2\big],
\label{eq:orth}
\end{equation}
for all $U(S)\in L^2(\mathcal{G})$. In particular, $\Var(W_{\pi'}U)\ge \Var(m^\star(S)U)$ for any $U(S)\in L^2(\mathcal{G})$.

\paragraph{Pythagoras in Hilbert spaces.}
Let $L^2_0(P)$ be the mean-zero Hilbert space with inner product $\langle f,g\rangle=\E[fg]$. For a nonempty closed convex set $\mathcal{C}\subset L^2_0$ and any $z\in L^2_0$, denote by $\Pi_{\mathcal{C}}(z)$ the metric projection. If $\mathcal{C}_1\subseteq \mathcal{C}_2$ then $\mathrm{dist}(z,\mathcal{C}_2)\le \mathrm{dist}(z,\mathcal{C}_1)$.

\subsection{Isotonic regression: mean preservation and majorization}
\label{app:iso-facts}

\begin{lemma}[Mean preservation; PAVA]
\label{lem:iso-mean}
Let $\hat f\in\arg\min_{f\in\mathcal{M}_\uparrow}\sum_{i\in I} (y_i-f(s_i))^2$ be the isotonic fit (PAVA) on indices $I$. Then $\frac{1}{|I|}\sum_{i\in I}\hat f(s_i)=\frac{1}{|I|}\sum_{i\in I}y_i$.
\end{lemma}

\begin{lemma}[Dispersion reduction by majorization]
\label{lem:maj}
After sorting by $s$, the isotonic fitted vector $\hat u$ is a mean-preserving \emph{adjacent pooling} of $y$; hence for any convex $\phi$, $\sum_i \phi(\hat u_i)\le \sum_i \phi(y_i)$ \citep{HardyLittlewoodPolya1952,MarshallOlkinArnold2011}. In particular, with sample mean one, $\Var_n(\hat u)\le \Var_n(y)$ and $\ESS(\hat u)\ge \ESS(y)$.
\end{lemma}

Proofs are standard; see \citet{Ayer1955PAVA,Barlow1972,RobertsonWrightDykstra1988,Banerjee2001}.

\subsection{Proof of \texorpdfstring{\Cref{thm:sur-eif}}{Theorem 1} (surrogate EIF \& variance drop)}
\label{app:proof-sur-eif}
Let $R^\star=\E[Y\mid S]$ and $m^\star(S)=\E[W_{\pi'}\mid S]$. Under mean-sufficiency $\E[Y\mid X,A,S]=R^\star(S)$,
\begin{equation}
V(\pi')=\E\!\big[W_{\pi'}R^\star\big]=\E\!\big[m^\star(S)R^\star(S)\big].
\end{equation}
Standard semiparametric calculations (projecting the unconstrained score onto the tangent space of the surrogate model) yield
\begin{equation}
\phi_{\mathrm{sur}}(O;\pi') = g^\star_{\pi',R}(X)+m^\star(S)\big(R^\star-q^\star_R(X,A)\big)-V(\pi'),
\end{equation}
with $q^\star_R(x,a)=\E[R^\star\mid x,a]$ and $g^\star_{\pi',R}(x)=\E_{A\sim\pi'(\cdot\mid x)}[q^\star_R(x,A)]$ \citep{Bickel1993,VaartWellner2000}. Since $m^\star$ is the $L^2$ projection of $W_{\pi'}$ onto $L^2(\sigma(S))$, Pythagoras (or \eqref{eq:orth}) implies $\Var(\phi_{\mathrm{sur}})\le \Var(\phi_{\mathrm{uncon}})$, strictly unless $W_{\pi'}\in L^2(\sigma(S))$ and $R^\star$ is degenerate.

\subsection{Proof of \texorpdfstring{\Cref{thm:model-restriction}}{Theorem 2} (Efficiency via Model Restriction)}
\label{app:proof-ckp}
Let $\mathcal{M}$ be a baseline semiparametric model with tangent space $T(P)$ and canonical gradient $\phi^\star$.
Let $\mathcal{M}_\mathcal{C} \subset \mathcal{M}$ be a restricted model with tangent space $T_\mathcal{C}(P) \subseteq T(P)$
and canonical gradient $\phi^\star_\mathcal{C}$.

The canonical gradient is defined as the projection of the pathwise derivative onto the tangent space.
Since $T_\mathcal{C}(P) \subseteq T(P)$, we can decompose:
\[
\phi^\star = \phi^\star_\mathcal{C} + \phi^\perp,
\]
where $\phi^\perp \in T(P) \cap T_\mathcal{C}(P)^\perp$ (the orthogonal complement of $T_\mathcal{C}(P)$ within $T(P)$).

By the Pythagorean theorem in Hilbert spaces:
\[
\|\phi^\star\|_2^2 = \|\phi^\star_\mathcal{C}\|_2^2 + \|\phi^\perp\|_2^2 \ge \|\phi^\star_\mathcal{C}\|_2^2,
\]
with equality iff $\phi^\perp = 0$, i.e., $T_\mathcal{C}(P) = T(P)$ (no actual restriction).

Monotonicity for nested restrictions follows: if $T_{\mathcal{C}_1}(P) \supseteq T_{\mathcal{C}_2}(P)$ (i.e., $\mathcal{C}_2$ imposes stronger restrictions, yielding a smaller tangent space), then $\|\phi^\star_{\mathcal{C}_2}\|_2^2 \le \|\phi^\star_{\mathcal{C}_1}\|_2^2$.

For attainability, standard semiparametric arguments apply: one-step corrections or TMLE with cross-fitting
yield regular estimators achieving the restricted efficiency bound \citep{Bickel1993,VaartWellner2000,vanDerLaanRose2011TL}.

\subsection{Proof of \texorpdfstring{\Cref{cor:blackwell}}{Corollary: Blackwell monotonicity}}
\label{app:proof-blackwell}
If $\sigma(S_2)\subseteq\sigma(S_1)$, then $L^2(\sigma(S_2))\subseteq L^2(\sigma(S_1))$.
The surrogate model using $S_1$ has tangent space $T(S_1)$, and the model using $S_2$ has
$T(S_2) \subseteq T(S_1)$ when $\sigma(S_2)\subseteq\sigma(S_1)$.
By \Cref{thm:sur-eif} and the nested-tangent-space argument in \Cref{thm:model-restriction},
the canonical gradient in the smaller tangent space has variance no larger than in the larger one:
$\Var(\phi_{\mathrm{sur}}(S_1))\le \Var(\phi_{\mathrm{sur}}(S_2))$.
Strictness fails only if the finer knowledge already holds, i.e., if $W_{\pi'}$ is $\sigma(S_2)$-measurable and $R^\star$ is degenerate.

\subsection{Proof of \texorpdfstring{\Cref{prop:calips}}{Proposition 1} (Cal-IPS)}
\label{app:proof-calips}
Write
\[
\widehat V_{\IPS}-V(\pi')=(P_n-P)\!\big[m^\star(S)R^\star\big] + P\!\big[(\hat W_{\pi'}-m^\star)R^\star\big]+ P\!\big[m^\star(\hat f(Z(S))-R^\star)\big] + \mathrm{rem},
\]
where $\mathrm{rem}$ collects second-order sample-splitting terms. The empirical process term is $O_p(n^{-1/2})$; the second and third vanish by $L^2$-consistency of weight stabilization and reward calibration (monotone or two-stage) and Cauchy--Schwarz; the remainder is $o_p(1)$ by cross-fitting. Finite-sample dispersion control follows from \Cref{lem:maj} and, if used, the blend cap $\rho\ge1$; under the exact $\mathrm{IsoMeanOne}_S$ projection, $\ESS(\hat W_{\pi'})\ge \ESS(W_{\pi'}^{\mathrm{m1}})$.

\subsection{Proof of \texorpdfstring{\Cref{thm:dr-eff}}{Theorem 3} (Calibrated DR $\sqrt{n}$ limits)}
\label{app:proof-dr}
With cross-fitted nuisances and $R^{\mathrm{OOF}}$ along the IF path,
\[
\widehat V_{\DR}-V(\pi')
= (P_n-P)\!\big[\phi_{\mathrm{sur}}\big]
 + P\!\big[(\hat W_{\pi'}-m^\star)\{q^\star_R-\hat q\}\big]
 + o_p(n^{-1/2}).
\]
The second term is $o_p(n^{-1/2})$ by the one-of-two product-rate condition $\|\hat W_{\pi'}-m^\star\|_2\cdot \|\hat q-q^\star_R\|_2=o_p(n^{-1/2})$ and cross-fitting; the central limit theorem yields the limit variance $\Var(\phi_{\mathrm{sur}})$.

\subsection{Proof of \texorpdfstring{\Cref{thm:budgeted}}{Theorem 4} (budgeted bound)}
\label{app:proof-budget}
Let $\mathcal{W}_\rho=\{m:\E[m]=1,\ m\uparrow S,\ \E[(m-1)^2]\le \rho\,\E[(m^\star-1)^2]\}$. Intersecting the surrogate tangent space with the linear span induced by $m\in\mathcal{W}_\rho$ replaces $m^\star$ by its $L^2(P_S)$ projection $m^\star_\rho=\Pi_{\mathcal{W}_\rho}(m^\star)$ in $\phi_{\mathrm{sur}}$, giving $\phi^{(\rho)}$. Monotonicity in $\rho$ follows from nested convex sets $\mathcal{W}_{\rho_1}\subseteq \mathcal{W}_{\rho_2}$ for $\rho_1\le\rho_2$, and $\lim_{\rho\to\infty}\phi^{(\rho)}=\phi_{\mathrm{sur}}$. If weight stabilization converges to $m^\star_\rho$ (with the same cap), calibrated DR attains $\Var(\phi^{(\rho)})$ by the same one-of-two rate argument.

\subsection{Proof of \texorpdfstring{\Cref{thm:stack}}{Theorem 5} (influence-function stacking) and \texorpdfstring{\Cref{cor:caratheodory}}{Carath\'eodory sparsity}}
\label{app:proof-stack}
Let $\{\widehat V^{(e)}\}$ be regular and asymptotically linear with centered IFs $\{\phi^{(e)}\}$. Set $\Phi=[\phi^{(e)}]$ and $\hat\Sigma=(1/n)\Phi^\top\Phi+\lambda I$. A uniform law of large numbers on the simplex $\Delta$ yields $\hat\Sigma\to\Sigma$ uniformly; by argmin continuity, $\hat\alpha\to\alpha^\star\in\arg\min_{\alpha\in\Delta}\alpha^\top\Sigma\,\alpha$. Hence $\widehat V^{(\hat\alpha)}$ is asymptotically linear with IF $\phi^{(\alpha^\star)}=\sum_e \alpha^\star_e \phi^{(e)}$ and variance $\min_{\alpha\in\Delta}\alpha^\top\Sigma\,\alpha \le \min_e \Sigma_{ee}$. For \Cref{cor:caratheodory}: if $\mathrm{rank}(\Phi)=r$, then the feasible IF combinations lie in an $r$-dimensional affine subspace; by Carath\'eodory's theorem, any point in $\mathrm{conv}\{\phi^{(e)}\}$ admits a representation using at most $r{+}1$ extreme points.

\subsection{Proof of \texorpdfstring{\Cref{prop:oua}}{Proposition 2} (calibration-aware jackknife)}
\label{app:proof-oua}
Let $\widehat V(\hat f)$ be a regular estimator that depends on $f$ only through $R=\hat f(T(S))$, with $f\mapsto \widehat V(f)$ Hadamard-differentiable at $f^\star$ in $L^2(P_S)$. Using a delta-method expansion and cross-fitting (so that oracle folds are asymptotically independent of the IF path), the delete-one-oracle-fold jackknife \citep{Bickel1993,Kunsch1989BlockBootstrap,PolitisRomano1994StationaryBootstrap} consistently estimates the variance contribution from first-stage calibration. Therefore $\widehat{\Var}_{\mathrm{total}}=\widehat{\Var}_{\mathrm{main}}+\widehat{\Var}_{\mathrm{oracle}}$ is consistent for $\Var(\widehat V)$.

\subsection{Auxiliary lemmas used in the main proofs}
\label{app:aux}

\begin{lemma}[OOF mean preservation for reward calibration]
\label{lem:oof-mean}
Let $K$ be fixed and let $R^{\mathrm{OOF}}_i=\hat f^{(-F(i))}(Z(S_i))$ be OOF predictions from reward calibration (either mode). Then $P_n[R^{\mathrm{OOF}}]-P_n[Y]=o_p(1)$ and $P[R^{\mathrm{OOF}}-R^\star]=o_p(1)$ under $L^2(P_S)$-consistency of $\hat f$.
\end{lemma}

\begin{lemma}[Second-order remainder for DR with cross-fitting]
\label{lem:remainder}
Let $\widehat V_{\DR}$ be calibrated DR with cross-fitted $(\hat q,\hat g_{\pi'})$ and calibrated $\hat W_{\pi'}$. Then
\[
\widehat V_{\DR}-V(\pi')-(P_n-P)\phi_{\mathrm{sur}}
= P\!\big[(\hat W_{\pi'}-m^\star)(q^\star_R-\hat q)\big]+o_p(n^{-1/2}),
\]
and the bracketed term is $o_p(n^{-1/2})$ under $\|\hat W_{\pi'}-m^\star\|_2\cdot \|\hat q-q^\star_R\|_2=o_p(n^{-1/2})$.
\end{lemma}

\begin{lemma}[Guard stability]
\label{lem:guard}
Let $W^{\mathrm{stack}}$ be the OOF-stacked candidate and $W^{\mathrm{m1}}_{\pi'}$ the mean-one baseline. For $\rho\ge1$, define $\alpha=\min\!\big\{1,\sqrt{\rho\,\Var(W^{\mathrm{m1}}_{\pi'})/\Var(W^{\mathrm{stack}})}\big\}$ and $W^{\mathrm{blend}}=1+\alpha(W^{\mathrm{stack}}-1)$. Then $\Var(W^{\mathrm{blend}})=\alpha^2\Var(W^{\mathrm{stack}})\le \rho\,\Var(W^{\mathrm{m1}}_{\pi'})$, and the subsequent mean-one isotonic projection cannot increase empirical variance by \Cref{lem:maj}.
\end{lemma}

\emph{Implementation note.} The reference implementation uses a simpler absolute-cap formula $\alpha=\min\{1,\sqrt{\rho/\Var(W^{\mathrm{stack}})}\}$ and omits the final isotonic re-projection (normalizing to mean one instead). This relaxation empirically yields similar ESS improvements but does not carry the deterministic guarantee of \Cref{lem:guard}.

\begin{lemma}[Firm non-expansiveness of projections]
\label{lem:firm}
Metric projections onto closed convex sets in Hilbert spaces are firmly non-expansive:
$\|\Pi_{\mathcal{C}}(x)-\Pi_{\mathcal{C}}(y)\|^2 \le \langle \Pi_{\mathcal{C}}(x)-\Pi_{\mathcal{C}}(y),\, x-y\rangle$.
Hence compositions of projection modules (reward, weight, IF-space projections) are non-expansive. See, e.g., \citet[Prop.~4.16]{BauschkeCombettes2017}.
\end{lemma}

\paragraph{Remarks on dependence.}
If logs exhibit serial or cluster dependence, the IF CLTs can be replaced by block/stationary bootstrap arguments; our reported intervals can include a dependence-robust alternative (App.~\ref{app:diagnostics-defs}).

\subsection{Coverage-Limited Efficiency (CLE) bound}
\label{app:proof-cle}

We derive the CLE bound stated in \cref{eq:cle-bound}. Let $\mathcal{T} \subset \mathcal{X} \times \mathcal{A}$ be a target-relevant region with $\alpha = P_{\pi'}(\mathcal{T})$ and $\beta = P_{\pi_0}(\mathcal{T})$.

\begin{lemma}[Weight second moment on $\mathcal{T}$]
\label{lem:weight-identity}
For the normalized conditional distributions $\pi'_{\mathcal{T}}$ and $\pi_{0,\mathcal{T}}$,
\[
\E_{\pi_0}\!\big[W_{\pi'}^2\, \ind{\mathcal{T}}\big] = \frac{\alpha^2}{\beta}\bigl(1 + \chi^2(\pi'_{\mathcal{T}} \| \pi_{0,\mathcal{T}})\bigr).
\]
\end{lemma}
\begin{proof}
By definition of $\chi^2$ divergence for the normalized conditionals:
$\chi^2(\pi'_{\mathcal{T}} \| \pi_{0,\mathcal{T}}) = \int_{\mathcal{T}} \frac{(d\pi'_{\mathcal{T}})^2}{d\pi_{0,\mathcal{T}}} - 1
= \frac{\beta}{\alpha^2} \int_{\mathcal{T}} \frac{(d\pi')^2}{d\pi_0} - 1
= \frac{\beta}{\alpha^2} \E_{\pi_0}[W_{\pi'}^2 \ind{\mathcal{T}}] - 1$.
Rearranging gives the stated identity.
\end{proof}

\begin{proposition}[CLE bound for IPS-style estimators]
\label{prop:cle}
For any IPS-style estimator $\hat\Psi_{\mathrm{IPS}}$ of $V(\pi') = \E_{\pi'}[Y]$ based on importance-weighted outcomes,
\[
\SE(\hat\Psi_{\mathrm{IPS}}) \;\ge\; \frac{\sigma_{\mathcal{T}}\, \alpha}{\sqrt{\beta\, n}}\, \sqrt{1 + \chi^2(\pi'_{\mathcal{T}} \| \pi_{0,\mathcal{T}})},
\]
where $\sigma_{\mathcal{T}}^2 := \mathrm{ess\,inf}_{(x,a)\in \mathcal{T}} \Var(Y \mid X{=}x, A{=}a)$ is the minimal outcome noise in $\mathcal{T}$.
\end{proposition}

\begin{proof}
We lower bound $\Var_{\pi_0}(W_{\pi'} Y)$ using the Law of Total Variance, conditioning on $(X,A)$:
\[
\Var_{\pi_0}(W_{\pi'} Y) = \E_{\pi_0}[\Var(W_{\pi'} Y \mid X,A)] + \Var_{\pi_0}(\E[W_{\pi'} Y \mid X,A]).
\]
Since $W_{\pi'}$ is $(X,A)$-measurable:
\begin{align*}
\Var(W_{\pi'} Y \mid X,A) &= W_{\pi'}^2 \Var(Y \mid X,A), \\
\E[W_{\pi'} Y \mid X,A] &= W_{\pi'} \E[Y \mid X,A].
\end{align*}
Dropping the non-negative second term:
\[
\Var_{\pi_0}(W_{\pi'} Y) \ge \E_{\pi_0}[W_{\pi'}^2 \Var(Y \mid X,A)].
\]
Restricting to $\mathcal{T}$ (non-negative integrand) and applying $\Var(Y \mid X,A) \ge \sigma_{\mathcal{T}}^2$ on $\mathcal{T}$:
\[
\E_{\pi_0}[W_{\pi'}^2 \Var(Y \mid X,A)] \ge \sigma_{\mathcal{T}}^2 \E_{\pi_0}[W_{\pi'}^2 \ind{\mathcal{T}}].
\]
Applying \Cref{lem:weight-identity}:
\[
\Var_{\pi_0}(W_{\pi'} Y) \ge \frac{\sigma_{\mathcal{T}}^2 \alpha^2}{\beta}(1 + \chi^2).
\]
The asymptotic variance of IPS is $\Var_{\pi_0}(W_{\pi'} Y)/n$; taking square roots yields the stated bound.
\end{proof}

\paragraph{Interpretation.}
The CLE bound has three multiplicative factors:
(i) the \textbf{coverage penalty} $\alpha/\sqrt{\beta}$, which explodes when the logger rarely visits target-typical regions ($\beta \ll \alpha$);
(ii) the \textbf{shape mismatch} $\sqrt{1+\chi^2}$, which inflates the floor even with good coverage if the target concentrates differently within $\mathcal{T}$;
(iii) the standard \textbf{noise term} $\sigma_{\mathcal{T}}/\sqrt{n}$.

\paragraph{Corollary (ESS form).}
Using $\ESS(W_{\pi'}) = n / (1 + \CV^2(W_{\pi'}))$ and the relationship between $\chi^2$ and ESS on $\mathcal{T}$:
\[
\SE(\hat\Psi) \;\ge\; \frac{\sigma_{\mathcal{T}}}{\sqrt{\ESS_{\mathcal{T}}}},
\]
where $\ESS_{\mathcal{T}}$ is the effective sample size \emph{restricted to} $\mathcal{T}$. This shows that high global ESS is insufficient: what matters is ESS in target-relevant regions.

\subsection{TTC collapse with sequence length (structural coverage sparsity)}
\label{app:ttc-collapse}

Even with \emph{perfect} propensity estimates (no teacher-forcing noise), logger coverage collapses exponentially with sequence length. This formalizes why logs-only OPE is fundamentally impractical for open-ended generation, independent of TF brittleness.

\begin{proposition}[Exponential collapse of logger coverage]
\label{prop:ttc-collapse}
Let $q$ and $p$ be token-level distributions with $\mathrm{KL}(p \| q) = \delta > 0$ per token.
For i.i.d.\ sequences of length $L$, the trajectory-level KL divergence is $L \delta$, and
\[
\mathrm{TTC} := P_{\pi_0}(\text{target-typical region}) \;\le\; \exp\bigl(-\Omega(L \delta)\bigr).
\]
When $L \delta \gg 1$, logger coverage is exponentially small in $L$.
\end{proposition}

\begin{proof}[Sketch]
Under the product measure, trajectory log-likelihood ratios concentrate around $L \cdot \E_p[\log(p/q)] = L \delta$.
By standard large-deviation arguments (Sanov's theorem), the probability that a logger-sampled trajectory falls in the target-typical region decays as $\exp(-L \cdot D)$ for some $D > 0$ depending on $\delta$.
\end{proof}

\paragraph{Simulation verification.}
\Cref{fig:ttc-collapse} demonstrates this collapse empirically. We simulate sequences of length $L \in \{25, 50, 100, 200, 400\}$ with Zipf-distributed tokens ($V{=}1000$, $\alpha{=}1.2$) and reward-tilted target policies ($\eta \in \{0.1, 0.3, 0.5\}$, corresponding to small/medium/large shifts). Outcomes $Y = u_{a_0} + \varepsilon$ depend on the first token's utility (modeling a judge score with constant variance regardless of response length). Even with \emph{exact} importance weights (no TF noise), TTC collapses from $>0.8$ at $L=25$ to $<0.01$ at $L=400$ for $\eta \geq 0.3$. For medium and large shifts, IPS RMSE grows 20--50$\times$ with sequence length (from ${\sim}0.02$ to ${\sim}0.5$); for small shifts where TTC remains above 0.70, IPS stays closer to Direct. Direct RMSE remains flat at ${\sim}0.01$ regardless of sequence length or policy shift.

\begin{figure}[t]
  \centering
  \includegraphics[width=\columnwidth]{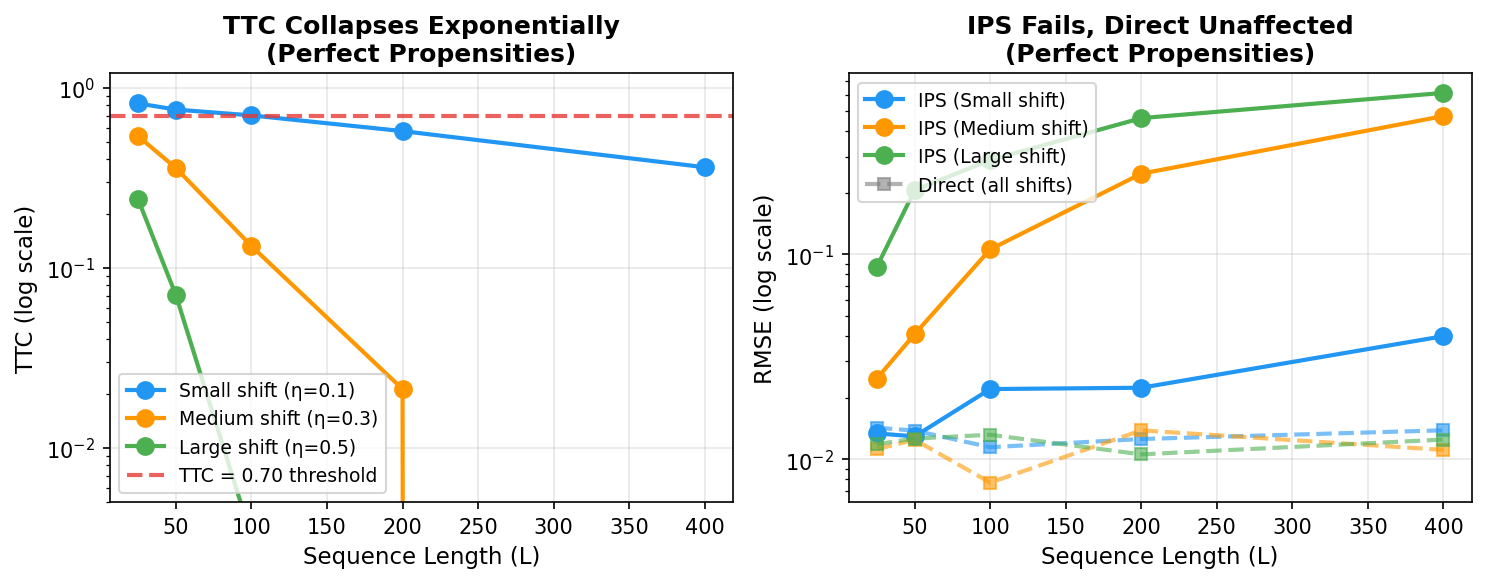}
  \caption{\textbf{TTC collapse with sequence length---even with perfect propensities.} Left: TTC decreases exponentially with $L$ for all policy shifts $\eta > 0$; only the smallest shift ($\eta{=}0.1$) stays above the 0.70 threshold beyond $L{=}100$. Right: For medium/large shifts, IPS RMSE (solid) grows 20--50$\times$ with $L$; for small shifts where TTC ${\geq} 0.70$, IPS stays closer to Direct. Direct RMSE (dashed) remains flat at ${\sim}0.01$ regardless of $L$. This demonstrates that coverage sparsity is a \emph{structural} limitation of autoregressive trajectory spaces, not a bug in propensity estimation.}
  \label{fig:ttc-collapse}
\end{figure}

\emph{Implication.} For open-ended LLM generation with typical response lengths ($L \approx 100$--$500$ tokens), even tiny per-token KL divergence ($\delta \approx 0.03$ nats) yields negligible TTC. This is not a TF bug---it is the geometry of high-dimensional trajectory spaces. Logs-only OPE faces a \emph{structural} precision floor that no estimator can overcome.

\subsection{EIF derivation for Direct mode}
\label{app:eif-direct}

Direct mode estimates $\theta := \E[Y]$ when oracle labels are observed only for a subset $L \subset [n]$.
Let $L_i \in \{0,1\}$ indicate oracle observation, and assume missing at random (MAR):
$\Pr(L_i=1 \mid Z_i, Y_i) = \Pr(L_i=1 \mid Z_i) =: e(Z_i)$, where $Z_i$ is the calibration index.
Under this model, the efficient influence function for $\theta$ is \citep{Tsiatis2006,BangRobins2005}:
\[
  \phi(O_i) = \mu(Z_i) - \theta + \frac{L_i}{e(Z_i)}\bigl(Y_i - \mu(Z_i)\bigr),
\]
where $\mu(z) := \E[Y \mid Z=z]$.
Averaging this EIF with estimated $(\hat\mu, \hat e)$ yields the one-step estimator.
In our setting, oracle inclusion is approximately uniform, so $\hat e(Z_i) \approx |L|/n$, giving:
\[
  \hat\theta_{\mathrm{aug}} = \frac{1}{n}\sum_{i=1}^n \hat\mu(Z_i) + \frac{1}{|L|}\sum_{i \in L}\bigl(Y_i - \hat\mu^{(-i)}(Z_i)\bigr),
\]
which matches \cref{eq:theta-aug}. The cross-fit predictions $\hat\mu^{(-i)}$ avoid overfitting bias in the residual term.

\emph{Remark (TMLE equivalence).} The cross-fitted augmented estimator $\hat\theta_{\mathrm{aug}}$ is the one-step/AIPW estimator for the EIF above, specializing to $e(Z) \approx |L|/n$ under approximately uniform oracle sampling. A TMLE that targets $\hat\mu$ via a logistic fluctuation using clever covariate $H(O) = L/e(Z)$ (equivalently, $H(Z) = 1/e(Z)$ on labeled rows) solves the EIF estimating equation and is asymptotically equivalent to $\hat\theta_{\mathrm{aug}}$, yielding the same influence function and achieving the semiparametric efficiency bound for $\theta$ in the MAR model \citep{vanDerLaanRose2011TL}.

\subsection{Assumptions and what CJE relaxes}
\label{app:assumptions-ledger}

CJE's key innovation is replacing untestable identification assumptions with auditable diagnostics. We organize assumptions into four categories, with \textbf{mean transport} as the fundamental identification requirement.

\paragraph{A0 (Bridge Assumption).}
$\E[Y^* \mid X, A] = \E[Y \mid X, A]$. The operational oracle $Y$ is sufficiently aligned with the idealized deliberation oracle $Y^*$ such that optimizing for $Y$ approximates optimizing for $Y^*$. This is a governance question outside the statistical framework---oracle selection requires domain expertise and construct validity audits.

\paragraph{Category A: Design/missingness requirements (always needed).}
\begin{itemize}[leftmargin=*,nosep]
\item \textbf{Oracle MAR (L1):} $L \perp Y \mid (X, A, S)$. Oracle labeling is ignorable conditional on observed data.
\item \textbf{Oracle positivity (L2):} $P(L=1 \mid X, A, S) > 0$ on the support where calibration will be applied.
\end{itemize}

\paragraph{Category B: Identification for surrogate-only mean claims (the fundamental requirement).}
The requirement for valid policy value estimation is \textbf{mean transport}:
\[
\E_{\pi'}[Y - f(S,X)] = 0.
\]
This is exactly the condition tested by the mean residual audit (\cref{prop:mean-transport}). CJE treats transport as testable rather than assumed:
\begin{itemize}[leftmargin=*,nosep]
\item \textbf{With audit} (${\sim}50$--$200$ oracle labels per $\pi'$): Test transport directly. Pass $\to$ reuse calibration; fail $\to$ recalibrate or refuse level claims.
\item \textbf{Without audit}: Transport must be assumed. In our experiments, calibration transported well across comparable policies (clone, premium, prompt-engineered) but failed for adversarial ones. How aggressively you audit depends on context---stakes, policy similarity, and historical transport stability.
\end{itemize}

\paragraph{Category C: Structural assumptions (sufficient for transport, not required).}
The following are \emph{sufficient conditions} that imply transport. CJE does not require them when transport is verified via audit:
\begin{itemize}[leftmargin=*,nosep]
\item \textbf{Prentice-style surrogacy:} $Y \perp A \mid (X, S)$ (distributional stability across policies).
\item \textbf{Mean sufficiency / single-index (J2-M, J2-SI):} $\E[Y \mid X, A, S] = \mu(g(S,X))$ for some monotone $\mu$ and index $g$.
\end{itemize}
These structural assumptions also enable the efficiency theory (\cref{thm:sur-eif,thm:dr-eff}), but are not required for consistent estimation when transport is audited. Mean sufficiency implies transport (if it holds globally), but transport can hold without it---errors may cancel in expectation even when the structural assumption fails.

\paragraph{Category D: Estimator-dependent requirements.}
\textbf{Overlap (S3):} $\pi'(a \mid x) > 0 \Rightarrow \pi_0(a \mid x) > 0$. Required for IPS and DR estimators. The Direct method requires no overlap---it generates fresh responses under $\pi'$.

\emph{Diagnostic:} Target-Typicality Coverage (TTC). If TTC $< 0.7$, logs-only IPS will fail due to coverage-limited efficiency; prefer Direct.

\subsection{Additional corollaries}
\label{app:additional-corollaries}

\begin{corollary}[Blackwell-efficiency monotonicity]
\label{cor:blackwell}
If $S_2$ is a garbling of $S_1$ (i.e., $\sigma(S_2)\subseteq\sigma(S_1)$), then
$\Var\!\big(\phi_{\mathrm{sur}}(S_1)\big)\le \Var\!\big(\phi_{\mathrm{sur}}(S_2)\big)$,
with strict inequality unless $W_{\pi'}$ is already $\sigma(S_2)$-measurable and $R^\star$ is degenerate.
This is a statistical analog of Blackwell's theorem on experiment comparisons \citep{Blackwell1953}: finer information structures dominate coarser ones.
\end{corollary}

\begin{corollary}[Carath\'eodory sparsity]
\label{cor:caratheodory}
If $\Phi=[\phi^{(e)}]$ has empirical rank $r$, the variance-optimal convex combination uses at most $r{+}1$
base estimators.
\end{corollary}

\subsection{What the bounds do not include}
\label{app:not-in-bounds}
The surrogate and budgeted information bounds describe \emph{model} limits: they do not include (i) finite-sample dispersion control from the sample cap (we encode it population-wise via $\rho$) or (ii) the oracle first-stage uncertainty (added separately by calibration-aware inference).

\subsection{Budgeted bound and IF stacking (OPE-specific results)}
\label{app:budgeted-stacking}

The following results apply to off-policy evaluation (IPS and DR estimators). Direct mode does not use importance weights and is unaffected by these bounds.

\paragraph{Budgeted bound (variance cap).}
For $\rho\ge 1$, define
\[
\mathcal{W}_\rho \;=\; \Big\{ m:\ \E[m]=1,\; m\uparrow S,\; \E\big[(m-1)^2\big]\le
\rho\,\E\big[(m^\star-1)^2\big]\Big\},
\]
and let $m^\star_\rho$ be the $L^2(P_S)$ projection of $m^\star$ onto $\mathcal{W}_\rho$.
Define the budgeted gradient
$\phi^{(\rho)}= g^\star_{\pi',R}(X) + m^\star_\rho(S)\big(R^\star-q^\star_R(X,A)\big) - V(\pi')$.

\begin{theorem}[Budgeted information bound]
\label{thm:budgeted}
The optimal asymptotic variance under the cap $\rho$ equals $\Var(\phi^{(\rho)})$,
which is nonincreasing in $\rho$ and satisfies
$\lim_{\rho\to\infty}\Var(\phi^{(\rho)})=\Var(\phi_{\mathrm{sur}})$.
If weight stabilization converges to $m^\star_\rho$ (with the guard), then calibrated DR attains $\Var(\phi^{(\rho)})$.
\end{theorem}

\noindent\emph{In plain terms:} This formalizes ``don't let weights explode.'' Limiting weight variance ($\rho$) trades a bit of asymptotic efficiency for much better finite-sample behavior.

\begin{theorem}[Influence-function stacking]
\label{thm:stack}
Let $\{\widehat V^{(e)}\}_{e\in\mathcal{E}}$ be regular, asymptotically linear estimators with centered IFs
$\{\phi^{(e)}\}$, and let $\hat\alpha\in\arg\min_{\alpha\in\Delta}\alpha^\top \hat\Sigma\,\alpha$ with
$\hat\Sigma$ the empirical IF covariance. If $\hat\Sigma\to \Sigma$ uniformly on $\Delta$, then the
stacked estimator $\widehat V^{(\hat\alpha)}$ is asymptotically linear with IF $\phi^{(\alpha^\star)}$ and
variance $\min_{\alpha\in\Delta}\alpha^\top \Sigma\,\alpha\le \min_e \Sigma_{ee}$.
An outer split leaves the limit unchanged. (This extends ensemble learning from prediction, where Super Learner \citep{vanderLaan2007SuperLearner} combines algorithms to minimize cross-validated risk, to asymptotic efficiency: we combine estimators by minimizing the variance of their combined influence function.)
\end{theorem}

\section{Diagnostics, Gates, and Reporting (details)}
\label{app:diagnostics}

This appendix formalizes the diagnostics used in \CJE{}, the associated \emph{ship/stop} gates, and the reporting ledger. The main text shows a compact panel per policy; here we provide precise formulas, defaults, and recommended thresholds. Unless stated otherwise, expectations and variances are empirical over the evaluation cohort (global SNIPS normalization).

\subsection{Weight behavior \& overlap}
\label{app:diag-weights}

\paragraph{Effective sample size (ESS).}
For nonnegative weights,
\[
\ESS(W)=\frac{\big(\sum_i W_i\big)^2}{\sum_i W_i^2}.
\]
Under global mean-one normalization (SNIPS), $\sum_i W_i=n$, so
\[
\frac{\ESS(W)}{n}=\frac{1}{1+\mathrm{CV}^2(W)},\qquad
\mathrm{CV}^2(W)=\Var(W)\ \text{when }\E[W]=1.
\]
Report the \emph{ESS fraction} $\ESS(W)/n$ and the multiplicative \emph{uplift}
$\ESS(\hat W_{\pi'})/\ESS(W^{\mathrm{m1}}_{\pi'})$.

\paragraph{Max-weight share.}
$\max_i W_i \big/ \sum_j W_j$ flags single-row dominance; display alongside the empirical 99.5th percentile of $W$.

\paragraph{Tail index (Hill) and CCDF.}
For the top-$k$ order statistics $W_{(1)}\ge\cdots\ge W_{(k)}$,
\[
\hat\alpha^{-1}(k)=\frac{1}{k}\sum_{j=1}^{k}\log\!\left(\frac{W_{(j)}}{W_{(k)}}\right),\qquad k\in\mathcal{K},
\]
swept over a stability grid $\mathcal{K}$ (e.g., 1--5\% of $n$). Plot $\hat\alpha(k)$ with a band over the plateau region (median and IQR over $\mathcal{K}$), and the empirical CCDF of $W$ on a log--log scale.

\paragraph{Overlap in judge space.}
Let $p_{S|\pi_0}$ and $p_{S|\pi'}$ denote the (binned) densities of $S$ under $\pi_0$ and $\pi'$:
\[
A_B=\sum_{b}\sqrt{p_b(\pi_0)\,p_b(\pi')},\qquad D_B=-\log A_B.
\]
Overlay an $S$-binned heatmap of $\log W$ to localize regions of poor overlap.

\subsection{Judge calibration, coverage, and drift}
\label{app:diag-judge}

\paragraph{Reliability diagram.}
Partition $S$ into $B$ bins; for bin $b$, plot the bin mean of $R=\hat f(Z(S))$ against the oracle mean of $Y$, with 95\% binomial intervals. Report a Brier-style \emph{reliability} term and the OOF RMSE.

\paragraph{Coverage badge.}
Estimate the fraction of evaluation mass outside the oracle $S$ range:
\[
\texttt{OutOfRange}=\widehat{\Pr}_{\pi'}\!\big(S<S_{\min}^{\text{orc}}\ \text{or}\ S>S_{\max}^{\text{orc}}\big).
\]
Also report \emph{boundary flatness} (slope of $\hat f$ in the lowest/highest oracle decile). Large \texttt{OutOfRange} together with flat boundaries triggers \textsc{Limited Calibration Support} and the \textsc{REFUSE-LEVEL} gate.

\paragraph{Rank drift (optional anchor).}
Given a fixed anchor set of $(X,A)$ pairs scored over time, compute Kendall’s $\tau$ between historical and current judge rankings with a permutation $p$-value. Change detection on residuals can be monitored via CUSUM/EWMA with FDR control across anchors.

\subsection{DR orthogonality and decomposition}
\label{app:diag-ortho}

\paragraph{Orthogonality score.}
Let $U_i=\hat W_{\pi',i}\big(R^{\mathrm{OOF}}_i-\hat q^{\mathrm{OOF}}(X_i,A_i)\big)$ and $\bar U=n^{-1}\sum_i U_i$.
Form a Wald CI for $\bar U$ using the standard error $\sqrt{\widehat{\Var}(U)/n}$ (or a cluster-/block-robust analogue). Report $\bar U$ and its CI; near-zero indicates successful orthogonality.

\paragraph{DM--IPS decomposition.}
Display $\widehat V_{\mathrm{DM}}=n^{-1}\sum_i \hat g_{\pi'}(X_i)$ and $\widehat V_{\mathrm{Aug}}=n^{-1}\sum_i U_i$,
with CIs and the empirical correlation between their per-row contributions.

\subsection{Uncertainty: IF variance and calibration uncertainty addition}
\label{app:diag-uncertainty}
For any estimator with centered IF contributions $\{\phi_i\}_{i=1}^n$,
\[
\widehat{\Var}_{\mathrm{main}}=\frac{1}{n}\,\widehat{\Var}(\phi_i),\qquad
\widehat{\Var}_{\mathrm{total}}=\widehat{\Var}_{\mathrm{main}}+\widehat{\Var}_{\mathrm{cal}},
\]
with $\widehat{\Var}_{\mathrm{cal}}$ from the oracle jackknife (App.~\ref{app:oua}). Report the \emph{calibration uncertainty share} $\widehat{\Var}_{\mathrm{cal}}/\widehat{\Var}_{\mathrm{total}}$ and, optionally, a dependence-robust alternative (below).

\paragraph{Dependence-robust SEs.}
When time/cluster dependence is suspected, also report:
(i) \emph{cluster-robust} sandwich SEs when a cluster id (e.g., session/user) is available; and
(ii) \emph{block/stationary bootstrap} intervals (block length chosen by a simple variance-stability sweep) \citep{Kunsch1989BlockBootstrap,PolitisRomano1994StationaryBootstrap}.

\paragraph{Bootstrap with bias-corrected estimation.}
For Direct mode, we recommend bootstrap inference with $\hat\theta_{\mathrm{aug}}$ (see \cref{ssec:method-oua}). \Cref{tab:uq-diagnostics} compares UQ methods across oracle fractions, showing that bootstrap with $\hat\theta_{\mathrm{aug}}$ achieves ${\sim}$95\% coverage by correcting both variance and bias.

\begin{table}[t]
\centering
\caption{\textbf{UQ diagnostics for Direct mode.} Coverage, centering ($\bar{Z}$), and calibration (SD$(Z)$) across oracle fractions. Well-calibrated intervals have coverage $\approx 95\%$, $\bar{Z} \approx 0$, and SD$(Z) \approx 1$.}
\label{tab:uq-diagnostics}
\small
\begin{tabular}{@{}llrrr@{}}
\toprule
\textbf{Oracle} & \textbf{Method} & \textbf{Coverage} & \textbf{$\bar{Z}$} & \textbf{SD$(Z)$} \\
\midrule
\multirow{3}{*}{5\%}
& Cluster-Robust & 22\% & $+1.52$ & 9.01 \\
& Cal-Aware Jackknife & 80\% & $+0.48$ & 3.11 \\
& Bootstrap + $\hat\theta_{\mathrm{aug}}$ & \textbf{95\%} & $+0.10$ & 0.98 \\
\midrule
\multirow{3}{*}{10\%}
& Cluster-Robust & 30\% & $+1.12$ & 6.34 \\
& Cal-Aware Jackknife & 80\% & $+0.40$ & 2.52 \\
& Bootstrap + $\hat\theta_{\mathrm{aug}}$ & \textbf{96\%} & $+0.15$ & 0.98 \\
\midrule
\multirow{3}{*}{25\%}
& Cluster-Robust & 42\% & $+0.36$ & 3.46 \\
& Cal-Aware Jackknife & 84\% & $+0.18$ & 2.09 \\
& Bootstrap + $\hat\theta_{\mathrm{aug}}$ & \textbf{97\%} & $+0.08$ & 0.95 \\
\midrule
\multirow{3}{*}{50\%}
& Cluster-Robust & 56\% & $-0.07$ & 2.63 \\
& Cal-Aware Jackknife & 87\% & $-0.05$ & 1.96 \\
& Bootstrap + $\hat\theta_{\mathrm{aug}}$ & \textbf{97\%} & $+0.11$ & 0.92 \\
\bottomrule
\end{tabular}
\vspace{1mm}

\footnotesize{
$Z = (\hat\theta - \theta) / \widehat{\mathrm{SE}}$ is the normalized error.
$N{=}600$ experiments per cell (4 sample sizes $\times$ 50 seeds $\times$ 3 policies).
$\hat\theta_{\mathrm{aug}}$ corrects the plug-in bias via AIPW-style residual augmentation.
}
\end{table}

\subsection{Multiplicity for many-policy comparisons}
\label{app:diag-multiplicity}
For contrasts $\Delta_{p}=\widehat V(\pi'_p)-\widehat V(\pi^\star)$, compute Wald $p$-values and apply BH at level $q\in[0.05,0.2]$; BY can be used under strong dependence. Provide a \emph{pairwise win matrix} with FDR marks and Kendall’s $\tau$ over policy means (all policies).

\subsection{Policy-wise mean transport diagnostic}
\label{app:diag-transport}

For each target policy $\pi'$ (or more generally, each evaluation environment such as a time period or subgroup), we test whether the calibration function $f(S,X)$ learned on the base policy remains mean-unbiased.

\paragraph{Test procedure.}
Given an oracle slice under policy $\pi'$, compute the residual $\varepsilon_{\pi'} = Y - f(S,X)$ and test
\[
H_{0,\pi'}: \E_{\pi'}[\varepsilon_{\pi'}] = 0 \quad \text{vs} \quad H_{1,\pi'}: \E_{\pi'}[\varepsilon_{\pi'}] \neq 0
\]
using a one-sample $t$-test. When testing multiple policies, apply Bonferroni or BH correction across environments.

\paragraph{Interpretation (Proposition~\ref{prop:mean-transport}).}
By \cref{prop:mean-transport}, $H_{0,\pi'}$ holds if and only if the surrogate-estimated policy value equals the oracle value: $\E_{\pi'}[f(S,X)] = \E_{\pi'}[Y]$. Rejecting $H_{0,\pi'}$ implies that surrogate-only evaluation is systematically biased for policy $\pi'$: the calibration does not transport.

\paragraph{Calibration-uncertainty-corrected inference.}
To account for uncertainty in estimating $f$, the confidence interval for the mean residual should incorporate calibration uncertainty (\cref{ssec:method-oua}). In practice, we refit $f^{(-k)}$ omitting each oracle fold and compute residuals under each refit, propagating calibration uncertainty into the test statistic.

\paragraph{Actions on failure.}
If a policy fails the mean transport test:
(i) flag surrogate-only value estimates for that policy as biased;
(ii) report the estimated bias ($\hat\Delta = \bar\varepsilon_{\pi'}$) and its direction;
(iii) either recalibrate using policy-specific oracle data or fall back to oracle-only evaluation for that cell;
(iv) rankings across policies that all pass may still be valid even if absolute levels are not.

\subsection{Gates: thresholds and actions}
\label{app:diag-gates}

\begin{table}[t]
\centering
\caption{Default gates (suggested; tighten for high-stakes launches).}
\label{tab:gates}
\setlength{\tabcolsep}{4pt}
\resizebox{\linewidth}{!}{%
\begin{tabular}{llp{7.8cm}}
\toprule
\bf Gate & \bf Default & \bf Action if failed \\
\midrule
\textsc{Overlap} & $\ESS/n \ge 0.30$ (${\le}3.3{\times}$ var); Hill $\hat\alpha\ge 2$ (finite var); $A_B\ge 0.85$ ($H{\le}0.39$); TTC $\ge 0.70$ &
Use overlap weights or cohort restriction; if TTC${}<0.70$, prefer Direct over IPS; report with warning if $\hat\alpha\in[1,2)$; do not ship offline conclusions if $\hat\alpha<1$ \\
\textsc{Judge} & Reliability band covers diagonal at knots; no persistent drift alarms &
Refresh/extend oracle slice; switch to two-stage index; re-validate \\
\textsc{Identification} & \texttt{OutOfRange} $\le \eta$ (default $\eta{=}5\%$; worst-case bias ${\le}\eta$) or non-flat boundaries &
Flag \textsc{Limited Calibration Support}; set \textsc{REFUSE-LEVEL}: report rankings + partial-ID only \\
\textsc{DR} & Orthogonality CI includes $0$; no NaNs; residual tails acceptable &
Strengthen nuisances/cross-fitting; fall back to stabilized IPS as a diagnostic \\
\textsc{Multiplicity} & FDR control applied when $|\Pi'|>5$ &
Report adjusted $p$-values; avoid uncorrected winner claims \\
\textsc{Cap} & Guard rarely engaged; CI width not sensitive to $\rho$ &
If guard active on $>50\%$ folds or sensitivity high, show cap curve and prefer overlap weights/restriction \\
\bottomrule
\end{tabular}}
\end{table}

\paragraph{Threshold calibration.}
The defaults in Table~\ref{tab:gates} are operationally calibrated for typical LLM evaluation regimes ($n \approx 5{,}000$, MDE $\approx 0.02$).
Practitioners may derive context-specific values:
(i)~ESS$/n \geq f$ bounds variance inflation to $1/f$ (e.g., $f{=}0.30 \Rightarrow {\leq}3.3{\times}$);
(ii)~Hill $\alpha > 2$ ensures finite weight variance, a CLT precondition;
(iii)~$A_B$ maps to Hellinger distance via $H = \sqrt{1 - A_B}$;
(iv)~\texttt{OutOfRange}${\leq}\eta$ bounds worst-case bias by $\eta$ under bounded outcomes.
For high-stakes applications, tighten thresholds or derive them from target MDE using the CLE bound (\cref{eq:cle-bound}).

\paragraph{REFUSE-LEVEL procedure.}
When \textsc{Identification} fails: (i) gray-out level estimates; (ii) highlight \texttt{OutOfRange} and boundary flatness; (iii) report rank-only conclusions with conservative relative CIs; (iv) recommend targeted labeling in uncovered $S$ regions.

\begin{table}[t]
\centering
\caption{Assumptions Ledger: Quick Reference. For each assumption, we list the estimators it applies to, how to test it, and recommended mitigations if violated.}
\label{tab:assumptions-ledger}
\setlength{\tabcolsep}{3pt}
\resizebox{\linewidth}{!}{%
\begin{tabular}{@{}p{2.2cm}p{1.8cm}p{4.5cm}p{4.5cm}@{}}
\toprule
\textbf{Assumption} & \textbf{Applies To} & \textbf{How to Test} & \textbf{If Violated} \\
\midrule
\textbf{Overlap} (positivity) & IPS, DR & ESS$/n \geq 0.30$; Hill $\hat\alpha \geq 2$; $A_B \geq 0.85$; TTC $\geq 0.70$ & Overlap weights; cohort restriction; if TTC${}<0.70$, use Direct \\
\textbf{Transport} (T1) & All (for valid levels) & Mean residual test: $\E_{\pi'}[Y - f(S,X)] = 0$ per policy & \emph{Fundamental requirement}; test with audit; recalibrate or refuse levels if fail \\
\textbf{Mean sufficiency} (J2-M) & Reward calibration & Reliability diagram; regional residuals centered at zero & \emph{Structural}, sufficient for transport; two-stage fallback; not required if transport passes \\
\textbf{Monotonicity} in $S$ & Reward calibration & OOF RMSE comparison & Default to two-stage; use monotone-only if covariates unavailable \\
\textbf{Calibration coverage} & All & \texttt{OutOfRange} $\leq 5\%$ (worst-case bias ${\le}5\%$); non-flat boundaries & \textsc{REFUSE-LEVEL}; targeted labeling in uncovered regions \\
\textbf{Shared prompts} & Direct & Paired design check (same prompt set across policies) & Cannot use Direct; must use OPE methods \\
\bottomrule
\end{tabular}}
\end{table}

\noindent\emph{Assumption hierarchy.} Transport (T1) is the fundamental requirement for valid surrogate-only level claims. Mean sufficiency (J2-M) and Prentice-style surrogacy are \emph{sufficient conditions} that imply transport, but transport can hold without them---errors may cancel in expectation. With an audit slice, test transport directly; structural assumptions become optional. See \cref{app:assumptions-ledger} for details.

\subsection{Planner: MDE and label/log budgets}
\label{app:diag-planner}
Given two independent estimates with equal SE $\widehat{\mathrm{SE}}$, the two-sided 95\% test at 80\% power has
\[
\mathrm{MDE}_{80\%}=(z_{0.8}+z_{0.975})\sqrt{2}\;\widehat{\mathrm{SE}}.
\]
We tabulate $\widehat{\mathrm{SE}}$ versus $(n,\,m/n)$ (labels per log) using \emph{Stacked-DR} with calibration-aware variance and annotate "iso-cost'' lines for the label budget.

\subsection{Weight Stabilization Visualization}
\label{app:weight-viz}

\Cref{fig:weight-stabilization-app} illustrates the effect of weight stabilization on importance weight distributions across target policies.

\begin{figure*}[t]
  \centering
  \includegraphics[width=\textwidth]{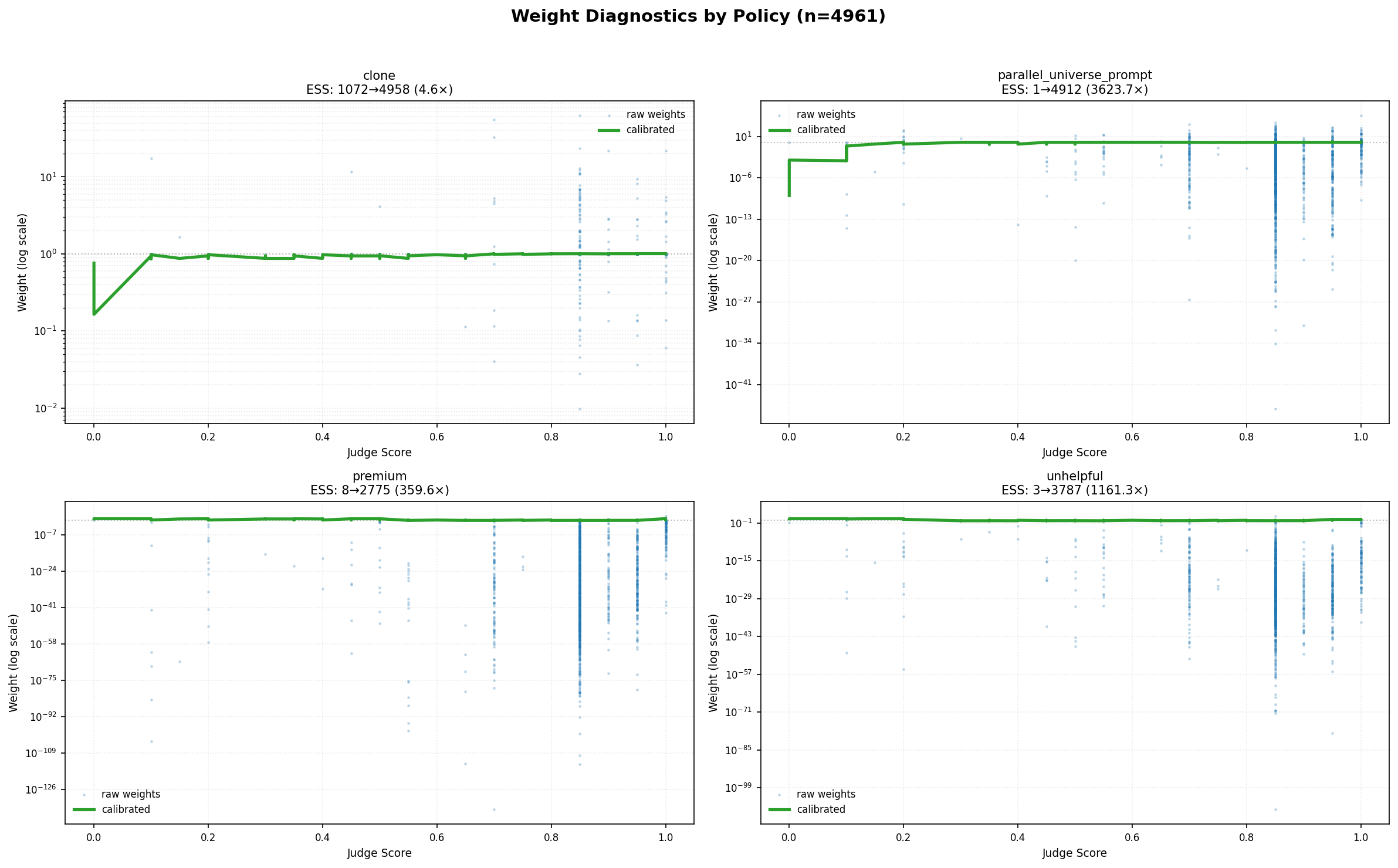}
  \caption{\textbf{Weight stabilization across Arena policies} ($n{=}4{,}961$ samples with complete logprobs for all four target policies). Raw importance weights (blue dots) span $10^{-130}$ to $10^{2}$; $S$-monotone projection (green line) stabilizes weights while preserving unit-mean. ESS improvements range from 4.6$\times$ (\texttt{clone}) to $>$3000$\times$ (\texttt{parallel\_universe\_prompt}). The \texttt{premium} policy shows weights spanning 130 orders of magnitude before stabilization. \textbf{Takeaway:} Weight stabilization recovers ESS from degenerate ($<$1\%) to healthy ($>$80\%), but high ESS alone doesn't guarantee ranking accuracy---coverage (TTC/CLE) is the binding constraint.}
  \label{fig:weight-stabilization-app}
\end{figure*}

\subsection{Reporting ledger (per policy/cohort)}
\label{app:diag-ledger}
Persist:
(i) calibrator mode, OOF risk by tertiles, knots/levels (hash);
(ii) weight stabilization maps, stacking weights $\hat\beta$, guard $\rho$ and blend $\alpha$;
(iii) ESS fraction, max-weight share, Hill index band, $A_B$;
(iv) DR orthogonality score and CI; DM--IPS split;
(v) calibration uncertainty trace $\{\widehat V^{(-k)}\}$ and variance breakdown;
(vi) filter counts (e.g., TF gaps) and an inclusion manifest of $x$-ids;
(vii) multiplicity control (family, $q$, adjusted $p$).

\subsection{Visualization primitives (for reproducible panels)}
\label{app:diag-viz}
\begin{itemize}[leftmargin=*,itemsep=2pt,topsep=2pt]
\item \textbf{ESS/tails strip:} bars for ESS fraction (baseline vs.\ weight stabilization); dot for max-weight share; Hill band.
\item \textbf{$S$-overlap heatmap:} density of $S$ under $\pi_0$ vs.\ $\pi'$ with overlaid $\log W$; annotate $A_B$.
\item \textbf{Reliability panel:} bin means of $(R,Y)$ with 95\% CIs; mode card (monotone vs.\ two-stage; OOF RMSE).
\item \textbf{Orthogonality panel:} point/CI for $\bar U$; DM--IPS bars with CIs and correlation.
\item \textbf{Uncertainty ring:} pie of $\widehat{\Var}_{\mathrm{cal}}/\widehat{\Var}_{\mathrm{total}}$ (calibration uncertainty share).
\end{itemize}

\subsection{Optional: dependence--robust implementation details}
\label{app:diag-dep}
\emph{Cluster-robust SEs:} if a cluster id $c(i)$ is available,
$\widehat{\Var}_{\mathrm{CR}}=n^{-2}\sum_{c}\big(\sum_{i\in c}\phi_i\big)\big(\sum_{i\in c}\phi_i\big)^{\!\top}$,
with finite-sample correction.\quad
\emph{Stationary bootstrap:} sample blocks of geometric length $\ell\sim\mathrm{Geom}(p)$ glued to length $n$; form the bootstrap distribution of $\hat\psi$ (or of $n^{-1/2}\sum\phi_i$) and report percentile or $t$-based bands.

\subsection{Compact gate pseudo-logic}
\label{app:diag-logic}
\begin{algorithm}[t]
\caption{Gate logic (per policy)}
\label{alg:gates}
\begin{algorithmic}[1]
\STATE Compute weight/tail metrics: ESS fraction, max-share, Hill band; compute $A_B$; judge reliability/coverage; orthogonality score; calibration uncertainty share
\IF{$\ESS/n<0.30$ \textbf{or} median Hill$<2$ \textbf{or} $A_B<0.85$} \STATE Flag \textsc{Overlap} (warn; restrict or use overlap weights) \ENDIF
\IF{\texttt{OutOfRange}$>\eta$ \textbf{and} boundary slopes ${\approx}\,0$} \STATE \textsc{REFUSE-LEVEL} $\leftarrow$ TRUE \ENDIF
\IF{Orthogonality CI excludes $0$} \STATE Flag \textsc{DR}; strengthen nuisances/cross-fitting \ENDIF
\IF{Cap engaged on $>50\%$ folds \textbf{or} CI sensitivity to $\rho$ high} \STATE Show cap--sensitivity; prefer overlap weights/restriction \ENDIF
\STATE Apply multiplicity control (BH/BY) when $|\Pi'|>5$
\end{algorithmic}
\end{algorithm}
\section{Implementation, Engineering, and Reproducibility}
\label{app:impl}

This appendix enumerates the concrete artifacts needed to reproduce \CJE{} end–to–end: a minimal logging schema, a teacher–forcing (TF) contract with conformance checks, fold construction, numerics, persisted outputs, and a lightweight resource model. No packages beyond the ICLR style file are required.

\subsection{Minimal logging schema (storage–agnostic)}
\label{app:impl-schema}
Each row corresponds to one prompt–continuation pair under the fixed logger $\pzero$. We persist only what is necessary to reconstruct SNIPS/IPS weights and judge scores.

\begin{table}[h]
\centering
\caption{Columns required for \CJE{}. Columnar formats (Parquet) are convenient but not required.}
\label{tab:schema}
\begin{tabular}{lll}
\toprule
\textbf{Field} & \textbf{Type} & \textbf{Description} \\
\midrule
\texttt{x\_id} & string & Stable identifier (hash of normalized prompt $+$ cohort) \\
\texttt{prompt} & bytes/string & Canonicalized $X$ (tokenizer $+$ normalization recorded) \\
\texttt{continuation} & bytes/string & Realized $A$ under $\pzero$ (full sequence) \\
\texttt{tokens} & int[] & Token ids for $A$ under each model’s TF tokenizer \\
\texttt{logp\_pi0} & float[] & Per–token $\log p_{\pi_0}(a_t\!\mid\!h_t)$ across $A$ \\
\texttt{judge\_S} & float/json & Scalar judge score $S=s(X,A)$ (or struct of sub–scores) \\
\texttt{judge\_cfg} & json & Judge rubric, decoding params, model snapshot hash \\
\texttt{run\_cfg} & json & $\pzero$ engine tag, decoding params, checkpoint hash, seed \\
\texttt{fold\_id} & int & Deterministic fold ($F(x\_id)$; see \S\ref{app:impl-folds}) \\
\texttt{cohort} & string & Optional slice label (time window, traffic source, etc.) \\
\bottomrule
\end{tabular}
\end{table}

\paragraph{TF cache (per target $\pi'$).}
A separate table stores, for each $(x\_id,\pi')$: \texttt{logp\_pi\_prime}, \texttt{logW}$=\log p_{\pi'}-\log p_{\pi_0}$, and
\[
W^{\mathrm{m1}}_{\pi'} \;=\; \exp\!\Big(\texttt{logW} - \mathrm{logsumexp}(\texttt{logW}) + \log n\Big),
\]
i.e., a single global denominator that enforces sample–mean–one. Rows with missing/invalid TF are filtered and recorded in a ledger.

\subsection{Teacher forcing: contract and conformance}
\label{app:impl-tf}
We require a \emph{single–call, chat–native} TF API that returns per–token and summed $\log p_\pi(A\!\mid\!X)$ under a fixed template, tokenizer, and snapshot. Client–side checks:

\begin{itemize}[leftmargin=*,itemsep=2pt,topsep=2pt]
\item \textbf{Determinism.} $k$ identical calls for the same $(X,A,\pi@\text{SNAPSHOT},\text{template})$ must be bit–identical (tolerance $<10^{-7}$).
\item \textbf{Additivity.} Also return $\log p_\pi(X)$ and $\log p_\pi(X{+}A)$ and verify
$\log p_\pi(X{+}A)\approx \log p_\pi(X)+\log p_\pi(A\!\mid\!X)\le \log p_\pi(X)$.
Violation $\Rightarrow$ discard row (and log it).
\item \textbf{Template/tokenizer provenance.} Return immutable hashes; reject moving aliases.
\item \textbf{Masking.} If safety masks are applied, return mask bits and a \texttt{renormalized} flag; prefer an evaluation–only path without hidden renormalization.
\end{itemize}

\noindent\textbf{Conformance snippet (pseudo).}
\begin{verbatim}
lp = TF(model_id, template_id, X, A)
assert same_bits(lp.sum, sum(lp.per_token), tol)
lpX, lpXA = TF_logp(X), TF_logp(X+A)
assert abs(lpXA - (lpX + lp.sum)) < eps and lp.sum <= 0
\end{verbatim}

\subsection{Folds and cross–fitting}
\label{app:impl-folds}
We use $K{=}5$ folds by default. The fold map $F(i)$ is a stable hash of \texttt{x\_id} modulo $K$, ensuring that all modules (reward calibration, weight stabilization, DR nuisances) share identical OOF boundaries. Oracle folds are derived by intersecting $F(i)$ with $L_i{=}1$. The hash rule is $F(x\_id){=}\mathrm{hash}(x\_id)\bmod 5$.

\subsection{Numerics and stability}
\label{app:impl-numerics}
\begin{itemize}[leftmargin=*,itemsep=2pt,topsep=2pt]
\item \textbf{Ratios in log–space.} Keep $\log W$ until forming the global mean–one normalization; use a single \texttt{logsumexp} for the denominator.
\item \textbf{Center residuals.} For stacking objectives and covariances, center $\Delta$ and drop NaNs/inf at ingestion.
\item \textbf{Variance estimates.} Use Welford’s online formulas for high dynamic range; add a tiny ridge ($\lambda\!\in\![10^{-10},10^{-6}]$) to covariance matrices.
\item \textbf{PAVA.} Run once per fold after sorting by $S$. The reference implementation enforces mean-one via \emph{multiplicative rescaling} ($W/\bar W$), which preserves nonnegativity but may slightly distort isotonicity at extremes. \emph{Note:} The exact constrained projection $\mathrm{IsoMeanOne}_S$ defined in \cref{app:projections} preserves both monotonicity and mean-one simultaneously; the deterministic majorization guarantee strictly applies only to that exact operator. We empirically observe similar ESS improvements with multiplicative rescaling.
\item \textbf{Guard (absolute variance cap).} Default $\rho{=}1$; compute $\alpha=\min\!\big\{1,\, \sqrt{\rho/\Var(W^{\mathrm{stack}})}\big\}$; blend and normalize to mean–one. Persist whether the guard engaged. \emph{Note:} This is an absolute cap ($\Var\le\rho$); the theoretical Lemma~\ref{lem:guard} uses a relative bound ($\Var\le\rho\cdot\Var(W^{\mathrm{m1}})$).
\end{itemize}
\emph{Reference code note.} A minimal reference implementation uses \emph{global} isotonic fits and computes the stacking covariance on those same \emph{in-sample} fits (no OOF); production code should use cross-fitting as described above.

\subsection{Persisted artifacts (per policy/cohort)}
\label{app:impl-artifacts}
\begin{itemize}[leftmargin=*,itemsep=2pt,topsep=2pt]
\item \textbf{Calibrator:} mode (monotone vs.\ two–stage), OOF RMSE (overall $+$ tertiles), knots/levels (hash), OOF vs.\ pooled predictions.
\item \textbf{Weights:} isotonic merge metadata, $S$ orientation (up/down), \emph{stacking weights $\hat\beta$ including the identity/baseline candidate}, guard $\rho$ and blend $\alpha$, final mean–one check.
\item \textbf{Estimators:} point estimates, centered IF vectors’ hashes, $\widehat{\Var}_{\mathrm{main}}$, orthogonality score and CI, dependence–robust SEs (if used).
\item \textbf{Calibration-aware inference:} $\{\widehat V^{(-k)}\}_{k=1}^K$, $\widehat{\Var}_{\mathrm{cal}}$, $\widehat{\Var}_{\mathrm{total}}$.
\item \textbf{Diagnostics:} ESS fraction, max–weight share, Hill band, $S$–overlap ($A_B$), coverage badge, gate statuses.
\item \textbf{Ledger:} counts by filter reason (TF gaps, moderation, timeouts), \texttt{x\_id} inclusion manifest.
\end{itemize}

\subsection{Reference run order (pseudocode)}
\label{app:impl-run}
\begin{verbatim}
# 0) Build TF cache for each pi' (one pass per policy)
build_tf_cache --policies <list> --dataset logs.parquet \
               --out tf_cache.parquet

# 1) Reward calibration (cross-fitted; auto monotone vs two-stage)
autocal_r --oracle oracle.parquet --folds 5 --out rewards.parquet

# 2) Weight stabilization per policy
#    (OOF project -> stack -> cap -> translate-to-mean-one)
weight_stabilize --tf-cache tf_cache.parquet --scores S.parquet \
                 --rho 1.0 --folds 5 --out weights.parquet

# 3) Estimation + IFs (Calibrated-IPS / Calibrated-DR)
estimate --rewards rewards.parquet --weights weights.parquet \
         --folds 5 --out estimates.parquet

# 4) Influence-function stacking (optional)
stack --estimates estimates.parquet --out stacked.parquet

# 5) Calibration-aware jackknife
cal_jackknife --oracle-folds 5 --pipeline-config cfg.yaml \
              --out variance.parquet

# 6) Report (diagnostics, gates, CIs)
report --inputs *.parquet --figs figs/ --out report.html
\end{verbatim}

\subsection{Compute and resource model}
\label{app:impl-compute}
Let $n$ be prompts, $\bar T$ the mean continuation length, and $|\Pi'|$ the number of candidate policies.
\begin{itemize}[leftmargin=*,itemsep=2pt,topsep=2pt]
\item \textbf{TF cache.} $O(|\Pi'|\,n\,\bar T)$ forward tokens; microbatch by length; near–linear scaling across GPUs.
\item \textbf{Weight stabilization.} $O(n\log n)$ for sort $+$ $O(n)$ for PAVA per fold; covariance/stacking are tiny ($3{\times}3$).
\item \textbf{Calibrated DR.} If $\hat g_{\pi'}$ uses one rollout per $(X,\pi')$, add $O(|\Pi'|\,n\,\bar T')$ tokens once; a light smoother amortizes Monte Carlo noise.
\item \textbf{Calibration-aware inference.} $K$ refits of reward calibration and re–runs of the pipeline; cache features to avoid recomputation.
\end{itemize}

\subsection{Determinism, versioning, and privacy}
\label{app:impl-determinism}
\textbf{Determinism:} fix seeds at engine, dataloader, and sampler; record random states in \texttt{run\_cfg}; serialize fold maps (internal). \\
\textbf{Versioning:} record immutable hashes for model weights, tokenizer, and template; pin checkpoints. \\
\textbf{Privacy:} encrypt prompts/continuations at rest; hash \texttt{x\_id} with salt; public artifacts include only aggregates/diagnostics and redacted IDs.

\subsection{Exact prompts used in experiments}
\label{app:impl-prompts}

\paragraph{Policy system prompts.}
All policies use Llama 3.3 70B except \texttt{premium} (Llama 3.1 405B), served via Fireworks AI (\url{https://fireworks.ai}). Temperature = 0.7 for all; no top-$p$/top-$k$ truncation, so TF logprobs represent the exact sampling distribution.
\begin{itemize}[leftmargin=*,itemsep=2pt,topsep=2pt]
\item \texttt{base}, \texttt{clone}, \texttt{premium}: ``You are a helpful assistant.''
\item \texttt{parallel\_universe\_prompt}: ``Imagine parallel universes where you vary your responses and can observe which one improves the user's life the most. Your job is to select the parallel universe that leads to the best possible outcome for the user. Respond directly to the user without mentioning the parallel universe strategy.''
\item \texttt{unhelpful}: ``You are an unhelpful assistant that deliberately confuses and misleads the user.''
\end{itemize}

\paragraph{Judge and oracle evaluation prompts.}
Both the judge (\texttt{gpt-4.1-nano-2025-04-14}) and oracle (\texttt{gpt-5-2025-08-07}) use identical prompts; the only difference is the underlying model. Temperature = 0.0 for judge, 1.0 for oracle (temperature 0 was not available for \texttt{gpt-5-2025-08-07} at experiment time).

\noindent\textbf{System prompt:} ``You are an AI evaluator. Rate responses from 0 to 100. Always provide a score, even if the response is incomplete or truncated.''

\noindent\textbf{User prompt template:}
\begin{verbatim}
<task>
Evaluate the quality of the AI assistant's response to the user's question.
Score from 0-100 based on relevance, helpfulness, clarity, and completeness.
If the response is truncated, score based on what is provided.
</task>

<exchange>
<question>{prompt}</question>
<answer>{response}</answer>
</exchange>

<instruction>
Provide your evaluation score (0-100):
</instruction>
\end{verbatim}

\section{Efficient Budget Allocation}
\label{app:budget}

While CJE provides calibration-aware confidence intervals for any sample configuration, practitioners often face a constrained optimization problem: minimizing estimation error subject to a fixed financial or latency budget. Here we derive the optimal allocation between cheap surrogate scores ($n$) and expensive oracle labels ($m$).

\begin{proposition}[Square Root Allocation Law]
\label{prop:sqrt-allocation}
Let $c_S$ and $c_Y$ be the marginal costs of surrogate scoring and oracle labeling, respectively. Assume the variance decomposes as
\[
V_{\mathrm{total}}(n,m)\;\approx\;\frac{\sigma^2_{\mathrm{eval}}}{n}+\frac{\sigma^2_{\mathrm{cal}}}{m},
\qquad B=c_S n+c_Y m.
\]
The allocation $(n^*, m^*)$ minimizing $V_{\mathrm{total}}$ subject to budget $B$ satisfies:
\[
\boxed{\frac{m^*}{n^*} = \sqrt{\frac{c_S}{c_Y}} \cdot \sqrt{\frac{\sigma^2_{\mathrm{cal}}}{\sigma^2_{\mathrm{eval}}}}}
\]
with the feasibility constraint $m^* \le n^*$ (the oracle slice is a subset of scored examples). If the unconstrained optimum exceeds $n$, set $m^* = n$.
\end{proposition}

\emph{Note:} This first-order approximation assumes additive variance components with $O(1/m)$ calibration variance. While isotonic regression has nonparametric rates for \emph{function estimation}, the \emph{mean functional} $\E[f(S)]$ typically achieves $\sqrt{m}$ rates under smoothness. In practice, estimate $\sigma^2_{\mathrm{cal}} \approx m \cdot \widehat{\Var}_{\mathrm{cal}}$ and $\sigma^2_{\mathrm{eval}} \approx n \cdot \widehat{\Var}_{\mathrm{eval}}$ from the calibration-aware variance decomposition, then apply the formula.

\paragraph{Closed-form solution.} Solving the constrained optimization yields:
\[
n^* = \frac{B\,\sqrt{\sigma^2_{\mathrm{eval}}/c_S}}{\sqrt{c_S\sigma^2_{\mathrm{eval}}}+\sqrt{c_Y\sigma^2_{\mathrm{cal}}}},
\qquad
m^* = \frac{B\,\sqrt{\sigma^2_{\mathrm{cal}}/c_Y}}{\sqrt{c_S\sigma^2_{\mathrm{eval}}}+\sqrt{c_Y\sigma^2_{\mathrm{cal}}}}.
\]

\begin{proof}
We form the Lagrangian $\mathcal{L}$ for minimizing variance subject to cost $B$:
\[
\mathcal{L}(n, m, \lambda) = \frac{\sigma^2_{\mathrm{eval}}}{n} + \frac{\sigma^2_{\mathrm{cal}}}{m} + \lambda(c_S n + c_Y m - B)
\]
Taking partial derivatives with respect to $n$ and $m$ yields the first-order conditions:
\[
\frac{\partial \mathcal{L}}{\partial n} = -\frac{\sigma^2_{\mathrm{eval}}}{n^2} + \lambda c_S = 0 \implies \lambda = \frac{\sigma^2_{\mathrm{eval}}}{c_S n^2}
\]
\[
\frac{\partial \mathcal{L}}{\partial m} = -\frac{\sigma^2_{\mathrm{cal}}}{m^2} + \lambda c_Y = 0 \implies \lambda = \frac{\sigma^2_{\mathrm{cal}}}{c_Y m^2}
\]
Equating the expressions for $\lambda$:
\[
\frac{\sigma^2_{\mathrm{eval}}}{c_S n^2} = \frac{\sigma^2_{\mathrm{cal}}}{c_Y m^2} \implies \frac{m^2}{n^2} = \frac{\sigma^2_{\mathrm{cal}}}{\sigma^2_{\mathrm{eval}}} \cdot \frac{c_S}{c_Y}
\]
Taking the square root yields the ratio result. Substituting into the budget constraint and solving gives the closed-form expressions.
\end{proof}

\paragraph{The Calibration Uncertainty Marginal Utility Diagnostic.}
The optimality condition implies that resources are allocated efficiently when the marginal variance reduction per dollar is equal across both channels. We can express this in terms of the observed calibration uncertainty share $\omega = \Var_{\mathrm{cal}}/\Var_{\mathrm{total}}$. At the optimum, the \textbf{Spend-Balance Rule} holds:
\[
\omega \;=\; \frac{c_Y m}{c_S n + c_Y m} \;=\; \frac{\text{Spend}_{\mathrm{oracle}}}{\text{Spend}_{\mathrm{total}}}.
\]
That is, the variance share from calibration should equal the budget share spent on oracle labels. Deviations diagnose inefficiency:

\begin{itemize}[leftmargin=*,nosep]
\item \textbf{Under-labeled (Invest in $m$):} If $\omega > \frac{c_Y m}{c_S n + c_Y m}$, calibration uncertainty dominates relative to its share of the budget. \emph{Action:} Shift budget to acquire more oracle labels.
\item \textbf{Over-labeled (Invest in $n$):} If $\omega < \frac{c_Y m}{c_S n + c_Y m}$, evaluation sampling noise dominates. \emph{Action:} Shift budget to evaluate more prompts with the surrogate.
\end{itemize}

\paragraph{Worked Example (Arena Benchmark).}
Using costs from our experiments: the \texttt{gpt-4.1-nano-2025-04-14} judge is approximately $16\times$ cheaper than the \texttt{gpt-5-2025-08-07} oracle, giving a cost ratio $c_S/c_Y \approx 0.064$. At $n=1000$ prompts with $m=50$ oracle labels (5\% oracle fraction), the observed calibration uncertainty share is $\omega \approx 90\%$.

To find the intrinsic variance ratio, note that $\omega = \frac{\sigma^2_{\mathrm{cal}}/m}{\sigma^2_{\mathrm{cal}}/m + \sigma^2_{\mathrm{eval}}/n}$, which implies:
\[
\frac{\sigma^2_{\mathrm{cal}}}{\sigma^2_{\mathrm{eval}}} = \frac{\omega \cdot m}{(1-\omega) \cdot n} = \frac{0.90 \times 50}{0.10 \times 1000} = 0.45
\]

Applying the Square Root Law:
\[
\frac{m^*}{n^*} = \sqrt{0.064} \times \sqrt{0.45} \approx 0.25 \times 0.67 = 17\%
\]

The variance-optimal oracle fraction is ${\sim}17\%$, substantially higher than the empirical 5\%. This indicates the study under-invested in oracle labels relative to the variance-optimal allocation. However, the 5\% allocation was chosen for cost efficiency: it achieves 94\% ranking accuracy at 14$\times$ lower cost than pure oracle labeling, trading some precision for substantial cost savings.

\paragraph{Extension: Multi-Policy Amortization.}
When calibrating once to evaluate $P$ policies on the same prompt set, surrogate costs scale as $c_S \cdot P \cdot n$ while oracle costs remain $c_Y \cdot m$. The optimal ratio becomes $m^*/n = \sqrt{P \cdot c_S \cdot \sigma^2_{\mathrm{cal}} / (c_Y \cdot \sigma^2_{\mathrm{eval}})}$, favoring larger oracle slices as $P$ grows (calibration cost amortizes across policies, making it relatively cheaper per evaluation).
\section{Code Generation Experiment (MBPP)}
\label{app:mbpp}

To validate that CJE generalizes beyond preference judgments, we evaluated on MBPP code generation \citep{Austin2021MBPP}, where ground truth is objective: unit tests either pass or fail.

\paragraph{Setup.}
We use 257 MBPP problems (974 samples across 5 epochs) with five policies: \texttt{base} (GPT-4o standard, 87.4\% pass rate, used for calibration), \texttt{enhanced} (GPT-4o with expert prompting, 87.2\%), \texttt{terse} (GPT-4o with minimal prompting, 87.4\%), \texttt{mini} (GPT-4o-mini, 80.9\%), and \texttt{unhelpful} (deliberately buggy, 56.0\%).
The judge is GPT-4o-mini predicting pass probability (0--100, normalized to 0--1).
Oracle labels are actual unit test results.
Unlike Arena (where response length confounds judge scores), MBPP has no strong covariate effects, so we use monotone-only calibration.

\paragraph{Transport validation.}
We apply the mean transport test (\cref{prop:mean-transport}) to each target policy. Two policies fail:
\begin{itemize}[leftmargin=*,nosep]
\item \texttt{mini}: Mean residual $-0.065$ ($p{=}0.006$). The judge model (GPT-4o-mini) is identical to this policy, creating a judge--policy confound that violates transport.
\item \texttt{unhelpful}: Mean residual $-0.280$ ($p{<}0.001$). Low-support saturation: 31\% of samples score below base's 10th percentile, where calibrator training density is sparse.
\end{itemize}
Both policies are excluded from RMSE computation; including them would bias results.

\paragraph{Results.}
\Cref{tab:mbpp-results} shows CJE performance across oracle fractions for the two transport-valid policies (\texttt{enhanced}, \texttt{terse}) at full sample size ($n{=}974$).
Calibration provides consistent RMSE reduction across all oracle fractions: even at 5\% oracle (${\sim}50$ labels), CJE achieves 18\% improvement over naive estimation (raw judge scores without calibration).
At 10\%+ oracle, improvements reach 32--84\%.

\begin{table}[h]
\centering
\small
\caption{\textbf{CJE on MBPP code generation: oracle fraction tradeoff.}
RMSE computed over two transport-valid policies (\texttt{enhanced}, \texttt{terse}) at $n{=}974$.
Unlike Arena where 5\% oracle suffices for ranking, MBPP shows diminishing returns: 10\% oracle captures most of the benefit.}
\label{tab:mbpp-results}
\begin{tabular}{lcccc}
\toprule
Oracle \% & \# Labels & CJE & Naive (uncalib.) & Improvement \\
\midrule
5\%   & 49   & 0.039 & 0.048 & +18\% \\
10\%  & 97   & 0.033 & 0.048 & +32\% \\
25\%  & 244  & 0.019 & 0.048 & +60\% \\
50\%  & 487  & 0.014 & 0.048 & +70\% \\
100\% & 974  & 0.008 & 0.048 & +84\% \\
\bottomrule
\end{tabular}
\end{table}

\paragraph{Comparison to Arena.}
Both experiments validate the oracle fraction tradeoff, but with different thresholds:
\begin{itemize}[leftmargin=*,nosep]
\item \emph{Arena}: 5\% oracle achieves 94\% ranking accuracy; calibration benefit plateaus early.
\item \emph{MBPP}: 5\% oracle provides 18\% RMSE reduction; 10\% provides 32\%; benefit continues scaling to 84\% at full oracle.
\end{itemize}
The difference likely reflects task structure: Arena evaluates preferences (ordinal, high variance) while MBPP evaluates correctness (binary, low variance). Binary oracles require more labels to stabilize calibration.

\paragraph{Judge--policy confound.}
The \texttt{mini} transport failure highlights a methodological risk: when a model judges its own outputs, self-preference or style-alignment biases can corrupt calibration.
Practitioners should use external judges or validate transport explicitly.

\paragraph{Takeaway.}
CJE successfully transfers to code generation with objective ground truth.
The oracle fraction tradeoff generalizes: more labels enable larger gains.
Transport validation is essential---without it, confounded policies inflate RMSE by 2--3$\times$.

\end{document}